\documentclass[a4paper,authorcolumns,UKenglish,cleveref,autoref,thm-restate]{lipics-v2021}

\hideLIPIcs  

\usepackage{times}
\usepackage{xcolor}
\usepackage{needspace}

\usepackage{algorithm}
\usepackage{algorithmicx}
\usepackage[noend]{algpseudocode}

\usepackage{amssymb,amscd,stmaryrd,textcomp}
\usepackage{mathtools}
\usepackage{mathrsfs}
\usepackage{mathabx}
\usepackage[nounderscore]{syntax}
\usepackage{semantic}
\usepackage{paralist}

\usepackage{array}
\usepackage{booktabs}

\usepackage{tikz}
\usetikzlibrary{decorations.pathreplacing}

\usepackage{enumitem}

\mathchardef\hy="2D

\newcommand{\stkout}[1]{\ifmmode\text{\sout{\ensuremath{#1}}}\else\sout{#1}\fi}

\definecolor{deeplilac}{rgb}{0.6, 0.33, 0.73}

\newcommand{\ar}[1]{\textcolor{blue}{\ifmmode \text{[AR: #1]}\else [AR: #1] \fi}}
\newcommand{\tv}[1]{\textcolor{magenta}{\ifmmode \text{[TV: #1]}\else [TV: #1] \fi}}
\newcommand{\lh}[1]{\textcolor{orange}{\ifmmode \text{[LH: #1]}\else [LH: #1] \fi}}
\newcommand{\fz}[1]{\textcolor{green}{\ifmmode \text{[FZ: #1]}\else [FZ: #1] \fi}}
\newcommand{\vs}[1]{\textcolor{deeplilac}{\ifmmode \text{[VS: #1]}\else [VS: #1] \fi}}

\renewcommand{\phi}{\varphi}
\newcommand{\set}[1]{\left\{#1\right\}}

\newcommand{\N}{\mathbb{N}}

\newcommand{\mw}{{-\!\!\ast}}

\DeclareMathOperator{\dom}{\mathsf{dom}}
\DeclareMathOperator{\assign}{\mathtt{assign}}
\newcommand{\evalop}[1]{\mathtt{assign}_{#1}}
\DeclareMathOperator{\load}{\mathtt{load}}
\DeclareMathOperator{\store}{\mathtt{store}}
\DeclareMathOperator{\const}{\mathtt{const}}

\DeclareMathOperator{\malloc}{\mathtt{malloc}}
\DeclareMathOperator{\realloc}{\mathtt{realloc}}
\DeclareMathOperator{\free}{\mathtt{free}}
\DeclareMathOperator{\assert}{\mathtt{assert}}
\DeclareMathOperator{\assume}{\mathtt{assume}}
\DeclareMathOperator{\memcpy}{\mathtt{memcpy}}
\DeclareMathOperator{\memmove}{\mathtt{memmove}}
\DeclareMathOperator{\memset}{\mathtt{memset}}

\DeclareMathOperator{\ls}{\mathsf{ls}}
\DeclareMathOperator{\dls}{\mathsf{dls}}

\newcommand{\substitution}[3]{#1[#2,#3)}

\mathlig{-->}{\rightharpoonup} 

\newcommand{\pure}{\pi}

\newcommand{\True}{true}

\newcommand{\Loc}{\N}

\newcommand{\EMA}{\triangleright} 

\newcommand{\ABD}[3]{ #1 * \left[ #2 \right] ~\EMA~ #3 }

\newcommand{\size}{\mathsf{size}}
\newcommand{\calloc}{\mathsf{calloc}}

\newcommand{\NULL}{\mathsf{null}}
\newcommand{\emp}{\mathsf{emp}}

\DeclareMathOperator{\ptrplus}{\mathsf{ptrplus}}
\DeclareMathOperator{\ptrminus}{\mathsf{ptrsub}}

\newcommand{\SPEC}[3]{ \{ #1 \}~#2~\{ #3 \}}

\newcommand{\true}{\mathsf{true}}

\newcommand{\biabduct}{\mathit{biabduct}}

\newcommand{\eentry}{\mathit{entry}}
\newcommand{\eexit}{\mathit{exit}}
\newcommand{\close}{\mathit{cut}}
\newcommand{\iinit}{\mathit{init}}
\newcommand{\symb}{\mathit{symb}}
\newcommand{\contracts}{\mathit{Contract}}

\newcommand{\after}{\mathit{after}}

\newcommand{\Byte}{\mathit{Byte}}
\newcommand{\Val}{\mathit{Val}}
\newcommand{\val}{k}
\newcommand{\Var}{\mathit{Var}}
\newcommand{\PVar}{\mathit{PVar}}
\newcommand{\LVar}{\mathit{LVar}}
\renewcommand{\Loc}{\mathit{Loc}}
\newcommand{\Block}{\mathit{Block}}
\newcommand{\Stack}{\mathit{Stack}}
\newcommand{\Interval}{\mathit{Interval}}
\newcommand{\lowerB}{l}
\newcommand{\upperB}{u}
\newcommand{\Blocks}{\mathit{Blocks}}
\newcommand{\Mem}{\mathit{Mem}}
\renewcommand{\stack}{S}
\newcommand{\blocks}{B}
\newcommand{\mem}{M}
\newcommand{\varIn}[1]{\mathit{var}(#1)}
\newcommand{\varX}{x}
\newcommand{\varXX}{X}
\newcommand{\varY}{y}
\newcommand{\varYY}{Y}
\newcommand{\varZ}{z}
\newcommand{\varZZ}{Z}
\newcommand{\varA}{a}

\newcommand{\varB}{b}
\newcommand{\varBB}{B}
\newcommand{\varC}{c}

\newcommand{\varU}{u}

\newcommand{\varV}{v}

\newcommand{\varW}{w}

\newcommand{\pre}{P}
\newcommand{\post}{Q}
\newcommand{\Frame}{F}
\newcommand{\FrameAlt}{G}
\newcommand{\prefix}{U}
\newcommand{\postFormulaSubscript}{\mathit{free}}
\newcommand{\postEQSubscript}{\mathit{eq}}
\newcommand{\woProgVar}{\post_\postFormulaSubscript}
\newcommand{\ProgVarEQ}{\post_\postEQSubscript}
\newcommand{\varUps}{\Upsilon}
\newcommand{\expression}{\varepsilon}
\newcommand{\expressionAlt}{\zeta}
\newcommand{\expressionSize}{\kappa}
\newcommand{\segment}{\Lambda}
\newcommand{\fbody}{\mathcal{B}}
\newcommand{\bBlock}{\mathfrak{b}}
\newcommand{\eBlock}{\mathfrak{e}}
\newcommand{\loc}{\ell}
\newcommand{\cloc}{v}
\newcommand{\Config}{\mathit{Config}}

\newcommand{\denot}[3]{\llbracket #1 \rrbracket_{#2,#3}}
\newcommand{\error}{\mathit{error}}
\newcommand{\fmapsto}{\hookrightarrow}
\newcommand{\ptsto}[2]{#1 \mapsto #2}
\newcommand{\ptstobyterhs}[2]{#1[#2]}
\newcommand{\ptstobyte}[3]{#1 \mapsto \ptstobyterhs{#2}{#3}}

\newcommand{\sll}[1]{\ls_{#1}}
\newcommand{\dll}[1]{\dls_{#1}}
\DeclareMathOperator{\uop}{\mathsf{uop}}
\DeclareMathOperator{\bop}{\mathsf{bop}}
\newcommand{\stmt}{\mathit{stmt}}
\newcommand{\BAP}{\mathtt{BAP}}
\newcommand{\elim}{\mathtt{eliminate}}

\newcommand{\state}[1]{{\footnotesize \begin{enumerate}[label=,topsep=0mm,itemsep=0mm,partopsep=0mm] \item {\color{brown}$#1$}\end{enumerate}}}

\newcommand{\discardstate}[1]{{\footnotesize \begin{enumerate}[label=,topsep=0mm,itemsep=0mm,partopsep=0mm] \item {\color{red}discard state: $#1$}\end{enumerate}}}

\newcommand{\finalstate}[1]{{\footnotesize{\color{blue} summary:} \begin{enumerate}[label=,topsep=0mm,itemsep=0mm,partopsep=0mm] \item {\color{blue}$#1$}\end{enumerate}}}

\newcommand{\candidate}[1]{{\footnotesize{\color{blue} candidate preconditions:} \begin{enumerate}[label=,topsep=0mm,itemsep=0mm,partopsep=0mm] \item {\color{blue}$#1$}\end{enumerate}}}

\newenvironment{cfunction}[2]
  {\vspace{1em}{\fontfamily{lmr}\selectfont #1 \{}\begin{enumerate}[label=#2,topsep=0mm,itemsep=0mm,partopsep=0mm] \item[] }
  {\end{enumerate}{\fontfamily{lmr}\selectfont \}}}

\newcommand{\LINE}[1]{\item {\fontfamily{lmr}\selectfont #1}}

\title{Low-Level Bi-Abduction}

\author{Luk\'{a}\v{s} Hol\'{\i}k}
  {Brno University of Technology, FIT, Czechia} 
  {holik@fit.vut.cz}
  {https://orcid.org/0000-0001-6957-1651}
  {} 

\author{Petr Peringer}
  {Brno University of Technology, FIT, Czechia} 
  {peringer@fit.vut.cz}
  {https://orcid.org/0000-0002-8264-8307}
  {} 

\author{Adam Rogalewicz}
  {Brno University of Technology, FIT, Czechia} 
  {rogalew@fit.vut.cz}
  {https://orcid.org/0000-0002-7911-0549}
  {} 

\author{Veronika \v{S}okov\'{a}}
  {Brno University of Technology, FIT, Czechia} 
  {isokova@fit.vut.cz}
  {https://orcid.org/0000-0003-1980-7245}
  {} 

\author{Tom\'{a}\v{s} Vojnar}
  {Brno University of Technology, FIT, Czechia} 
  {vojnar@fit.vut.cz}
  {https://orcid.org/0000-0002-2746-8792}
  {} 

\author{Florian Zuleger}
  {TU Wien, Informatics,
  Austria} 
  {florian.zuleger@tuwien.ac.at}
  {https://orcid.org/0000-0003-1468-8398}
  {} 

\authorrunning{L. Hol\'{\i}k \and P. Peringer \and A. Rogalewicz \and V.
\v{S}okov\'{a} \and T. Vojnar \and F. Zuleger}

\Copyright{Luk\'{a}\v{s} Hol\'{\i}k and Petr Peringer and Adam Rogalewicz and
Veronika \v{S}okov\'{a} and Tom\'{a}\v{s} Vojnar and Florian Zuleger}



\category{} 

\relatedversion{}
\funding{The Czech authors were supported by the project 20-07487S of the Czech
Science Foundation, the FIT BUT internal project FIT-S-20-6427, and L.
Hol\'{\i}k by the ERC.CZ  project LL1908.}
\nolinenumbers 

\EventEditors{Karim Ali and Jan Vitek}
\EventNoEds{2}
\EventLongTitle{35th European Conference on Object Oriented Programming (ECOOP 2022)}
\EventShortTitle{ECOOP 2022}
\EventAcronym{ECOOP}
\EventYear{2022}
\EventDate{June 6--10, 2022}
\EventLocation{Berlin, Germany}
\EventLogo{}
\SeriesVolume{XXX}
\ArticleNo{XXX}

\begin{document}

\maketitle

\begin{abstract} The paper proposes a new static analysis designed to handle
open programs, i.e., fragments of programs, with dynamic pointer-linked data
structures---in particular, various kinds of lists---that employ advanced
low-level pointer operations.
The goal is to allow such programs be analysed without a need of writing
analysis harnesses that would first initialise the structures being handled.
The approach builds on a special flavour of separation logic and the approach of
bi-abduction.
The code of interest is analyzed along the call tree, starting from its leaves,
with each function analysed just once without any call context, leading to a set
of contracts summarizing the behaviour of the analysed functions.
In order to handle the considered programs, methods of abduction existing in the
literature are significantly modified and extended in the paper.
The proposed approach has been implemented in a tool prototype and successfully
evaluated on not large but complex programs. \end{abstract}

\section{Introduction}

Programs with complex dynamic data structures and pointer operations are
notoriously difficult to write and understand.
This holds twice when a need to achieve the best possible performance drives
programmers, especially those working in the C language on which we concentrate,
to start using advanced low-level pointer operations such as pointer arithmetic,
bit-masking information on pointers, address alignment, block operations with
blocks that are split to differently sized fields (of size not known in
advance), which can then be merged again, and reinterpreted differently, and so
on.
It may then easily happen that the resulting programs contain nasty errors, such
as null-pointer dereferences, out-of-bound references, double free operations,
or memory leaks, which can manifest only under some rare circumstances, may
escape traditional testing, and be difficult to discover once the program is in
production.

To help discover such problems (or show their absence), suitable static analyses
with formal roots may help.
However, the problem of analysing programs with dynamic pointer-linked data
structures, sometimes referred to as \emph{shape analysis}, belongs among the
most difficult analysis problems, which is related to a need of efficiently
encoding and handling potentially infinite sets of graph structures of
in-advance unknown shape and unbounded size, corresponding to the possible
memory configurations.

Moreover, the problem becomes even harder when one needs to analyse not entire
programs, equipped with some analysis harness generating instances of the data
structures to be handled, but just \emph{fragments of code}, which simply start
handling some dynamic data structures through pointers without the structures
being initialised first.
At the same time, in practice, the possibility of analysing code fragments is
highly preferred since programmers do not like writing specialised analysis
harnesses for initialising data structures of the code to be analysed (not
speaking about that writing such harnesses is error-prone too).
Moreover, the possibility of analysing code fragments can also help scalability
of the analysis since it can then be performed in a modular way.

In this paper, we propose a new analysis designed to analyse programs and even
\emph{fragments} of programs with \emph{dynamic pointer-linked data structures}
that can use advanced \emph{low-level pointer-manipulating operations} of the
form mentioned above.
In particular, we concentrate on sequential C programs without recursion and
without function pointers manipulating various forms of \emph{lists}---singly-linked,
doubly-linked, circular, nested, and/or intrusive, which are
perhaps the most common kind of dynamic linked data structures in practice.

Our approach uses a special flavor of \emph{separation logic} (SL)
\cite{Reynolds:SepLogic:02, OHearn:BI:01} with \emph{inductive list
predicates}~\cite{SLNestedLists07} to characterize sets of program
configurations.
To be able to handle code fragments, we adopt the principle of \emph{bi-abductive
analysis} proposed over SL for analysing programs without low-level pointer
operations in \cite{BiAbd09, BiAbd11}.
Our work can thus be viewed as an extension of the approach of \cite{BiAbd09,
BiAbd11} to programs with truly low-level operations (i.e., pointer arithmetic,
bit-masking on pointers, block operations with blocks of variable size, their
splitting to fields of in-advance-not-fixed size, merging such fields back, and
reinterpreting them differently, etc.).
As will become clear, handling such programs requires rather non-trivial changes
to the abduction procedure used in \cite{BiAbd09, BiAbd11}---intuitively, one
needs new analysis rules for block splitting and merging, new support for
operations such as pointer plus, pointer minus, or block operations (like
$\memcpy$), and also modified support for operations like memory allocation or
deallocation (to avoid deallocation of parts of blocks).
Moreover, to support splitting of memory blocks to parts, gradually learning
their bounds and fields, and to allow for embedding data structures into other
data structures not known in advance (as commonly done, e.g., in the so-called
intrusive lists), we even switch from using the traditional \emph{per-object
separating conjunction} in our SL to a \emph{per-field separating conjunction}
(as used, e.g., in \cite{Mihaela:WeakSep:13} in the context of analysing
so-called overlaid data structures), requiring separation not on the level of
allocated memory blocks but their fields.
As an additional benefit, our usage of per-field separating conjunction then
allows us to represent more compactly even some operations on traditional data
structures (without low-level pointer manipulation).

As common in bi-abductive analyses, we analyse programs, or their fragments,
along their \emph{call tree}, starting from the \emph{leaves} of the call tree
(for the time being, we assume working with non-recursive programs only).
Each function is analysed just once, without any knowledge about its possible
call contexts.
For each function, the analysis derives a set of so-called \emph{contracts},
which can then be used when this function is called from some other function
higher up in the call hierarchy.
A contract for a function $f$ is a pair $(\pre,\post)$ where $\pre$ is a
precondition under which $f$ can be safely executed (without a risk of running
into some memory error such as a null-dereference), and $\post$ is a
postcondition that is guaranteed to be satisfied upon exit from $f$ provided it
was called under the given precondition.
Both $P$ and $Q$ are described using our flavor of SL.
In fact, as also done in \cite{BiAbd09, BiAbd11}, our analysis runs in two
phases: the first phase derives the preconditions, while the second phase
computes the postconditions.
Like in \cite{BiAbd09, BiAbd11}, the computed set of contracts may
\emph{under-approximate} the set of all possible safe preconditions of $f$
(e.g., some extreme but still safe preconditions need not be discovered).
However, for each computed contract $(\pre,\post)$, the post-condition $\post$
is guaranteed to \emph{over-approximate} all configurations that result from
calling the function under the pre-condition $\pre$.

\enlargethispage{4mm}

We have implemented our approach in a prototype tool called \emph{Broom}.
We have applied the tool to a selection of code fragments dealing with various
kinds of lists, including very advanced implementations taken from the Linux
kernel as well as the intrusive list library (for a
reference, see our experimental section).
Although the code is not large in the number of lines of code, it contains very
advanced pointer operations, and, to the best of our knowledge, Broom is
currently the only analyser that is capable of analysing many of the involved
functions.

\paragraph*{Related work}

In the past (at least) 25 years there have appeared numerous approaches to
automated \emph{shape analysis} or, more generally, analysis of programs with
unbounded dynamically-linked data structures.
These approaches differ in the formalisms used for encoding sets of
configurations of programs with such data structures, in their level of
automation, classes of supported data structures, and/or properties of programs
that are targeted by the analysis: see, e.g., \cite{pale97, Sagiv02,
SLNestedLists07, WiesVerComplexPropSymbShAnal07, sas07:chang_rival_necula,
KuncakFullFVLinkStr08, InvaderCAV08, juggrnaut10, sleek12, dragoi:atva12,
PredatorSAS13, Forester13, 2LS18}.

Not many of the existing approaches offer a reasonably general support of
\emph{low-level pointer operations} (such as pointer arithmetic, address
alignment, masking information on pointers, block operations, etc.).
Some support of low-level pointer operations appears in multiple of these
approaches, but it is often not much documented.
In fact, such a support often appears in some \emph{ad hoc} extension of the
tool implementing the given approach only, without any description whatsoever.
According to the best of our knowledge, the approach of \cite{PredatorSAS13},
based on so-called \emph{symbolic memory graphs (SMGs)}, currently provides
probably the most systematic and generic solution for the case of programs with
low-level pointer operations and various kinds of linked lists (including
advanced list implementations such as those used in the Linux kernel).
Specialised approaches to certain classes of low-level programs, namely,
\emph{memory allocators}, then appear, e.g., in \cite{Dino:BeyondReachShAbs:06,
Mihaela:FreeListMemAlloc:16}.

In this work, we get inspired by some of the analysis capabilities of
\cite{PredatorSAS13}, but we aim at removing one of its main
limitations---namely, the fact that it cannot be applied to a \emph{fragment of
code}.
Indeed, \cite{PredatorSAS13} expects the analysed program to be \emph{closed},
i.e., the analysed functions must be complemented by a \emph{harness} that
initializes all the involved data structures, which severely limits
applicability of the approach in practice (since programmers are often reluctant
to write specialised analysis harnesses).

Approaches allowing one to analyse \emph{open code}, i.e., \emph{code
fragments}, with dynamic linked data structures are not frequent in the
literature.
Perhaps the best known of these works is the approach of \emph{bi-abduction}
based on \emph{separation logic} with (possibly nested) list predicates proposed
in \cite{BiAbd09, BiAbd11} and currently available in the Infer
analyser \cite{Dino:Infer:16}.\footnote{The approach \cite{BiAbd09, BiAbd11}
mentions a generalisation to other classes of data structures, but---to
the best of our knowledge---this extension has not been implemented and
evaluated, and so it is not clear how well it would work in practice.}
This approach is another of the approaches that inspired our work, and we will
be referring to various technical details of that paper later on.
However, despite Infer contains some support of pointer arithmetic, it
is not very complete (as our experiments will show), and the approach presented
in \cite{BiAbd09, BiAbd11} does not at all study low-level pointer operations of
the form that we aim at in this paper.
Moreover, it turns out that adding a support of such operations (e.g., dealing
with blocks of memory of possibly variable size, splitting them to fields of
variable size, merging such fields back and reinterpreting their contents
differently, having pointers with variable offsets, supporting rich pointer
arithmetic, etc.) requires rather non-trivial changes and extensions to the
bi-abduction mechanisms used in~\cite{BiAbd09, BiAbd11}.

An approach of \emph{second-order bi-abduction} based also on separation logic
was proposed in \cite{SecOrderBiAbd14} and several follow-up papers such as
\cite{BiAbdShOrd19}.
The authors consider recursive programs with pointers and propose a calculus for
automatic derivation of sets of equations describing the behaviour of particular
functions.
A solution of such a set of equations leads to a set of contracts for the
considered functions.
The technique is in some sense quite general---unlike \cite{BiAbd09, BiAbd11}
and unlike our approach, it can even automatically learn \emph{recursive
predicates} describing the involved data structures, including trees, skip
lists, etc.
Moreover, the derivation of the equations is a cheap procedure, and no widening
is needed, again unlike in  \cite{BiAbd09, BiAbd11} and unlike in our approach.
On the other hand, finding a solution of the generated equations is a hard
problem, and the authors provide a simple heuristic designed for a~specific
shape of the equations only, which fails in various other cases.

\enlargethispage{4mm}

Finally, we mention the Gillian project, a \emph{language-independent framework}
based on separation logic for the development of \emph{compositional symbolic
analysis} tools, including tools for whole-program symbolic execution,
verification of annotated code, as well as bi-abduction~\cite{GilPartI20,
GilSymExecForAll20, GilMultiLangPlatf21, GilPartII21}.
The works on Gillian concentrate on the generic framework it develops, and the
published description of the supported bi-abductive analysis, perhaps most
discussed in \cite{GilSymExecForAll20}, is unfortunately not very detailed.
In particular, it is not clear whether and how much the approach supports the
low-level features of pointer manipulation that we are aiming at here (e.g.,
pointer arithmetic, bit-masking on addresses, etc.).
According to the source code that we were able to find in the Gillian
repository, the examples mentioned in the part of \cite{GilSymExecForAll20}
devoted to bi-abduction do not use low-level pointer manipulation features such
as pointer arithmetic.
It is also mentioned in~\cite{GilSymExecForAll20} that Gillian supports
bi-abduction up to a predefined bound only, whereas we do not require such
a~bound.
Further, in contrast to the present work, \cite{GilSymExecForAll20} assumes that
the size of memory chunks being dynamically allocated is known, and the complex
reasoning needed to resolve this issue is left for the future.

We also note that there is a vast body of work on \emph{automated decision
procedures} for various fragments of separation logic and problems such as
satisfiability and entailment---see, e.g., \cite{EneaLSV14, IRV14,
JensFlorian:BeyondSymHeap:20, LocLe:VMCAI21, Radu:CADE21}.
However, it is not immediate how to apply these logics inside a program analysis
tool.
%
%
For example, the results on the cited separation logics cannot be directly
applied to the \mbox{(bi-)abduction} problem, which is the central operation
needed for a compositional program analysis.
This is because the best (i.e., logically weakest) solution to the abduction
problem $\phi * [?] \models \psi$, which is a central problem for compositional
program analyses, with $*$ being the separating conjunction, is given by the
formula $\phi \mw \psi$, which makes use of the magic wand operator $\mw$, and
the cited logics do not provide support for the magic wand.
This is for principle reasons: it has been observed in the literature that magic
wand operators are ``difficult to eliminate''~\cite{books/daglib/0034962};
further, it has been shown that adding only the singly-linked list-segment
predicate to a propositional separation logic that includes the magic wand
already leads to undecidability of the satisfiability
problem~\cite{conf/fossacs/DemriLM18}.
A notable exception is the recent work~\cite{conf/esop/PagelZ21} on a new
semantics for separation logic, which enables decidability of a propositional
separation logic that includes the magic wand and the singly-linked list-segment
predicate (and also discusses applications to the abduction problem); however,
the fragment considered in~\cite{conf/esop/PagelZ21} is not expressive enough to
cover the low-level features considered in this work such as, pointer
arithmetic, memory blocks, etc., and, at present, it is unclear whether the
decidability result can be extended to a richer logic.
For the above reasons, we will in this paper  not target a complete procedure
for the (bi-)abduction problem, but rather, following~\cite{BiAbd09, BiAbd11},
develop approximate procedures and evaluate their usefulness in our case
studies.

\paragraph*{Main contributions of the paper}

The paper proposes a new approach for automated bi-abductive analysis of
programs and fragments of programs with pointers, different kinds of linked
lists, and low-level memory operations.  The approach is formalised, implemented
in a prototype tool, and experimentally evaluated. In summary, we make the
following contributions:\begin{itemize}

  \item A specialised dialect of separation logic suitable for automated
  abductive analysis of programs with lists and low-level memory operations (we
  use a separating conjunction between single fields and not whole memory blocks
  as in related approaches, and support fields of unknown and even variable size
  as well as unknown block boundaries).

  \item Contracts for basic programming statements that reflect our low-level
  memory model (see, e.g., the contracts of the \texttt{malloc} and
  \texttt{free} statements), and support for specific statements that permit
  low-level pointer manipulation (e.g., pointer addition).

  \item A set of rules for automated abductive analysis, which not only includes
  variants of rules from related approaches, but also new kinds of rules
  required for handling low-level memory operations (e.g., block splitting).

  \item  A prototype implementation that supports bit-precise reasoning based on
  a reduction of (un\nobreakdash-)sat\-is\-fi\-a\-bil\-i\-ty  of separation
  logic to (un-)satisfiability of SMT over the bit-vectors.

  \item An experimental evaluation of the approach on a number of challenging
  programs.

\end{itemize}

\section{An Illustration of the Approach on an Example}\label{sec:illustration}

Before we start with a systematic description of our approach, we present its
core ideas on an example.
We will consider the code manipulating cyclic doubly-linked lists shown in
Fig.~\ref{fig:example1}.\footnote{The code is written in C. Our later presented
low-level programming language for which we will formalise our approach is not C
but rather close to some of the intermediate languages used when compiling C.
We, however, feel that describing the example in such a language would not be
very understandable. Moreover, all constructions used in our example can be
translated to the later considered language.}
The example is inspired by the principle of \emph{intrusive lists} (as used,
e.g., in Linux kernel lists) where all list operations are defined on some
simple list-linking structure that is then nested into user-defined structures.
It is these user-defined structures that carry the data actually stored in the
lists.
The list manipulating functions, however, know nothing about these larger
structures.
However, the fact that contracts (summaries) derived for functions dealing with
the small linking structures are later to be applied on the larger, user-defined
structures is already problematic for some existing analyses.
When providing an introduction to the approach, we will try to informally
explain the involved notions, yet, due to the complexity of the issues, some
prior knowledge of separation logic  with
\emph{inductive list predicates}, e.g., \cite{SLNestedLists07}, and ideally also
bi-abduction analysis \cite{BiAbd09, BiAbd11} is helpful.

\begin{figure}
\begin{minipage}{\textwidth-1em}
\small
{\fontfamily{lmr}\selectfont
struct dll \{ struct dll *next, *prev; \};\\
struct emb\_dll \{int value; struct dll link; \}; }

\begin{cfunction}{void \textbf{init\_dll}(struct dll *x)}{}
    \state{\pre\equiv x=X, \quad\post\equiv x=X}
  \LINE{x-->next = x;}
    \state{\pre\equiv \ptsto{X}{L_1} * x=X, \quad\post\equiv \ptsto{X}{X} * x=X}
  \LINE{x-->prev = x;}
    \state{\pre\equiv \ptsto{X}{L_1} * \ptsto{X+8}{L_2} * \bBlock(X)=\bBlock(X+8) * x=X, \\
      \post\equiv \ptsto{X}{X} * \ptsto{X+8}{X} *
      \bBlock(X)=\bBlock(X+8) * x=X}
\end{cfunction}
\finalstate{\pre\equiv \ptsto{X}{L_1} * \ptsto{X+8}{L_2} * \bBlock(X)=\bBlock(X+8) *
    x=X, \\
    \post\equiv \ptsto{X}{X}
	* \ptsto{X+8}{X} * \bBlock(X)=\bBlock(X+8) * x=X}
\begin{cfunction}{void \textbf{insert\_after}(struct dll *l, *j)}{(\roman*)}
    \state{\pre\equiv l=L*j=J, \quad\post\equiv l=L * j=J}
  \LINE{struct dll *n = l-->next;}
    \state{\pre\equiv \ptsto{L}{N} * l=L * j=J, \quad\post\equiv \ptsto{L}{N} * l=L * j=J * n=N}
  \LINE{j-->next = n;}
    \state{\pre\equiv \ptsto{L}{N} * \ptsto{J}{B_1} * l=L*j=J, \quad\post\equiv  \ptsto{L}{N} * \ptsto{J}{N} * l=L * j=J * n=N}

  \LINE{j-->prev = l;}
    \state{\pre\equiv \ptsto{L}{N} * \ptsto{J}{B_1} * \ptsto{J+8}{B_2} * \bBlock(J)=\bBlock(J+8) *  l=L*j=J, \\
      \post\equiv \ptsto{L}{N} * \ptsto{J}{N} * \ptsto{J+8}{L} * \bBlock(J)=\bBlock(J+8) * l=L * j=J * n=N}

  \LINE{l-->next = j;}
    \state{\pre\equiv \ptsto{L}{N} * \ptsto{J}{B_1} * \ptsto{J+8}{B_2} * \bBlock(J)=\bBlock(J+8) *  l=L*j=J, \\
       \post\equiv \ptsto{L}{J} * \ptsto{J}{N} * \ptsto{J+8}{L} * \bBlock(J)=\bBlock(J+8) *l=L*j=J*n=N}

  \LINE{n-->prev = j;}
    \state{\pre\equiv \ptsto{L}{N} * \ptsto{J}{B_1} * \ptsto{J+8}{B_2} * \ptsto{N+8}{B_3} * \bBlock(J)=\bBlock(J+8) *
       \bBlock(N)=\bBlock(N+8)
       *  l=L*j=J, \\
       \post\equiv \ptsto{L}{J} * \ptsto{J}{N} * \ptsto{J+8}{L} * \ptsto{N+8}{J} * \bBlock(J)=\bBlock(J+8) *
	\bBlock(N)=\bBlock(N+8)
     * l=L *\\ j=J * n=N}

\end{cfunction}
\finalstate{\pre\equiv \ptsto{L}{N} * \ptsto{J}{B_1} * \ptsto{J+8}{B_2} * \ptsto{N+8}{B_3} *
	\bBlock(J)=\bBlock(J+8) * \bBlock(N)=\bBlock(N+8)
      * l=L * j=J, \\
       \post\equiv \ptsto{L}{J} * \ptsto{J}{N} * \ptsto{J+8}{L} * \ptsto{N+8}{J} *  \bBlock(J)=\bBlock(J+8) *
	\bBlock(N)=\bBlock(N+8)
    * l=L * j=J}

\begin{cfunction}{int \textbf{main}()}{(\alph*)}
    \state{\pre\equiv \emp, \quad\post\equiv \emp}
  \LINE{struct emb\_dll *x = malloc(sizeof(struct emb\_dll));}
    \state{\pre\equiv \emp, \quad\post\equiv \exists X. \  \ptstobyte{X}{\top}{24} * X=\bBlock(X)* x=X}

  \LINE{init\_dll(\&(x-->link));}
    \state{\pre\equiv \emp, \quad\post\equiv \exists X,L_1. \ \ptstobyte{X}{\top}{8} * \ptsto{L_1}{L_1} * \ptsto{L_1+8}{L_1} * X=\bBlock(X)=\bBlock(L_1)=\bBlock(L_1+8)
       * \\ L_1=X+8 * x=X}

  \LINE{struct emb\_dll *i = malloc(sizeof(struct emb\_dll));}
    \state{\pre\equiv \emp, \quad\post\equiv \exists I,X,L_1. \ \ptstobyte{I}{\top}{24} * \ptstobyte{X}{\top}{8} * \ptsto{L_1}{L_1} * \ptsto{L_1+8}{L_1}
       * L_1=X+8 * X=\bBlock(X)=\bBlock(L_1)=\bBlock(L_1+8)  * I=\bBlock(I)* x=X * i=I}

  \LINE{init\_dll(\&(i-->link));}
    \state{\pre\equiv \emp, \quad\post\equiv \exists I,X,L_1,L_2. \ \ptstobyte{i}{\top}{8} * \ptsto{L_2}{L_2} * \ptsto{L_2+8}{L_2} * X\rightarrow \top[8] * \ptsto{L_1}{L_1} * \ptsto{L_1+8}{L_1} *
       L_2=I+8 * L_1=X+8 * X=\bBlock(X)=\bBlock(L_1)=\bBlock(L_1+8)  *
	I=b(I)=b(L_2)=b(L_2+8)
      * x=X * i=I}

  \LINE{insert\_after(\&(x-->link), \&(i-->link));}
    \state{\pre\equiv \emp, \quad\post\equiv \exists I,X,L_1,L_2. \ \ptstobyte{I}{\top}{8} * \ptsto{L_2}{L_1} * \ptsto{L_2+8}{L_1} * \ptstobyte{X}{\top}{8} * \ptsto{L_1}{L_2} * \ptsto{L_1+8}{L_2} *
     L_2=I+8 * L_1=X+8 * X=\bBlock(X)=b(L_1)=\bBlock(L_1+8)  *
	I=\bBlock(I)=\bBlock(L_2)=\bBlock(L_2+8) * \\
        x=X * i=I}
  \item[] \dots
\end{cfunction}

\end{minipage}

  \caption{An illustrative example of a code working with cyclic doubly-linked
  lists and its analysis.
  The C expressions like \texttt{ptr->field} can be seen as syntactic sugar for
  expressions using pointer arithmetic of the form \texttt{*(ptr +
  offsetof(field))}.
  The $\eBlock(X)$ predicates representing \emph{the end of the block pointed by
  X} are dropped from the $(\pre,\post)$ pairs for simplicity.}

  \label{fig:example1}
\end{figure}

In the code of our illustrative example, the function \texttt{init\_dll} creates
an initial cyclic doubly-linked list consisting of a single node.
The function \texttt{insert\_after} can then insert a new element into the list
after its item pointed by $l$.

Let us note that while the code of the example in Fig.~\ref{fig:example1} may
seem to not use pointer arithmetic,
the code in fact uses pointer arithmetic
on the level of the intermediate code we analyse.
Indeed, each expression \texttt{x–>field} is translated to
\texttt{*(x+offsetof(field))}.
It is of course true that once all the types and fields are known and fixed, one
can avoid dealing with pointer arithmetic in this case. On the other hand, the
fact that we systematically handle it through pointer arithmetic allows us to
smoothly handle even the cases when the types and offsets stop being known
and/or constant (upon which approaches based on dealing with field names fail).

As indicated already in the introduction, we analyse the given code fragment
according to its \emph{call tree}, starting from the leaves (assuming there is
no recursion).
Each function is analysed just once, without any call context.
If successful, the analysis derives a set of contracts for the given function
where each contract is a pair $(\pre,\post)$ consisting of a (conjunctive)
pre-condition and (a possibly disjunctive) post-condition.
In our introductory example, we will restrict ourselves to the simplest case,
namely, having a single, purely conjunctive contract.
In the contracts, both the pre- and post-condition are expressed as SL formulae.
The analysis is \emph{compositional} in that contracts derived for some
functions are then used when analysing functions higher up in the call hierarchy
(moreover, we will view even particular pointer manipulating statements as
special atomic functions and describe them by pre-defined contracts).

As indicated already in the introduction, we analyse the given code fragment
according to its \emph{call tree}, starting from the leaves (assuming there is
no recursion).
Each function is analysed just once, without any call context.
If successful, the analysis derives a set of contracts for the given function
where each contract is a pair $(\pre,\post)$ consisting of a (conjunctive)
pre-condition and (a possibly disjunctive) post-condition.
In our introductory example, we will restrict ourselves to the simplest case,
namely, having a single, purely conjunctive contract.
In the contracts, both the pre- and post-condition are expressed as SL formulae.
The analysis is \emph{compositional} in that contracts derived for some
functions are then used when analysing functions higher up in the call hierarchy
(moreover, we will view even particular pointer manipulating statements as
special atomic functions and describe them by pre-defined contracts).

We begin the illustration of our analysis by analysing the \texttt{init\_dll}
function.
We start the analysis by annotating the first line by the pair $(x=X,x=X)$.
In this pair, the first component is the so-far derived pre-condition of the
function, and the second component is the current symbolic state of the function
under analysis.
Here, the variable $X$ records the value of the program variable $x$ at the
beginning of the function.
While $x$ will be changing in the function, $X$ will never change, and we will
be able to gradually generate constraints on its value to express what must hold
for $x$ at the entry of the function.

After symbolically executing the statement \texttt{x->next = x}, we derive that
the address $X$ must correspond to some allocated memory, containing some
unknown value $L_1$.
This gives us the pre-condition $\ptsto{X}{L_1}$ that is an SL formula stating
exactly the fact that $X$ is allocated and stores the value $L_1$.
The symbolic state is then advanced to say that $X$ is allocated and stores the
value $X$, i.e., it points to itself, which is encoded as $\ptsto{X}{X}$ in SL.

After the subsequent statement \texttt{x->prev = x}, assuming that we work with
64 bit (i.e., 8 bytes) wide addresses, we add to the precondition the fact that
the memory address $X+8$ is allocated as well.
Moreover, the formula $\bBlock(X)=\bBlock(X+8)$ says that $X$ and $X+8$ belong
to the same memory block, i.e., they were, e.g., allocated using one
\texttt{malloc} statement (in fact, we use $\bBlock(X)$ to denote the---so-far
unknown---base address of the block).
The symbolic state is updated by the fact that the value at the address $X+8$ is
also equal to $X$, i.e., $\ptsto{X+8}{X}$.

\enlargethispage{4mm}

Since there are no further statements in the function, there is no branching, no
loops, and all the statements are deterministic, the final \emph{contract} for
the function is unique and consists of the final pre-condition $\pre\equiv
\ptsto{X}{L_1} * \ptsto{X+8}{L_2} * \bBlock(X)=\bBlock(X+8) * x=X$ and the
post-condition $\post\equiv \ptsto{X}{X} * \ptsto{X+8}{X} *
\bBlock(X)=\bBlock(X+8) * x=X$ obtained from the final symbolic state.
Here, we use ``$*$'' to denote a \emph{per-field separating
conjunction}, which, intuitively, means that while the addresses $X$ and $X+8$,
which are allocated by the formulae $\ptsto{X}{L_1}$ and $\ptsto{X+8}{L_2}$,
may---though need not---belong to a single memory block, the values stored at
these addresses within the block do not overlap.\footnote{In a formula
$\ptsto{a}{b} * \ptsto{c}{d}$ with a per-object separating conjunction, $a$ and
$c$ are two distinct objects allocated in memory (while $b$ and $d$ need not be
allocated and may coincide). With a per-field separating conjunction, $a$ and $c$ are allowed to
be non-overlapping \emph{fields} of the same allocated object.}

The same principles are then used for the computation of the contracts for the
\texttt{insert\_after} and \texttt{main} functions.
Here, let us just highlight a situation that happens, e.g., upon the
\texttt{j->next = n} statement of \texttt{insert\_after}.
Notice that, in its case, the so-far computed precondition $P$ must be extended
by the new requirement $\ptsto{J}{B_1}$, stating that $J$ must be allocated, and
$Q$ is then extended by the fact $\ptsto{J}{N}$, which is the effect of
executing the given statement.
At the same time, however, the rest of the previously computed symbolic state of
the program $Q$ stays untouched (in general, only some part may be preserved).
Given the current symbolic state $Q$ and a statement, the problem of deriving
which precondition is missing and which part of the state will remain untouched
is denoted as the \emph{bi-abduction problem}, and a procedure looking for its
solution is a \emph{bi-abduction procedure}.
The computed missing part of the pre-condition is called the \emph{anti-frame},
and the computed part of the current symbolic state not modified by the
statement being executed is called the \emph{frame}.

When analysing the \texttt{main} function, one does already need not re-analyse
the \texttt{init\_dll} and \texttt{insert\_after} functions---instead, one
simply uses their contracts.
For simplicity, we assume here that \texttt{malloc} always succeeds, and hence
even \texttt{main} is deterministic.
After the execution of \texttt{malloc}, we use the special predicate
$\ptstobyte{x}{\top}{24}$ to express that a sequence of 24 bytes of undefined
contents was allocated.
We allow such blocks (as well as all other kinds of blocks that arise during the
analysis) be \emph{split} to smaller parts whenever this is needed for applying
a contract of some function (or statement).
That happens, e.g., on lines $b$ and $d$ of the \texttt{main} function where the
block $\ptstobyte{X}{\top}{24}$ created by \texttt{malloc} is split to 3 fields
as described by $\ptstobyte{X}{\top}{8} * \ptstobyte{X+8}{\top}{8} *
\ptstobyte{X+16}{\top}{8}$.
The last two of the fields then match the precondition of \texttt{init\_dll},
and the first one becomes a frame (untouched by the function).

Without now going into further details, we note that analysing more complex
functions requires one to solve multiple more problems.
For example, if there appears some non-determinism, one needs to start working
with contracts with disjunctive post-conditions and even with sets of such
contracts.
If the code contains loops, one needs to prevent the analysis from diverging
while generating more and more points-to predicates.
For that, one can use widening in the form of a list abstraction.
The resulting over-approximation may then, however, render some generated
pre-/post-condition pairs unsound, leading to a need to run another phase of the
analysis that will start from the computed pre-conditions and check, without
using abduction any more, what post-condition the code can really guarantee.
We will discuss all these issues in detail in the following.

However, before proceeding, let us stress how significantly the above-mentioned
use of the per-field separation distinguishes our approach from its predecessor
bi-abduction analysis \cite{BiAbd09, BiAbd11}.
That analysis would use \emph{whole-block} predicates of the form
$\ptsto{X}{dll(next:A,prev:B)}$ to describe instances of \texttt{struct dll},
while we use the formula $\ptsto{X}{A} * \ptsto{X+8}{B} *
\bBlock(X)=\bBlock(X+8)$.
The per-field separating conjunction allows us to (1)~express partial information
about a block and (2)~infer a precondition where two (or more) fields can be in
the same block as well as in different blocks.
Point 1 helps us to generate contracts of functions where we do not know the
exact sizes of the allocated block---e.g., \texttt{init\_dll} does not require
the pointer $x$ to point to an instance of \texttt{struct dll}, it can be, e.g.,
used on larger structures, such as, e.g., \texttt{struct emb\_dll}, that
\emph{embed} the original structure.
Point 2 is used in the contract of \texttt{insert\_after} where the formula
$\ptsto{L}{N} * \ptsto{N+8}{B_3}$ describes a memory where it may be that $L=N$
as well as $L\neq N$.
The contract for \texttt{insert\_after} can then be applied on a circular
doubly-linked list consisting of a single item ($L=N$) as well on lists
consisting of more items (\mbox{$L\neq N$})---see Fig. \ref{fig:insert-after} for an
illustration.

\begin{figure}[h]
  \centering
    \includegraphics[width=.4\textwidth]{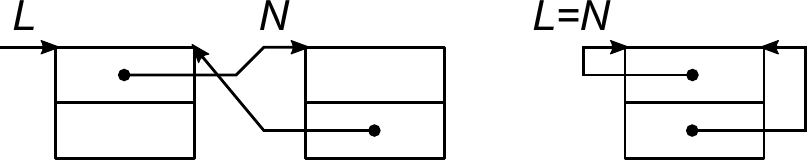}
  \caption{An illustration of the two possible input configurations of
  the \texttt{insert\_after} procedure}
  \label{fig:insert-after}
\end{figure}

Note that when one uses the whole-block predicate, the precondition of
\texttt{insert\_after} in the form $\ptsto{L}{dll(next:N,prev:\_)} *
\ptsto{N}{dll(next:\_,prev:B_3)} * \ptsto{J}{dll(next:B_1,prev:B_2)}$ requires
$L\neq N$, and hence it is not covering the two above mentioned cases.
One can of course sacrifice performance of the analysis and generate multiple
contracts by modifying the abduction rules---e.g., one can non-deterministically
introduce an alias $L=N$ before inferring the anti-frame on line $v$ of
\texttt{main} to get the pre-condition $\ptsto{L}{dll(next:L,prev:\_)} * L=N$.
Introducing such non-determinism is, however, costly.
That is why, as we will see in our experiments, it is not done in tools such as
Infer, which can then cause that such tools will miss some function contracts
(or generate incomplete contracts that will not be applicable in some common
cases: such as insertion into a list of length 1).

An additional Example is provided in Appendix \ref{app:illustr-bitmasking}, where pointer
arithmetic and bit-masking are directly visible in the C-code.

\section{Memory Model}

In the following, we introduce the memory model that we use in this paper.
\emph{Values} are sequences of bytes, i.e., $\Val = \Byte^{+}$, where bytes are
8-bit words.
Sequences of bytes can be interpreted as numbers---either signed or unsigned,
which we leave as a part of the operations to be applied on the sequences
(including conversion operations).
We designate a subset of the values $\Loc = \Byte^N \subseteq \Val$ as
\emph{locations} where $N \geq 1$ is the byte-width of words of a given
architecture and where byte sequences to be interpreted as locations are always
understood as unsigned.
The $\NULL$ pointer is represented by $0 \in \Loc$ in our memory model.

We will use so-called \emph{stack-block-memory triplets} (SBM triplets for
short)  as \emph{configurations} of our memory model in order to define the
operational semantics of programs (and also to define the semantics of our
separation logic later on):

\textbf{Stack.}
We assume some  set of \emph{variables} $\Var$ where each variable $\varX \in
\Var$ has some fixed positive size, denoted as $\size(\varX)$.
Then, $\Stack$ is the set of total functions $\Var \rightarrow \Val$ such that
each variable is mapped to a byte sequence whose length is according to the size
of the variable, i.e., for each stack $\stack \in \Stack$ and variable $\varX
\in \Var$, we have $\stack(\varX) \in \Byte^{\size(\varX)}$.

\textbf{Memory.}
$\Mem$ is the set of partial functions $\Loc \rightharpoonup \Byte$ that define
the contents of allocated memory locations.

\textbf{Blocks.}
We use $\Interval = \{ ~[\lowerB,\upperB) ~\mid~ \lowerB<\upperB ~\text{where}~
\lowerB,\upperB \in \Loc \}$ to denote intervals of subsequent memory locations
where we include the lower bound and exclude the upper bound.
Intuitively, an interval $[\lowerB,\upperB) \in \Interval$ will denote which
locations were allocated at the same time (and must thus also be deallocated
together, can be subtracted using pointer  subtraction, etc.).
$\Block = \{~[\lowerB,\upperB) \in \Interval \mid \lowerB \neq 0 \}$ are
intervals whose lower bound is not 0 (recall that $\NULL$ is represented by $0
\in \Loc$ in our memory model).
$\Blocks \subseteq (\mathbf{2}_\mathit{fin})^\Block$ is the set of all finite
sets of non-overlapping blocks, i.e., for all $\blocks \in \Blocks$ and for all
$[\lowerB_1,\upperB_1), [\lowerB_2,\upperB_2) \in \blocks$ such that either
$\lowerB_1 \neq \lowerB_2$ or $\upperB_1 \neq \upperB_2$, we have that either
$\upperB_1 \le \lowerB_2$ or $\upperB_2 \le \lowerB_1$.

\textbf{Configurations.}
$\Config$ consists of all triplets $(\stack,\blocks,\mem) \in \Stack \times
\Blocks \times \Mem$ such that the set of allocated blocks and the locations
whose contents is defined are linked as follows:\begin{itemize}

  \item For every $\loc \in \Loc$ s.t. $\mem(\loc)$ is defined, there is a block
  $[\lowerB,\upperB) \in \blocks$ s.t. $\loc \in
  [\lowerB,\upperB)$.\footnote{Note that we do not require the reverse, i.e.,
  that all locations of a block are allocated. This is because our separation
  logic is set up to work with partially allocated blocks. In particular, the
  separating conjunction needs to break up blocks into partial blocks. We note,
  however, that the semantics of our programming language maintains the
  invariant that each block is always fully allocated.}

\end{itemize}

We introduce functions $\bBlock_\blocks, \eBlock_\blocks: \Loc \rightarrow
\Loc$, parameterized by some set of blocks $\blocks \in \Blocks$, which return
the base or end address, respectively, of the block to which a given location
belongs, i.e., given some $\loc \in \Loc$, we set
$\bBlock_\blocks(\loc)=\lowerB$ in case there is some $[\lowerB,\upperB) \in
\blocks$ with $\loc \in [\lowerB,\upperB)$, and $\bBlock_\blocks(\loc) = 0$,
otherwise.
Likewise for $\eBlock_\blocks(\loc)$.

%

%

%

\textbf{Axioms.} For later use, we note that, building on the above notation, we
can express the requirements for locations to be within their associated block
and for blocks to be non-overlapping in the form of the following two
axioms:
$$\forall \loc. \ \bBlock_\blocks(\loc) = 0 \vee \bBlock_\blocks(\loc) \le \loc < \eBlock_\blocks(\loc)$$
\begin{multline*}
  \forall \loc, \loc'. \ (0 < \bBlock_\blocks(\loc) < \eBlock_\blocks(\loc')\le \eBlock_\blocks(\loc)
  \vee 0 < \bBlock_\blocks(\loc') < \eBlock_\blocks(\loc)\le \eBlock_\blocks(\loc')) \rightarrow  \\
  \bBlock_\blocks(\loc) = \bBlock_\blocks(\loc') \wedge \eBlock_\blocks(\loc) = \eBlock_\blocks(\loc')
\end{multline*}


\textbf{Notation.}
Given a (partial) function $f$, $f[a \fmapsto b]$ denotes the (partial) function
identical to $f$ up to $f[a \fmapsto b](a)=b$.
Moreover, $f[a \fmapsto \bot]$ denotes the (partial) function identical to $f$
up to being undefined for $a$.

\section{A Low-level Language and Its Operational Semantics}
\label{sec:language-plus-semantics}

We now state a simple low-level language together with its operational
semantics.
The language is close to common intermediate languages into which programs in C
are compiled by compilers such as \texttt{gcc} or \texttt{clang}.
We assume that a type checker ensures that variables of the right sizes are
used, guaranteeing, in particular, that the left-hand side (LHS) and right-hand
side (RHS) of an assignment are of the same size or that the dereference
operator is only applied to locations.
%
%
We do not include the operators of \emph{item access} (\texttt{.} and
\texttt{->}) nor \emph{indexing} (\texttt{[]}) into our language as their usage
can be compiled to using \emph{pointers}, \emph{pointer arithmetic}, and the
\emph{dereference operator} (\texttt{*}) as indeed commonly done by compilers.
Likewise, we do not include the \emph{address-of operator} (\texttt{\&}) whose
usage can be replaced by storing all objects whose address should be derived via
\texttt{\&} into dynamically allocated memory, followed by using pointers to
such memory, as also done automatically by some compilers.
Further, we assume the \texttt{sizeof} and \texttt{offsetof} operators be
resolved and transformed to constants.

We now present the statements of our low-level language together with their
operational semantics.
The semantics is defined over configurations, which we introduced in the
previous section.
The semantics maintains the following invariant:\begin{itemize}

  \item For every $[\lowerB,\upperB) \in \blocks$ and every $\loc \in
  [\lowerB,\upperB)$, $\mem(\loc)$ is defined.

\end{itemize}

We start with rules describing various \emph{assignment statements} possibly
combined with pointer dereferences either on the LHS or RHS.
In the rules (and further on), we use $\mem[\loc,\loc')$ to denote the byte
sequence $\mem(\loc)\mem(\loc + 1)\cdots \mem(\loc'-1)$:

\vspace*{-4mm} 

\begin{displaymath}
  (\stack,\blocks,\mem) \xrightarrow{\varX := \val} (\stack[\varX \fmapsto
  \val],\blocks,\mem) \text{ for some value } \val \in \Val
\end{displaymath}

\vspace*{-8mm} 

\begin{displaymath}
  (\stack,\blocks,\mem) \xrightarrow{\varX := \varY} (\stack[\varX \fmapsto
  \stack(\varY)],\blocks,\mem)
\end{displaymath}

\vspace*{-8mm} 

\begin{multline*}
  (\stack,\blocks,\mem) \xrightarrow{\varX := *\varY}
  \text{ if } \bBlock_\blocks(\stack(\varY))= 0 \text{ or }
  \stack(\varY) + \size(\varX) > \eBlock_\blocks(\stack(\varY)),\\
  \text{ then } \error \text{ else }
  (\stack[\varX \fmapsto
    \mem[\stack(\varY),\stack(\varY)+\size(\varX))],\blocks,\mem)
\end{multline*}
Note that, in the case of $\varX := *\varY$, one needs to read
$\size(\varX)$ bytes from the adress $\stack(\varY)$. This is impossible if
the condition $\stack(\varY) + \size(\varX) >
\eBlock_\blocks(\stack(\varY))$ holds.

\vspace*{-5mm} 

\begin{multline*}
  (\stack,\blocks,\mem) \xrightarrow{*\varX := \varY}
  \text{ if } \bBlock_\blocks(\stack(\varX))= 0
  \text{ or } \stack(\varX) + \size(\varY) > \eBlock_\blocks(\stack(\varX)),\\
  \text{ then } \error \text{ else }
  (\stack,\blocks,\mem[[\stack(\varX),\stack(\varX)+\size(\varY)) \fmapsto
    \stack(\varY)])
\end{multline*}

We continue by \emph{memory allocation}.
We treat 0-sized allocations as an error.\footnote{The C standard says that the
behaviour in this case is user-defined, the allocation can return $\NULL$ or a
non-$\NULL$ value, which, however, cannot be dereferenced.
However, since such an allocation is usually suspicious, many analysers flag it
as an error/warning.
We adopt the same approach, but if need be, the rules could be changed to handle
such allocations according to the standard.}
For non-zero-sized allocations, the allocation can always fail and return
$\NULL$, otherwise the successfully allocated memory block is initialized with
some arbitrary value\footnote{Notice that $2^{8N}$ gives the largest address
that can be expressed using words with the byte-width $N$.}:

\vspace*{-4mm} 

\begin{multline*}
    (\stack,\blocks,\mem) \xrightarrow{\varX = malloc(\varZ)}
    \text{ if } \stack(\varZ) = 0 \text{ then } \error \text{ else either } (\stack[\varX \fmapsto \NULL],\blocks,\mem) \text{ or }\\
    (\stack[\varX \fmapsto \loc],\blocks \cup
    \{ [\loc,\loc+\stack(\varZ)) \},\mem[[\loc,\loc+\stack(\varZ)) \fmapsto \val])
    \text{ for some } \val \in \Byte^{\stack(\varZ)} \text{ and } \loc > 0\\
    \text{ such that } \loc+\stack(\varZ) \le 2^{8N} \text{ and }
    [\loc,\loc+\stack(\varZ)) \text{ does not overlap with any }
    [\lowerB,\upperB) \in \blocks
\end{multline*}

%


The $\calloc$ function, which nullifies the allocated block, can be defined
analogically to $\malloc$, by just changing $\mem[[\loc,\loc+\stack(\varZ))
\fmapsto \val]$ to $\mem[[\loc,\loc+\stack(\varZ)) \fmapsto 0^{\stack(\varZ)}]$.
The $\realloc$ function, which shrinks or enlarges a block, possibly moving it
to a different memory location, can be reduced to a sequence of other
statements, and so we do not introduce it explicitly for brevity.


The \emph{deallocation} of memory is modelled by the following
rule:\footnote{Notice that we do not need a rule for deallocating zero-sized
blocks since we do not allow such blocks to be created.}

\vspace*{-4mm} 

\begin{multline*}
  (\stack,\blocks,\mem) \xrightarrow{free(\varX)}
  \text{ if }  \stack(\varX) \neq \bBlock_\blocks(\stack(\varX))
  \text{ then } \error\\
  \text{ else } (\stack,\blocks\setminus\{[\stack(\varX),
    \eBlock_\blocks(\stack(\varX))) \},\mem[[\stack(\varX),
    \eBlock_\blocks(\stack(\varX))) \fmapsto \bot])
\end{multline*}

We now state rules for \emph{(unary and binary) operations} excluding the
operation of adding an offset to a pointer and pointer subtraction, which will
be handled separately.
We assume that typing will ensure the right sizes of the operands and the
variables to which the results are assigned.
For instance, for arithmetic operations, we expect all operands and result
variables to have the same type.
%
Unary operators can be used to model type-casting that lifts variables to the
required size.
Finally, note also that we do not exclude division by zero in order not to
complicate the rules in aspects that are not specific for our work.

\vspace*{-4mm} 

\begin{displaymath}
    (\stack,\blocks,\mem) \xrightarrow{\varX = \varY \bop \varZ}
    (\stack[\varX \fmapsto \stack(\varY) \bop \stack(\varZ)],\blocks,\mem)
\end{displaymath}

\vspace*{-8mm} 

\begin{displaymath}
    (\stack,\blocks,\mem) \xrightarrow{\varX = \uop \varY}
    (\stack[\varX \fmapsto \uop \stack(\varY)],\blocks,\mem)
\end{displaymath}


The rule for \emph{adding a (possibly negative) offset to a pointer} requires
its pointer argument to be defined and, in accordance with the C standard, the
result be within the appropriate memory block plus one byte (i.e., it may point
just behind the end of the block).\footnote{We are aware that this requirement
is not respected in some real-life pograms, such as, e.g., the implementation of
lists in Linux. We will later mention that our approach can be relaxed to handle
such cases too.}

\vspace*{-4mm} 

\begin{multline*}
  (\stack,\blocks,\mem) \xrightarrow{\varX = \varY \ptrplus \varZ}
  \text{ if } \stack(\varY) = 0 \text{ or } \stack(\varY) + \stack(\varZ) <
  \bBlock_\blocks(\stack(\varY)) \text{ or } \stack(\varY) + \stack(\varZ) >
  \eBlock_\blocks(\stack(\varY)),\\ \text{ then } \error
  \text{ else } (\stack[\varX \fmapsto \stack(\varY)+\stack(\varZ)],\blocks,\mem)
\end{multline*}

The rule for the \emph{pointer subtraction} is special in that it requires its
pointer operands to be defined, to have the same base, and to point inside an
allocated block or just behind its end (note the usage of the interval closed at
both ends in the rule).

\vspace*{-4mm} 

\begin{multline*}
  (\stack,\blocks,\mem) \xrightarrow{\varX = \varY \ptrminus \varZ}
  \text{ if }  \stack(\varY) = 0 \text{ or } \stack(\varZ) \not\in [\bBlock_\blocks(\stack(\varY)),
  \eBlock_\blocks(\stack(\varY))],\\
  \text{ then } \error
  \text{ else } (\stack[\varX \fmapsto \stack(\varY)-\stack(\varZ)],\blocks,\mem)
\end{multline*}

Further, we state the semantics of the $\memcpy$ statement that copies $\varZ$
bytes from an address $\varY$ to an address $\varX$ and that assumes dealing
with non-overlapping source/target blocks.
The effect of the $\memmove$ statement, which can handle even overlapping
blocks, can be simulated by copying the source block to a fresh block and then
back, and so we skip it for brevity.
For brevity, we also skip the (easy to add) $\memset$ statement, which sets each
byte of some block to a given value.

\vspace*{-4mm} 

\begin{multline*}
    (\stack,\blocks,\mem) \xrightarrow{\memcpy(\varX,\varY,\varZ)}
    \text{ if } \bBlock_\blocks(\stack(\varX))= 0
    \text{ or } \bBlock_\blocks(\stack(\varY))= 0 \\
    \text{ or } \stack(\varX) + \stack(\varZ) > \eBlock_\blocks(\stack(\varX))
    \text{ or } \stack(\varY) + \stack(\varZ) > \eBlock_\blocks(\stack(\varY)),\\
    \text{ or } \stack(\varX) + \stack(\varZ) > \stack(\varY)
    \text{ or } \stack(\varY) + \stack(\varZ) > \stack(\varX),\\
    \text{ then } \error
    \text{ else } (\stack,\blocks,\mem[[\stack(\varX),\stack(\varX)+\stack(\varZ))
    \fmapsto \mem[\stack(\varY),\stack(\varY)+\stack(\varZ))])
\end{multline*}

To encode \emph{conditional branching} arising from conditional statements or
loops, we introduce the \emph{assume} statement that models conditions $\varX
\bowtie \varY$ for $\bowtie~ \in \{ =,\neq, \leq,<,\geq,> \}$.
The assume statement has a side condition which ensures that this statement can
only be taken for configurations satisfying the condition:

\vspace*{-4mm} 

\begin{displaymath}
    (\stack,\blocks,\mem) \xrightarrow{\assume(\varX \bowtie \varY)}
    (\stack,\blocks,\mem) \text{ under the condition that }
    \stack(\varX) \bowtie \stack(\varY)
\end{displaymath}

We also introduce the \emph{assert} statement that is similar to the assume
statement, but it checks at runtime whether the specified condition holds, and
it fails if this is not the case:

\vspace*{-4mm} 

\begin{displaymath}
    (\stack,\blocks,\mem) \xrightarrow{\assert(\varX \bowtie \varY)} \text{ if }
    \stack(\varX) \bowtie \stack(\varY) \text{ then }
    (\stack,\blocks,\mem)  \text{ else } \error
\end{displaymath}

We structure our programs into \emph{functions}.
We assume a control flow graph $(V,E,\eentry,\eexit)$ for each function
$f(\varX_1,\dots,\varX_n)$, where $V$ is a~set of \emph{nodes}, there are
dedicated nodes $\eentry,\eexit \in V$, and $E \subseteq V \times V$ is a set of
edges where each edge is labelled by a basic statement or a function call.
We model branching and looping as usual, i.e., by having multiple outgoing control flow edges labelled by an $\assume$ statement
with appropriate conditions.
For simplicity, we consider functions without a return value, not referring to
global variables, having parameters passed by reference only, with the names of
the parameters unique to each function, and not having local
variables.\footnote{This is all w.l.o.g. as values can be returned through
parameters passed by reference, locals can be replaced by parameters, globals
can be passed as parameters, and passing by value can be simulated by copying a
variable into a~fresh one before passing it. Our implementation does, however,
have a special support for all these features.}
We consider non-recursive functions only.\footnote{A support for recursion could
be added as, e.g., in \cite{BiAbd11}.}
Finally, when a function is called, no variable can be used as an argument
twice.

In particular, let $f(x_1,\ldots,x_n)$, $n > 0$, be a function whose body
$\fbody_f$ uses the variables $x_1,\ldots,x_n$ only.
Let $f(a_1,\ldots,a_n)$ be a call of $f$ where the arguments $a_i$ are pairwise
different variables.
We will now define the semantics of this function call with regard to some
configuration $(\stack,\blocks,\mem)$.
The definition will, however, be only partial due to the possible
non-termination of $f(a_1,\ldots,a_n)$.
Let $\stack_f$ be the stack defined by $\stack_f(x_i) = \stack(a_i)$ for all $1
\le i\le n$.
%
%
Then, we consider the execution of $\fbody_f$ wrt the initial configuration
determined by $\stack_f$ as well as the blocks $\blocks$ and the memory $\mem$
at the call of~$f$:
Let us assume that $(\stack_f,\blocks,\mem) \xrightarrow{\fbody_f}
(\stack_f',\blocks',\mem')$ where $\xrightarrow{\fbody_f}$ denotes an execution from the $\eentry$ to the $\eexit$ of $f$ (again this definition is only partial as the execution of $\fbody_f$ does not need to terminate).
If $\fbody_f$ terminates wrt configuration $(\stack_f,\blocks,\mem)$, we set
\begin{displaymath} (\stack,\blocks,\mem) \xrightarrow{f(a_1,\ldots,a_n)}
(\stack',\blocks',\mem') \end{displaymath}
where $\stack'$ is defined by $\stack'(a_i) = \stack_f'(x_i)$ for all $1 \le i
\le n$, and $\stack'(y) = \stack(y)$ for all $y \in \Var \setminus
\{a_1,\ldots,a_n\}$.

\section{Separation Logic} \label{sec:SL}

We now introduce a separation logic that supports reasoning about low-level
memory models as introduced earlier.
Our separation logic (SL) has the following syntax:

\begin{align*}
  \phi ~::=~ & \ptsto{\expression_1}{\expression_2} \mid
  \ptstobyte{\expression_1}{\val}{\expression_2} \mid
  \ptstobyte{\expression_1}{\top}{\expression_2} \mid
  \phi_1 * \phi_2 \mid \phi_1 \vee \phi_2 \mid \sll{\segment(\varX,\varY)}(\expression_1,\expression_2) \mid\\
  & \dll{\segment(\varX,\varY,\varZ)}(\expression_1,\expression_2,\expression_1',
    \expression_2') \mid
  \emp \mid \True \mid \expression_1 \bowtie \expression_2 \mid \exists \varX. \phi\\
  \bowtie ~::=~ & = \mid \neq \mid \leq \mid < \mid \geq \mid >
  \hspace*{5mm}
  \expression ~::=~ \val \mid \varX \mid \bBlock(\expression) \mid
  \eBlock(\expression) \mid \uop \expression \mid \expression_1 \bop \expression_2
\end{align*}

\textbf{Variables and Values.}
Our SL formulae are stated over the same set of variables $\Var$ and values $\Val$ that we introduced in the definition of our memory model.
In particular, the variables  $\varX,\varY,\varZ$ and the values $k$ of our SL formulae are drawn from $\Var$ and $\Val$, respectively.

\textbf{Size.} Variables, values, operators, and expressions in our logic are typed by their \emph{size}.
We will only work with formulae where the variables and values respect the sizes
expected by the involved operations and predicates.
For every expression $\expression$, we denote by $\size(\expression)$ the size
of the value to which this expression may evaluate.
We remark on the choice of working with fixed sizes:
%
%
We intentionally do not permit variables of variable size because (1)~such variables are typically not supported by low-level languages and (2)~variables of variable size allow one to model strings, which would make our language vastly more powerful (allowing one to model all kinds of string operations)\footnote{We believe that extending our later
presented analysis to such variables is possible (by recording the length of the
target object as another parameter of the points-to predicate), but we leave it
for future work in order not to complicate the basic approach we propose.}.

\textbf{Points-To Predicates.} The points-to predicate
$\ptsto{\expression_1}{\expression_2}$ denotes that the byte sequence
$\expression_2$ is stored at the memory location $\expression_1$.
Due to we are working with expressions of fixed size, every model of
$\ptsto{\expression_1}{\expression_2}$ must allocate exactly
$\size(\expression_2)$ bytes.
In addition, we introduce two restricted cases of points-to predicates where the
RHS is of parametric size: namely,
$\ptstobyte{\expression_1}{\val}{\expression_2}$ and
$\ptstobyte{\expression_1}{\top}{\expression_2}$ that allow us to say that
$\expression_1$ points to an array of $\expression_2$ bytes that either all have
the same constant value $\val$ or have any value, respectively.
These predicates allow us to, e.g., express that some block of memory is
nullified, which is often crucial to know when analysing advanced
implementations of dynamic data structures \cite{PredatorSAS13}.
We lift the notion of size to the RHS of these points-to predicates as follows:
$\size(\ptstobyterhs{\val}{\varY}) = \size(\ptstobyterhs{\top}{\varY}) = \varY$.
In $\ptstobyte{\expression_1}{\val}{\expression_2}$, we require $\val$
to be a single byte, i.e., $\size(\val)=1$.

\textbf{Notation.}
Given a formula $\phi$, we write $\varIn{\phi}$ to denote the \emph{free} variables of $\phi$ (as usual a variable is \emph{free} if it does appear within an existential quantification).
Further, given an expression $\expression$, we write $\varIn{\expression}$ for all variables appearing in $\expression$.

\textbf{Terminology.}
We call formulae that do not contain the disjunction operator ($\vee$)
\emph{symbolic heaps}.
We will mostly work with symbolic heaps in this paper.
Disjunctions of symbolic heaps will be only used on the RHS of (some) contracts.
We call formulae that do not contain existential quantification ($\exists$) \emph{quantifier-free}.
Our SL contains the relational predicates $\expression_1 \bowtie \expression_2$,
which include equality and disequality;
these predicates are traditionally called \emph{pure} in the separation logic literature.
We follow this terminology and call any separating conjunction of such predicates a
\emph{pure formula}.

\textbf{List-Segment Predicates.} List segments in our SL are parameterized by a
\emph{segment} predicate $\segment(\varX,\varY)$ or
$\segment(\varX,\varY,\varZ)$ for singly-linked or doubly-linked lists,
respectively; see Fig.~\ref{fig:ls-dls} for an illustration of the semantics of
$\sll{\segment(\varX,\varY)}(\expression_1,\expression_2)$ and
$\dll{\segment(\varX,\varY,\varZ)}(\expression_1,\expression_2,\expression_1',\expression_2')$
for $\segment(\varX,\varY) \equiv \ptsto{\varX}{\varY}$ and
$\dll{\segment(\varX,\varY,\varZ)} \equiv \ptsto{\varX}{\varZ}*
\ptsto{\varX+8}{\varY}$.
We note that our list-segment predicates only have two or three free variables,
respectively, which prevents the logic from, e.g., describing non-global heap
objects shared by list elements.
However, more parameters could be introduced in a similar fashion to other works
\cite{SLNestedLists07}.
We have not done so here since it would complicate the notation, and we take
this issue as orthogonal to the techniques we propose.

\textbf{Binary and Unary Operators.} $\uop$ and $\bop$ denote some arbitrary set
of binary and unary operators, respectively.
We assume this set to include at least the usual operators ($+$, $-$, $*$, $\&$,
$|$, \ldots) available in low-level languages as well as a special substring
operator $\substitution{\cdot}{\cdot}{\cdot}$ on byte sequences where
$\substitution{\val}{i}{j} $ for some $\val = b_0\cdots b_{l-1} \in \Byte^l$ and
$0 \le i \le j \le l$ denotes the byte sequence $b_i\cdots b_{j-1}$.
%
%
Since we work with variables of fixed size, we basically assume a version of
each $\uop$ and $\bop$ for every possible operand size.
We further remark that unary operators $\uop$ can be used for modelling the
casting to different sizes.

\begin{figure}[t]
  \centering
\tikzset{
  listitem/.style={circle,draw,thick,minimum size=18pt}
}
\pgfdeclarelayer{background}
\pgfsetlayers{background,main}
\begin{tikzpicture}
  \pgfmathsetmacro{\xr}{-0.5}; 
  \pgfmathsetmacro{\yr}{0.5};  

  \node[listitem,red,fill=white] (x) at (0,0) {$x$};
  \node[listitem]       (y) at (1.5,0) {\color{red}$y$};
  \node                 (a) at (2.5,0) {$...$};
  \node[listitem]       (b) at (3.5,0) {};
  \node[listitem]       (c) at (5,0) {};

  \draw (x) node[below,yshift=-15pt] {$\varepsilon_1$};
  \draw (c) node[below,yshift=-15pt] {$\varepsilon_2$};

  \draw (\xr,-\yr) rectangle(4,\yr) ;
  \draw (\xr,\yr) node[above right]{\color{red}$\Lambda(x,y)$} ;

  \draw[->,thick,red] (x) to (y);
  \draw[->,thick] (y) to (a);
  \draw[->,thick] (a) to (b);
  \draw[->,thick] (b) to (c);
  \begin{pgfonlayer}{background}
    \node (x) at (0,0) [draw,circle,,line width=1.6em,red!30] {};
  \end{pgfonlayer}

    \coordinate (shift) at (8,0);
    \begin{scope}[shift=(shift)]
      \node[listitem,densely dotted] (y) at (-1.5,0) {\color{red}$y$};
      \node[listitem,red,fill=white] (x) at (0,0) {$x$};
      \node[listitem]       (z) at (1.5,0) {\color{red}$z$};
      \node                 (a) at (2.5,0) {$...$};
      \node[listitem]       (b) at (3.5,0) {};
      \node[listitem]       (c) at (5,0) {};

      \draw (y) node[below,yshift=-15pt] {$\varepsilon_2$};
      \draw (x) node[below,yshift=-15pt] {$\varepsilon_1$};
      \draw (b) node[below,yshift=-15pt] {$\varepsilon'_1$};
      \draw (c) node[below,yshift=-15pt] {$\varepsilon'_2$};

      \draw (\xr,-\yr) rectangle(4,\yr) ;
      \draw (\xr,\yr) node[above right]{\color{red}$\Lambda(x,y,z)$} ;

      \draw[->,thick] (y) edge[bend left] (x);
      \draw[->,red,thick] (x) edge[bend left] (z);
      \draw[->,thick] (z) edge[bend left] (a);
      \draw[->,thick] (a) edge[bend left] (b);
      \draw[->,thick] (b) edge[bend left] (c);
      \path[->,red,thick] (x) edge[bend left] (y);
      \path[->,thick] (z) edge[bend left] (x);
      \path[->,thick] (a) edge[bend left] (z);
      \path[->,thick] (b) edge[bend left] (a);
      \path[->,thick] (c) edge[bend left] (b);
      \begin{pgfonlayer}{background}
        \node (x) at (0,0) [draw,circle,,line width=1.6em,red!30] {};
      \end{pgfonlayer}
    \end{scope}

\end{tikzpicture}

  \caption{An illustration of the meaning of the
  $\sll{\segment(\varX,\varY)}(\expression_1,\expression_2)$ and
  $\dll{\segment(\varX,\varY,\varZ)}(\expression_1,\expression_2,\expression_1',
  \expression_2')$ formulae.}

  \label{fig:ls-dls}
\end{figure}

\textbf{Semantics.}
We now define the semantics of our SL over SBM triplets $(\stack,\blocks,\mem)
\in \Config$:

\vspace*{-4mm} 

\begin{multline*}
  (\stack,\blocks,\mem) \models \ptsto{\expression_1}{\expression_2} \text{ iff }\\
  \dom(\mem) = [\denot{\expression_1}{\stack}{\blocks},
    \denot{\expression_1}{\stack}{\blocks}+\size(\expression_2)) \text{ and }
  \mem[\denot{\expression_1}{\stack}{\blocks},
    \denot{\expression_1}{\stack}{\blocks}+\size(\expression_2)) = \denot{\expression_2}{\stack}{\blocks}
\end{multline*}

\vspace*{-3mm}

where\vspace*{-4mm}
\begin{multline*}
  \denot{\val}{\stack}{\blocks} ~=~ \val,
  \denot{\varX}{\stack}{\blocks} ~=~ \stack(\varX),
  \denot{\bBlock(\expression)}{\stack}{\blocks} ~=~
    \bBlock_{\blocks}(\denot{\expression}{\stack}{\blocks}),
  \denot{\eBlock(\expression)}{\stack}{\blocks} ~=~
    \eBlock_{\blocks}(\denot{\expression}{\stack}{\blocks}), \\
  \denot{\uop \expression}{\stack}{\blocks} ~=~
    \uop(\denot{\expression}{\stack}{\blocks}), \text{ and }
  \denot{\expression_1 \bop \expression_2}{\stack}{\blocks}~ =~
    \denot{\expression_1}{\stack}{\blocks} ~\bop~
    \denot{\expression_2}{\stack}{\blocks}
\end{multline*}
\vspace*{-8mm} 

\begin{multline*}
  (\stack,\blocks,\mem) \models \ptstobyte{\expression_1}{\val}{\expression_2}
  \text{ iff }\\
  \dom(\mem) = [\denot{\expression_1}{\stack}{\blocks},
    \denot{\expression_1}{\stack}{\blocks}+\denot{\expression_2}{\stack}{\blocks}) \text{ and }
  \mem[\denot{\expression_1}{\stack}{\blocks}+i] = \val
    \text{ for all } 0 \le i < \denot{\expression_2}{\stack}{\blocks}
\end{multline*}

\vspace*{-8mm} 

\begin{displaymath}
  (\stack,\blocks,\mem) \models \ptstobyte{\expression_1}{\top}{\expression_2} \text{ iff }
  \dom(\mem) = [\denot{\expression_1}{\stack}{\blocks},\denot{\expression_1}{\stack}{\blocks}+\denot{\expression_2}{\stack}{\blocks})
\end{displaymath}

We remark on the difference between the three points-to predicates:
the predicate $\ptsto{\expression_1}{\expression_2}$ fixes the exact sequence of bytes $\expression_2$ that is stored from location $\expression_1$ onwards, and the number of bytes is known (the size of $\expression_2$); the predicate $\ptstobyte{\expression_1}{\val}{\expression_2}$ states that
there are $\expression_2$ number of bytes stored from location $\expression_1$ onwards (note that the number of bytes $\expression_2$ is symbolic), and each of these bytes equals $\val$; and the predicate $\ptstobyte{\expression_1}{\top}{\expression_2}$ works in the same way except that the bytes stored are not fixed.

\vspace*{-4mm} 

\begin{displaymath}
  (\stack,\blocks,\mem) \models \phi_1 * \phi_2 \text{ iff there are some }
  \mem_1, \mem_2 \text{ with } \mem = \mem_1 \uplus \mem_2,
  (\stack,\blocks,\mem_i) \models \phi_i
\end{displaymath}

\vspace*{-8mm} 

\begin{displaymath}
  (\stack,\blocks,\mem) \models \phi_1 \vee \phi_2 \text{ iff }
  (\stack,\blocks,\mem) \models \phi_1 \text{ or } (\stack,\blocks,\mem) \models \phi_2
\end{displaymath}

\vspace*{-8mm} 

\begin{displaymath}
  (\stack,\blocks,\mem) \models \emp \text{ iff } \dom(\mem) = \emptyset
  \hspace{10mm} 
  (\stack,\blocks,\mem) \models \True \text{ always holds }
\end{displaymath}


\vspace*{-8mm} 

\begin{displaymath}
  (\stack,\blocks,\mem) \models \expression_1 \bowtie \expression_2 \text{ iff }
  \dom(\mem) = \emptyset \text{ and }
  \denot{\expression_1}{\stack}{\blocks} \bowtie \denot{\expression_2}{\stack}{\blocks}
\end{displaymath}

We point out that pure formulae constrain the heap to be empty.
This is typically not required by separation logics that support classical (non-separating) conjunction at least on pure sub-formulae.
However, we exclude the classical conjunction in order to simplify the presentation and hence need to constrain the heap of pure formulae to be empty.


\begin{multline*}
  (\stack,\blocks,\mem) \models \exists \varX. \phi(\varX) \text{ iff there is some }
  v \in \Val\\
  \text{ and a fresh variable } \varU \in \Var \text{ s.t. }
  (\stack[\varU \fmapsto v],\blocks,\mem) \models \phi(\varU)
\end{multline*}


\vspace*{-8mm} 

\begin{multline*}
  (\stack,\blocks,\mem) \models \sll{\segment(\varX,\varY)}(\expression_1,\expression_2)
  \text{ iff } (\stack,\blocks,\mem) \models \expression_1 = \expression_2
  \text{ or }\\
  \quad \quad (\stack,\blocks,\mem) \models \expression_1 \neq \expression_2 * \true
  \text{ and there is some } \loc \in \Loc\\
  \quad \quad \text{ and a fresh variable } u \in \Var \text{ s.t. }
  (\stack[u \fmapsto \loc],\blocks,\mem) \models \segment(\expression_1,u) *
  \sll{\segment(\varX,\varY)}(u,\expression_2)
\end{multline*}

%
%
%

\vspace*{-8mm} 

\begin{multline*}
  (\stack,\blocks,\mem) \models
  \dll{\segment(\varX,\varY,\varZ)}(\expression_1,\expression_2,\expression_1',
    \expression_2')
  \text{ iff } (\stack,\blocks,\mem) \models \expression_1 = \expression_2' *
    \expression_2 = \expression_1' \text{ or }\\
    \quad \quad (\stack,\blocks,\mem) \models \expression_1 \neq \expression_2' *
    \expression_2 \neq \expression_1' * \true \text{ and there is some }
  \loc\in \Loc \text{ and a fresh variable} \\
  \quad \quad  u \in \Var \text{ such that }
  (\stack[u \fmapsto \loc],\blocks,\mem) \models \segment(\expression_1,u,\expression_2) *
  \dll{\segment(\varX,\varY,\varZ)}(u,\expression_1,\expression_1',\expression_2')
\end{multline*}

\textbf{Satisfiability and Entailment.} We say that an SL formula $\phi$ is
\emph{satisfiable} iff there is a model $(\stack,\blocks,\mem)$ such that
$(\stack,\blocks,\mem) \models \phi$.
We say that an SL formula $\phi_1$ \emph{entails} an SL formula $\phi_2$,
denoted $\phi_1 \models \phi_2$, iff we have that $(\stack,\blocks,\mem) \models
\phi_2$ for every model $(\stack,\blocks,\mem)$ such that $(\stack,\blocks,\mem)
\models \phi_1$.

\textbf{Restrictions on the Segment Predicates.} From now on, we put further
restrictions on the segment predicates $\segment(\varX,\varY)$ and
$\segment(\varX,\varY,\varZ)$:
(1) $\segment$ needs to be of the shape $\exists \varX_1,\ldots,\varX_k. \phi$
for some quantifier-free symbolic heap $\phi$.
%
%
Intuitively, this condition is required since quantifier-free symbolic heaps are
the formulae on which the symbolic execution described in
Section~\ref{sec:contract-generation} is based on and the existential
quantification allows to hide some nested data.
(2) $\segment$ needs to be \emph{block-closed} in the sense defined below.
%

\textbf{Block-closedness.} A formula $\phi$ is \emph{block-closed} iff, for all
$(\stack,\blocks,\mem) \models \phi$ and $\loc \in \dom(\mem)$, we have that
$[\bBlock(\loc),\eBlock(\loc)) \subseteq \dom(\mem)$.
Intuitively, block-closedness ensures that all points-to assertions in a formula
add up to whole blocks.
%
%
We require block-closedness in order to ensure that list-segments correspond to
our intuition and connect different memory blocks (i.e., we exclude models where
multiple or all nodes of list-segments belong to the same block).
Technically, the requirement of block-closedness makes it easier to formulate
rules for materialisation of list-segment nodes in the abduction procedure and
for entailment checking.
We leave lifting the restriction of block-closedness for future work.
A sufficient condition for block-closedness, which is easy to check, is that all
points-to assertions in $\phi$ can be organized in groups
$\ptsto{\expression_i}{\varUps_i}$, for $1 \le i \le n$, where $\varUps$
represents either $\varY$, $\ptstobyterhs{\val}{\varY}$, or
$\ptstobyterhs{\top}{\varY}$, such that $\expression_i = \expression_{i-1} +
\size(\varUps_i)$ for all $1 < i \le n$, and $\phi$ implies that
$\eBlock(\expression_1)-\bBlock(\expression_1) = \sum_{i=1..n}
\size(\varUps_i)$.


\section{Contracts of Functions and Their Generation}
\label{sec:contract-generation}

Our analysis is based on generating \emph{contracts of functions} along the call
tree, starting from its leaves.
The contracts summarize the semantics of the functions under analysis.
(We may also compute multiple contracts for the same function where each
contract provides a valid summary of the function; the contracts might, however,
differ in the preconditions under which they apply.)
We introduce our notion of contracts in Subsection~\ref{subsec:contracts}.
We note that basic statements of our core language can be viewed as special
built-in functions.
Hence, we then introduce contracts for the basic statements in
Subsection~\ref{subsec:contracts-basic-statements}.
These serve as the starting point of our analysis:
Using the contracts of the built-in functions, we derive contracts for any
function built of them as described in Subsection~\ref{sec:contract-gen}.

\subsection{Contracts of Functions} \label{subsec:contracts}

We assume a set of \emph{variables} $\Var = \PVar \uplus \LVar$ that is
partitioned into two disjoint infinite set of \emph{program variables} $\PVar$
and \emph{logical variables} $\LVar$ (also called \emph{ghost} variables).
For functions $f(\varX_1,\ldots,\varX_n)$ with parameters $\varX_i$, we always
require $\varX_1,\ldots,\varX_n \in \PVar$ (also recall that we assume that
$\varX_1,\ldots,\varX_n$ are the only variables occurring in the body of $f$).
To summarize the semantics of a function $f(\varX_1,\ldots,\varX_n)$, we use
(sets of) \emph{contracts} of the form $\{\pre\} f(\varX_1,\ldots,\varX_n)
\{\post\}$ where \begin{itemize}

  \item the \emph{pre-condition} $\pre$ is a quantifier-free symbolic heap, and

  \item the \emph{post-condition} $\post$ is a disjunction of formulas of the
  form $\exists \prefix_\post. (\woProgVar * \ProgVarEQ)$ such that $\woProgVar$
  is a quantifier-free symbolic heap with $\varIn{\woProgVar} \subseteq \LVar$,
  $\ProgVarEQ$ is the formula $\varX_1 = \expression_1 * \cdots * \varX_n =
  \expression_n$ for some expressions $\expression_i$ with
  $\varIn{\expression_i} \subseteq \LVar$, and $\prefix_\post =
  (\varIn{\woProgVar * \ProgVarEQ} \cap \LVar) \setminus \varIn{\pre}$.
  Note that every disjunct of the post-condition $\post$ describes the heap by a
  formula over the logical variables (the formula $\woProgVar$) and fixes the
  values of the program variables in terms of expressions over the logical
  variables (the formula $\ProgVarEQ$) where all logical variables that do not
  appear in the pre-condition $\pre$ are existentially quantified (on the other
  hand, those logical variables that appear in the pre-condition may be
  implicitly considered as universally quantified).

  \item We call a contract \emph{conjunctive} if the post-condition $\post
  \equiv \post_1 \vee \cdots \vee \post_l$ consists of a single disjunct (i.e.,
  $l=1$), and \emph{disjunctive} otherwise.

\end{itemize}

\paragraph*{Soundness of contracts}

We will now state what it means for a contract to be sound.
As usual we stipulate that configurations satisfying the pre-condition lead to
configurations satisfying the post-condition.
In addition, we also require that we can always add a \emph{frame} to the
pre-/post-condition, i.e., a formula describing a part of the heap untouched by
the function\footnote{That is, we directly incorporate the well-known
\emph{frame rule} from the separation-logic literature into our notion of
soundness.
We choose to do so for economy of exposition and for making the paper self-contained.
As an alternative one could derive the validity of the frame rule from the fact that all contracts of the basic statements, as stated in Section~\ref{subsec:contracts-basic-statements}, are \emph{local actions} in the sense of~\cite{conf/lics/CalcagnoOY07} (which is equivalent to Lemma~\ref{lem:soundness-built-in} stated in this paper).}.
Here, a frame $\Frame$ is any symbolic heap with $\varIn{\Frame} \subseteq
\LVar$.
A contract $\{\pre\} f(\varX_1,\ldots,\varX_n) \{\post\}$ is called \emph{sound}
iff, for all frames $\Frame$, all triples $(\stack,\blocks,\mem)$ such that
$(\stack,\blocks,\mem) \models \Frame * \pre$, and all executions of
$f(\varX_1,\ldots,\varX_n)$ that start from $(\stack,\blocks,\mem)$ and end in
some configuration $(\stack',\blocks',\mem')$\footnote{Note that $\dom(\stack')
= \dom(\stack)$ and that we have $\stack'(x) = \stack(x)$ for all $x \in
\LVar$ because logical variables do not occur in the program and hence are never updated.}, it holds that $(\stack',\blocks',\mem') \models
\Frame*\post$.
%

\subsection{Contracts for Basic Statements} \label{subsec:contracts-basic-statements}

We give below contracts for the basic statements of our programming language
stated as functions (basic statements may be viewed as special
built-in functions).
For simplicity (and w.l.o.g.), we assume that it never happens that the same variable appears both at the LHS and RHS of an assignment\footnote{We may assume this because assignments such as $\varX := *\varX$ can always be rewritten to the sequence $\varY := *\varX; \varX := \varY$ (at the cost of introducing a fresh variable $\varY$).}.
%
%
Recall that $\emp$ is implicit in all otherwise pure constraints (and so we do
not need to repeat it):\begin{itemize}
  \abovedisplayskip=2mm
  \belowdisplayskip=2mm
  \item Function $\assign(\varX,\varY)$ with the body $\varX := \varY$:
  \begin{displaymath}
     \SPEC{\varY = \varYY}
     {\assign(\varX,\varY)}
     {\varX = \varYY ~*~ \varY = \varYY}.
  \end{displaymath}

  \item Function $\const_\val(\varX)$ with the body $\varX := \val$:
  \begin{displaymath}
     \SPEC{\emp}
     {\const_\val(\varX)}
     {\varX = \val}.
  \end{displaymath}

  \item Function $\load(\varX,\varY)$ with the body $\varX := *\varY$:
  \begin{displaymath}
    \SPEC{\varY = \varYY ~*~ \ptsto{\varYY}{\varZ}}
    {\load(\varX,\varY)}
    {\varX = \varZ ~*~ \varY = \varYY ~*~ \ptsto{\varYY}{\varZ}}
  \end{displaymath} with $\woProgVar \equiv \ptsto{\varYY}{\varZ}$ and
  $\ProgVarEQ \equiv \varX = \varZ ~*~ \varY = \varYY$.

  \item Function $\store(\varX,\varY)$ with the body $*\varX := \varY$:
  \begin{displaymath}
  \SPEC{\varX = \varXX ~*~ \varY = \varYY ~*~ \ptsto{\varXX}{\varZ}}
    {\store(\varX,\varY)}
    {\varX = \varXX ~*~ \varY = \varYY ~*~ \ptsto{\varXX}{\varYY}}
  \end{displaymath} with $\woProgVar \equiv \ptsto{\varXX}{\varYY}$ and
  $\ProgVarEQ \equiv \varX = \varXX ~*~ \varY = \varYY$.

  \item Function $\malloc(\varX,\varY)$ that either succeeds or fails to
  allocate memory through $\varX := \malloc(\varY)$:
  \begin{displaymath}
    \SPEC{\varY = \varYY}
    {\malloc(\varX,\varY)}
    {x = \NULL ~*~ \varY = \varYY \vee \exists \varU. \ \varX = \varU ~*~
    \nu(\varU,\varYY) ~*~ \varY = \varYY}
  \end{displaymath} where $\nu(\varU,\varYY) = \ptstobyte{\varU}{\top}{\varYY} *
  \bBlock(\varU) = \varU ~*~ \eBlock(\varU) = \varU+\varYY$.
  Note that either $\woProgVar \equiv \nu(\varU,\varYY)$ and $\ProgVarEQ \equiv
  \varX = \varU ~*~ \varY = \varYY$ or $\woProgVar \equiv \emp$ and $\ProgVarEQ
  \equiv \varX =   \NULL ~*~ \varY = \varYY$.
  A very similar contract can be used for $\calloc$, just with
  $\ptstobyte{\varU}{\top}{\varYY}$ changed to $\ptstobyte{\varU}{0}{\varYY}$.
  %
  %
  We remark that the contracts for $\malloc$ and $\calloc$ are the only
  disjunctive contracts among the contracts for the basic statements of our
  programming language.

%

  \item Function $\free(\varX)$ called with the $\NULL$ argument:
  \begin{displaymath}
    \SPEC{\varX = \varXX ~*~ \varXX = \NULL}
    {\free(\varX)}
    {\varX = \varXX ~*~ \varXX = \NULL}
  \end{displaymath}

  \item Function $\free(\varX)$ called over a non-$\NULL$ argument:
  \begin{displaymath}
    \SPEC{\varX = \varXX ~*~ \ptstobyte{\varXX}{\top}{\varY} * \bBlock(\varXX) =
    \varXX ~*~ \eBlock(\varXX) = \varXX+\varY}
    {\free(\varX)}
    {\varX = \varXX}
  \end{displaymath}
  Note that a block to be freed may be split into multiple fields at the time of
  freeing.
  We, however, do not need to deal with this issue here since the later
  presented bi-abduction rules will split the LHS of the contract of $\free$
  such that it can match the fragmented block.
  %


  \item Functions $\evalop{\bop}(\varX,\varY,\varZ)$ with the body $\varX :=
  \varY \bop \varZ$ for binary operators $\bop$ (and likewise for unary
  operators $\uop$):
  \begin{displaymath}
     \SPEC{\varY = \varYY ~*~ \varZ = \varZZ}
     {\varX := \varY \bop \varZ}
     {\varX = \varYY \bop \varZZ ~*~ \varY = \varYY ~*~ \varZ = \varZZ}
  \end{displaymath}

  \item Function $\ptrplus(\varX,\varY,\varZ)$ with the body $\varX := \varY
  \ptrplus \varZ$ for the case when the result is within the block of the
  pointer to which an offset is added:
  \begin{displaymath}
     \SPEC{\varY = \varYY ~*~ \varZ = \varZZ ~*~ \varphi_{Y,Z}}
     {\varX := \varY \ptrplus \varZ}
     {\varX = \varYY + \varZZ ~*~ \varY = \varYY ~*~ \varZ = \varZZ ~*~ \varphi_{Y,Z}}
  \end{displaymath} for $\varphi_{Y,Z} \equiv \bBlock(\varYY) \neq 0  ~*~ \bBlock(\varYY) =
  \bBlock(\varYY + \varZZ) * \eBlock(\varYY) = \eBlock(\varYY +
  \varZZ)$.

  \item Function $\ptrplus(\varX,\varY,\varZ)$ with the body $\varX := \varY
  \ptrplus \varZ$ for the case when the result points one byte past the block of
  the pointer to which an offset is added:
  \begin{displaymath}
     \SPEC{\varY = \varYY ~*~ \varZ = \varZZ ~*~ \varphi_{Y,Z}}
     {\varX := \varY \ptrplus \varZ}
     {\varX = \varYY + \varZZ ~*~ \varY = \varYY ~*~ \varZ = \varZZ  ~*~ \varphi_{Y,Z}}
  \end{displaymath} for $\varphi_{Y,Z} \equiv \bBlock(\varYY) \neq 0 ~*~ \varYY
  + \varZZ = \eBlock(\varYY)$.

  \item Function $\ptrminus(\varX,\varY,\varZ)$ with the body $\varX := \varY
  \ptrminus \varZ$:
  \begin{displaymath}
     \SPEC{\varY = \varYY ~*~ \varZ = \varZZ ~*~ \varphi_{Y,Z}}
     {\varX := \varY \ptrminus \varZ}
     {\varX = \varYY - \varZZ ~*~ \varY = \varYY ~*~ \varZ = \varZZ}
  \end{displaymath} for $\varphi_{Y,Z} \equiv \bBlock(\varYY)
  \neq 0 ~*~ \bBlock(\varYY) \leq \varZZ \leq \eBlock(\varYY)$.

  \item To deal with the $\memcpy(\varX, \varY, \varZ)$ statement, we use the
  below contract schema encoding concrete contracts for all the ways the address
  spaces $[x,x+z)$ and $[y,y+z)$ can be divided into $m,n \geq 1$ fields,
  respectively:
  \begin{displaymath}
     \SPEC{\varphi_{\varXX,\varYY,\varZZ} ~*~ \varphi_{\varXX,\varZZ} ~*~
     \varphi_{\varYY,\varZZ}}
     {\memcpy(\varX, \varY, \varZ)}
     {\varphi_{\varXX,\varYY,\varZZ} ~*~ \varphi'_{\varXX,\varZZ} *
     \varphi_{\varYY,\varZZ}}
  \end{displaymath} where\begin{itemize}

    \item $\varphi_{\varXX,\varYY,\varZZ} \equiv \varX = \varXX ~*~ \varY =
    \varYY ~*~ \varZ = \varZZ$,

    \item $\varphi_{\varXX,\varZZ} \equiv \ptstobyte{\varXX}{\top}{\varZZ}$,

    \item $\varphi_{\varYY,\varZZ} \equiv \varYY_1 = \varYY * \bigstar_{i=1}^{i
    \leq n} (\ptsto{\varYY_i}{\varBB_i} * \varYY_{i+1} = \varYY_i +
    \size(\varBB_i)) * \varYY_{n+1} = \varYY + \varZZ$,

    \item $\varphi'_{\varXX,\varZZ} \equiv \varXX_1 = \varXX * \bigstar_{i=1}^{i
    \leq n} (\ptsto{\varXX_i}{\varBB_i} * \varXX_{i+1} = \varXX_i +
    \size(\varBB_i)) * \varXX_{n+1} = \varXX + \varZZ$.

  \end{itemize} Of course, in practice, a particular instance of the contract
  schema will be chosen on-the-fly according to the current state of the
  computation.
  One starts with $\varphi_{\varYY,\varZZ} \equiv \ptsto{\varYY}{\varBB}$ and
  $\varphi'_{\varXX,\varZZ}\equiv \ptsto{\varXX}{\varBB}$ and uses the splitting
  bi-abduction rules on $\varphi_{\varYY,\varZZ}$ according to how the block
  pointed by $\varYY$ is partitioned.
  This splitting is synchronized with splitting $\varphi'_{\varXX,\varZZ}$ in
  the same way.
  Thus, the $\varZZ$ bytes stored at $\varYY$ will be transferred to $\varXX$
  preserving their interpretation.

  \item Function $\assume_\bowtie(\varY,\varZ)$ with the body $\assume(\varY
  \bowtie \varZ)$:
  \begin{displaymath}
     \SPEC{\varY = \varYY ~*~ \varZ = \varZZ}
     {\assume(\varY \bowtie \varZ)}
     {\varY = \varYY ~*~ \varZ = \varZZ ~*~ \varY \bowtie \varZ}
  \end{displaymath}

  \item Function $\assert_\bowtie(\varY,\varZ)$  with the body $\assert(\varY
  \bowtie \varZ)$:
  \begin{displaymath}
     \SPEC{\varY = \varYY ~*~ \varZ = \varZZ ~*~ \varY \bowtie \varZ}
     {\assert(\varY \bowtie \varZ)}
     {\varY = \varYY ~*~ \varZ = \varZZ ~*~ \varY \bowtie \varZ}
  \end{displaymath}

\end{itemize}

We now state the soundness of the contracts for the basic statement of our
programming language:

\begin{lemma} \label{lem:soundness-built-in} Let $\stmt$ be a basic statement
and let $\{\pre\} \ f(\varX_1,\ldots,\varX_n) \ \{ \post\}$ be a contract for
$\stmt$ as stated above.
Then, the contract is sound, i.e., for all frames $\Frame$, all configurations $(\stack,\blocks,\mem)$ such that $(\stack,\blocks,\mem) \models \Frame * \pre$, and all executions of $f(\varX_1,\ldots,\varX_n)$ that start from
$(\stack,\blocks,\mem)$ and end in some configuration
$(\stack',\blocks',\mem')$, it holds that $(\stack',\blocks',\mem') \models \Frame*\post$. \end{lemma}

\begin{proof} Direct from the semantics of our programming language as stated in
Section~\ref{sec:language-plus-semantics}. \end{proof}

\subsection{Contract Generation} \label{sec:contract-gen}

We now discuss the generation of contracts for arbitrary user-defined functions.
Our analysis proceeds along the call tree, starting from its leaves.
Hence, we can assume to already have computed contracts for nested function
calls.
(Recall that, in this paper, we limit ourselves to non-recursive functions.)
We first discuss a special case where we assume that the function to be analyzed
consists of a sequence of calls of other functions (i.e., no branching and
looping), and we assume the contracts of the nested function calls to be
conjunctive.
After that we state the general set-up of our analysis.

\subsubsection{Sequences of Function Calls with Conjunctive Contracts}

We now consider functions whose body consists of a sequence of function calls
$s_1;\ldots;s_l$ for some $l > 0$, whose contracts are conjunctive (this
includes for example all built-in functions apart from $\malloc$).
We formulate a symbolic execution that, given such a function
$f(\varX_1,\ldots,\varX_n)$, derives a sound contract for $f$.
The symbolic execution starts at the beginning of $f$ and maintains a pair of
formulae $\pre$ and $\post$, representing the so-far computed part of the
\emph{pre-condition} of the function $f$ and the \emph{current symbolic state}.
The symbolic execution will guarantee that configurations that satisfy $\pre$
lead to configurations satisfying $\post$ after executing the so-far analysed
statements.
$\pre$ and $\post$ will change throughout the symbolic execution because we keep
restricting the precondition $\pre$ and advancing the symbolic state $\post$.
The symbolic execution is set up such that the program variables
$\varX_1,\ldots,\varX_n$ may be updated, while all other variables will never be
modified (but, of course, fresh variables may be introduced and assigned at any
time).
The symbolic execution is initialised by introducing fresh logical variables
$\varXX_1,\ldots,\varXX_n$ and setting $\pre \equiv \post \equiv \varX_1 =
\varXX_1 * \cdots * \varX_n = \varXX_n$.
The symbolic execution then proceeds iteratively, considering the sequences
$s_1;\ldots;s_i$, starting with $i = 0$.
For every $0 \le i \le l$, the formulae $\pre$ and $\post$ are updated such that
they form a contract for the sequence $s_1;\ldots;s_i$.
We argue in the proof of Theorem~\ref{thm:correctness-symbolic-execution} below
that the symbolic execution maintains the contract $\{\pre\} \ s_1;\ldots;s_i \
\{\post\}$.
Note that our initialization of the formulae $\pre$ and $\post$ ensures the
correctness condition for $i=0$ and we derive a sound contract for $f$ once we
reach $i=l$.

We now describe how to apply a contract for a function call $s_i$.
We consider some function $\{C\} g(\varY_1,\ldots,\varY_m) \{D\}$ with $D \equiv
\exists \prefix_D. \ D_\postFormulaSubscript* D_\postEQSubscript$.
Let $s_i \equiv g(\varA_1,\ldots,\varA_m)$ be the function call with arguments
$\{\varA_1,\ldots,\varA_m\} \subseteq \{\varX_1,\ldots,\varX_n\}$, where we have
by assumption that the $\varA_i$ are pairwise different.
Let $\pre$ and $\post \equiv \exists \prefix_\post. \ \woProgVar * \ProgVarEQ$
be the current pre-condition and symbolic state, respectively.
We can assume w.l.o.g. that $\varIn{C} \cap \varIn{\woProgVar * \ProgVarEQ}
\subseteq \PVar$ because we can always rename the logical variables of a
contract (here, $\{C\} \ g(\varY_1,\ldots,\varY_m) \ \{D\}$) in order to avoid
name clashes.
Then, we consider the bi-abduction problem \begin{displaymath}
\woProgVar * [?] \models C' * [?] \quad \quad \quad \quad (\BAP)
\end{displaymath} where \begin{displaymath}
C' \equiv (C[\varA_i/\varY_i])[\expression_j/\varX_j]
\end{displaymath} is the formula $C$ from the precondition of
$g(\varY_1,\ldots,\varY_m)$ with the parameters $\varY_i$ substituted by the
arguments $\varA_i$, and the variables $\varX_j$ substituted by the expressions
$\expression_j$ according to the formula $\ProgVarEQ \equiv \varX_1 =
\expression_1 * \cdots * \varX_n = \expression_n$.
(Note that no formula in the bi-abduction problem contains any of the program
variables $\varX_1,\ldots,\varX_n, \varY_1,\ldots,\varY_m$.)
A solution to the bi-abduction problem ($\BAP$) consists of formulae $M$ and
$\Frame$, denoted as an \emph{antiframe} and a \emph{frame}, respectively, such
that \begin{displaymath} \woProgVar * M \models \exists \prefix. \ C' * \Frame
\quad \quad \text{ where } \quad \quad \prefix = \varIn{C'*\Frame} \setminus
\varIn{\woProgVar * \ProgVarEQ * M}.  \end{displaymath}
We require $M$ to be a quantifier-free symbolic heap with $\varIn{M} \subseteq
\LVar \setminus \prefix_\post$ (note that all existentially quantified variables
in $\prefix_\post$ are introduced during the symbolic execution, hence we cannot
restrict them at the beginning of the function where they are not known yet).
We also require  $\Frame$ to be a quantifier-free symbolic heap with
$\varIn{\Frame} \subseteq \LVar$.
Note that any such $M$ and $\Frame$ ensure that the function call $s_i \equiv
g(\varA_1,\ldots,\varA_m)$ can be safely executed:
The soundness of the contract $\{C\} \ g(\varY_1,\ldots,\varY_m) \ \{D\}$
guarantees that $g(\varY_1,\ldots,\varY_m)$ can be executed for any valuation of
the universally quantified variables; hence, it is sufficient to find some
valuation (as required by the existential quantification in $\exists \prefix. \
C' * \Frame$).
%
%
%
Given a solution $M$ and $F$ to the bi-abduction problem ($\BAP$), we can then
update the current state of the symbolic execution (we denote the new pre- and
post-condition by $\pre_\after$ and $\post_\after$, respectively):
\begin{itemize}

  \item $\pre_\after := M * \pre$.

  \item $\woProgVar' := \Frame * D_\postFormulaSubscript$.
  %

  \item $\ProgVarEQ'$ is the (separating) conjunction of $\varX_i =
  \expression_j'$
  %
  %
  for each $\varX_i$ passed by reference as the argument $\varA_j$ to $g$ where
  $\varA_j = \expression_j'$ is the final value according to the post-condition
  $D_\postEQSubscript$; and $\varX_i = \expression_i$ for each variable not
  passed to $g$ where $\varX_i = \expression_i$ is according to $\ProgVarEQ$
  before the call of $g$.

  \item $\post_\after := \exists \prefix_{\post_\after}. \ \woProgVar' *
  \ProgVarEQ'$ where $\prefix_{\post_\after} = (\varIn{\woProgVar' *
  \ProgVarEQ'} \cap \LVar ) \setminus \varIn{\pre_\after}$.
  %

\end{itemize}
For later use, we write $(\pre_\after,\post_\after) :=
\biabduct(\pre,\post,g(\varA_1,\ldots,\varA_m),C,D)$ for the result of advancing
the symbolic execution by one step wrt the bi-abduction solution.

\textbf{Formula Simplification by Quantifier Elimination.}
We will always try to simplify formulas with existentially quantified formulas.
For this, we assume the existence of a quantifier elimination procedure $\elim(\prefix, \phi)$, which attempts to eliminate as many variables $\varU \in \prefix$ as possible from a quantifier-free symbolic heap $\phi$.
We require that $\exists \prefix'. \psi$ is equivalent to $\exists \prefix. \phi$, where $\psi = \elim(\prefix, \phi)$ is the formula returned by the quantifier elimination procedure and $\prefix'$ is the set $\prefix$ minus the eliminated variables.
We use the following simple quantifier elimination procedure $\elim(\prefix, \phi)$ in our implementation:
For every equality in $\varU = \expression$ in $\phi$ with $\varU \in \prefix$,
we replace $\varU$ everywhere in $\phi$ by $\expression$ and then delete the
equality $\varU = \expression$. It is easy to verify that this procedure satisfies the above requirement.


\begin{example}
\label{ex:abdInSeq}

We consider the call $g(x)$ of a function with contract $\set{C} g(y) \set{D}$ during the analysis of some function $f(x)$, where the so-far derived precondition $\pre$, the current symbolic state $\post$ as well as the pre- and post-conditions $C$ and $D$ look as follows:

\begin{itemize}
  \item $\pre \equiv \ptsto{\varXX}{\varA} * \ptsto{\varXX+8}{\varB} *  \varX=\varXX$;

  \item $\post \equiv \exists \prefix_\post. \ \woProgVar * \ProgVarEQ$ where $\woProgVar \equiv \ptsto{\varXX}{\varA} * \ptsto{\varXX+8}{\varZ}$,  $\ProgVarEQ \equiv \varX=\varXX$ and $\prefix_\post = \{z\}$;

  \item $C \equiv \ptsto{\varYY}{\varU} * \ptsto{\varYY+8}{\varW} * \ptsto{\varU}{\varV} * \varY=\varYY$, and

  \item $D \equiv \exists\varC.\ D_\postFormulaSubscript * D_\postEQSubscript$ where $D_\postFormulaSubscript \equiv
    \ptsto{\varYY}{\varU} * \ptsto{\varYY+8}{\varW} * \ptsto{\varU}{\varV} * \ptsto{\varV}{\varC}$
  and $D_\postEQSubscript \equiv \varY=\varV$.

\end{itemize}

\noindent The formula $C'$ wrt which we will be solving the bi-abduction problem will look as follows:
$$C'\equiv (C[\varX/\varY])[\varXX/\varX] \equiv \ptsto{\varYY}{\varU} * \ptsto{\varYY+8}{\varW} * \ptsto{\varU}{\varV} * \varXX=\varYY.$$
Hence, we need to solve the bi-abduction problem
$$\ptsto{\varXX}{\varA} * \ptsto{\varXX+8}{\varZ} *
[?] \models \ptsto{\varYY}{\varU} * \ptsto{\varYY+8}{\varW} * \ptsto{\varU}{\varV} * \varXX=\varYY* [?].$$
Then, our bi-abduction procedure returns
$$M \equiv \ptsto{\varA}{\varV} \quad \quad \text{and} \quad \quad \Frame \equiv  \varXX=\varYY * \varA=\varU * \varZ = \varW,$$
i.e., we have
$$\woProgVar * M \models \exists \prefix. \ C' * \Frame
\quad \text{ where } \quad \prefix = \varIn{C'*\Frame} \setminus
\varIn{\woProgVar * \ProgVarEQ * M} = \{Y,u,w\}.$$
(For how $M$ and $\Frame$ were generated, see the description of the bi-abduction
procedure below.)
The new missing precondition $\pre_\after$ of the function $f(\varX)$ and the new
current symbolic state $\post_\after$ of $f(\varX)$ after the call of $g(\varX)$ is then as follows:

\begin{itemize}

  \item $\pre_\after \equiv \pre * M \equiv \ptsto{\varXX}{\varA} * \ptsto{\varXX+8}{\varB} *  \varX=\varXX *   \ptsto{\varA}{\varV}$,

  \item $\woProgVar' \equiv \Frame *
  D_\postFormulaSubscript \equiv \varXX=\varYY * \varA=\varU * \varZ = \varW * \ptsto{\varYY}{\varU} * \ptsto{\varYY+8}{\varW} * \ptsto{\varU}{\varV} * \ptsto{\varV}{\varC}$,

  \item $\post'_\postEQSubscript \equiv \varX=\varV$,

  \item $\post_\after \equiv \exists \varC,\varU,\varW, \varYY,\varZ.\   \varXX=\varYY * \varA=\varU * \varZ = \varW * \ptsto{\varYY}{\varU} * \ptsto{\varYY+8}{\varW} * \ptsto{\varU}{\varV} * \ptsto{\varV}{\varC} * \varX=\varV$,

  \item we then intend to simplify $\post_\after$ and compute

      \vspace{-0.6cm}
      $$\elim(\{\varC,\varU,\varW, \varYY,\varZ\},\woProgVar'* \post'_\postEQSubscript) \equiv \ptsto{\varXX}{\varA} * \ptsto{\varXX+8}{\varW} * \ptsto{\varA}{\varV} * \ptsto{\varV}{\varC} * \varX=\varV,$$

      \vspace{-0.2cm}
      i.e., we obtain

      \vspace{-0.3cm}
      $$ \hspace{-3.8cm}\post_\after^\mathit{simp} \equiv
      \exists \varC,\varW.\  \ptsto{\varXX}{\varA} * \ptsto{\varXX+8}{\varW} * \ptsto{\varA}{\varV} * \ptsto{\varV}{\varC} * \varX=\varV.$$
\end{itemize}

\end{example}
We now state the soundness of our approach:

\begin{theorem} \label{thm:correctness-symbolic-execution} Let
$f(\varX_1,\ldots,\varX_n)$ be some function and let $\{\pre\} \
f(\varX_1,\ldots,\varX_n) \ \{ \post\}$ by the contract inferred by the
bi-abductive inference as stated above.
Then, the contract is sound, i.e., for all frames $\Frame$, all configurations
$(\stack,\blocks,\mem)$ with $(\stack,\blocks,\mem) \models \Frame*\pre$, and
all executions of $f(\varX_1,\ldots,\varX_n)$ that start with
$(\stack,\blocks,\mem)$ and end with some configuration
$(\stack',\blocks',\mem')$ we have $(\stack',\blocks',\mem') \models
\Frame*\post$.  \end{theorem}

\begin{proof} We prove that $\set{\pre} s_1;\ldots;s_i  \set{\post}$ for all $0
\le i \le l$ (*), where $\pre$ and $\post$ are the formulae after step $i$ of
the symbolic execution.
We prove (*) by induction on $i$.
Note that (*) holds for $i=0$ due to the initialization of $\pre$ and $\post$.

We now consider some $i >0$ and the sequence $s_1;\ldots;s_i$.
By induction assumption we have that $\set{\pre} s_1;\ldots;s_i \set{\post}$.
Let $s_{i+1} \equiv g(\varA_1,\ldots,\varA_m)$ be some function call for which
we assume the contract $\{C\} \ g(\varY_1,\ldots,\varY_m) \ \{D\}$  with $D
\equiv \exists \prefix_D. \ D_\postFormulaSubscript* D_\postEQSubscript$, whose
soundness has already been established.
%
%
Let $M$ and $\Frame$ be a solution to the biabuction problem, i.e., let $M$ and
$\Frame$ be formulae with $\woProgVar * M \models \exists \prefix. \ C' *
\Frame$ for $\prefix = \varIn{C'*\Frame} \setminus \varIn{\woProgVar *
\ProgVarEQ * M}$ (\#) where $C'$ is the formula $C$ after parameter
instantiation as described above.
We will now argue that $\set{\pre_\after}
s_1;\ldots;s_i;s_{i+1} \set{\post_\after}$ where $\pre_\after$ and $\post_\after$ are the pre- and post-condition after step $i+1$.
Let $\FrameAlt$ be some frame.
We now instantiate the frame in the induction
assumption with $\FrameAlt * M$, i.e., we have that $\set{\FrameAlt * M * \pre}
s_1;\ldots;s_i \set{\FrameAlt * M * \post}$.
Let $(\stack,\blocks,\mem)$ be some configuration with $(\stack,\blocks,\mem)
\models \FrameAlt * M * \pre$ and let $(\stack',\blocks',\mem')$ be the result
of executing $s_1;\ldots;s_i$ starting from configuration
$(\stack,\blocks,\mem)$.
From the instantiation of the induction assumption, we get that
$(\stack',\blocks',\mem') \models \FrameAlt * M * \post$.
By (\#), we have that $(\stack',\blocks',\mem') \models  \FrameAlt * \exists
\prefix_\post. \ ( \ProgVarEQ * \exists \prefix. \ ( C' * \Frame))$.
Hence, there is some stack $\stack''$ with $\dom(\stack'') = \dom(\stack') \cup
\prefix_\post \cup \prefix \cup \{\varY_1,\ldots,\varY_m\}$ such that
$(\stack'',\blocks',\mem') \models \FrameAlt * \Frame * C$ and $\stack''(y_j) =
\stack''(x_i)$ according to $\varX_i = \varA_j$ in the function call
$g(\varA_1,\ldots,\varA_m)$.
Let $(\stack''',\blocks'',\mem'')$ be the configuration that results from
executing $g(\varY_1,\ldots,\varY_m)$ starting from $(\stack'',\blocks',\mem')$.
By the correctness of $\{C\} \ g(\varY_1,\ldots,\varY_m) \ \{D\}$, we get that
$(\stack''',\blocks'',\mem'') \models \FrameAlt * \Frame * D$.
We now consider the stack $\stack^\circ$ defined by $\stack^\circ(\varX_i) =
\stack'''(\varY_j)$ for all $1 \le i \le m$ with $\varX_i = \varA_j$ in the
function call $g(\varA_1,\ldots,\varA_m)$, and $\stack^\circ(\varX_i) =
\stack(\varX_i)$ for all $\varX_i \in \Var \setminus
\{\varA_1,\ldots,\varA_n\}$.
Note that $(\stack^\circ,\blocks'',\mem'')$ is the resulting configuration after
returning from the function call $g(\varA_1,\ldots,\varA_m)$.
By the above, we then have $(\stack^\circ,\blocks'',\mem'') \models \FrameAlt *
\exists (\prefix \cup \prefix_\post \cup \prefix_D). \ (\Frame *
D_\postFormulaSubscript * \post_\postEQSubscript' )$.
Because of $(\prefix \cup \prefix_\post \cup \prefix_D) \cap \varIn{\pre_\after}
= \emptyset$, we then get $(\stack^\circ,\blocks'',\mem'') \models \FrameAlt *
\exists \prefix_{\post_\after}. \ (\post_\postFormulaSubscript' *
\post_\postEQSubscript' )$.
This establishes the claim. \end{proof}

\subsubsection{Branching, Looping and Disjunctive Contracts}
\label{section:branching-looping}

So far we have described the generation of contracts for functions $\set{\pre}
f(\varX_1,\dots,\varX_n)\set{\post}$ whose body consists of a sequence of
function calls with conjunctive contracts and without branching and looping.
In this section, we lift all these restrictions.
We follow~\cite{BiAbd11} and present a two-round analysis for the general case (we
present a short summary of these two analysis phases here to make the paper
self-contained, but we refer the reader to~\cite{BiAbd11} for a more detailed
exposition and for a formal statement on the soundness of the analysis):
The first round (called \texttt{PreGen} in~\cite{BiAbd11}) infers a set of
pre-/post-condition pairs $(\pre,\post)$, but in contrast to
Theorem~\ref{thm:correctness-symbolic-execution} there is no guarantee about the
soundness of the inferred $(\pre,\post)$.
For each pre-/post-condition pair $(\pre,\post)$ computed in the first round,
the second round (called \texttt{PostGen} in~\cite{BiAbd11}) discards the
post-condition $\post$ and re-starts the symbolic execution from the
pre-condition $\pre$ \emph{not allowing the strengthening of the pre-condition
throughout the symbolic execution}, which either fails or results in a set of
pre-/post-condition pairs $(\pre,\post_1),\ldots, (\pre,\post_l)$.
In the latter case, we return $(\pre,\post_1 \vee \cdots \vee \post_l)$, which
is guaranteed to be a~sound contract.

We now describe the general set-up that is used by the first as well as the
second analysis round.
Recall that we assume a control flow graph $(V,E,\eentry,\eexit)$ for each function
$f(\varX_1,\dots,\varX_n)$ where $V$ is a~set of \emph{nodes}, there are
dedicated nodes $\eentry,\eexit \in V$, and $E \subseteq V \times V$ is a set of
edges where each edge is labelled by a function call (recall that we also model
basic statements as functions).
By $V_\close\subseteq V$, we denote a set of \emph{cut-points}, where each loop
must contain at least one cut-point (usually the header location of a loop).
A function without loops has $V_\close = \emptyset$.
The symbolic execution maintains a mapping $\symb: V \rightarrow
(\mathbf{2}_\mathit{fin})^\contracts$ where each location $\cloc \in V$ is
mapped to a finite set of pre-/post-condition pairs $(\pre,\post) \in
\contracts$.
Initially, $\symb(\eentry)=\set{(\pre_\iinit,\post_\iinit)}$ where $\pre_\iinit$
and $\post_\iinit$ are specified as stated below for the first and second
analysis round, and $\symb(\cloc) = \emptyset$ for all $\cloc \neq \eentry$.

The symbolic execution then proceeds as a work-list algorithm.
In each computational step, the symbolic execution picks a pre-/post-condition
pair $(\pre,\post) \in \symb(\cloc)$, for some $\cloc \in V$, that has not been
processed yet.
Then, the symbolic execution performs the following for each outgoing edge
$(\cloc,\cloc')\in E$:
Let $g(\varA_1,\dots,\varA_k)$ be the function call labelling the edge
$(\cloc,\cloc')$.
For each contract $\set{C} g(\varY_1,\dots,\varY_k)\set{D_1 \vee \cdots \vee
D_l}$ and each disjunct $D_i$ of such a contract, we invoke bi-abduction procedure from
Section \ref{sec:contract-gen}, i.e., let $(\pre_\after,\post_\after)
= \biabduct(\pre,\post,g(\varA_1,\ldots,\varA_m),C,D_i)$; then,
\begin{enumerate}

  \item for $\cloc' \not\in V_\close$, we add $(\pre_\after,\post_\after)$ into $\symb(\cloc')$, and

  \item for $\cloc' \in V_\close$, \begin{itemize}

    \item we check whether there is a $(\pre',\post')\in \symb(\cloc')$ that \emph{covers}
    $(\pre_\after,\post_\after)$, i.e., we have  $\exists \LVar(\pre_\after). \
    \pre_\after \models \exists \LVar(\pre'). \ \pre'$ and $\exists \LVar(\post_\after). \
    \post_\after \models \exists \LVar(\post'). \ \post'$, where $\LVar(\phi) =
    \varIn{\phi} \cap \LVar$ denotes the free logical variables appearing in the formula;
	we note that a covering guarantees that the all configurations that satisfy the
	pre- resp. post condition of the pair $(\pre_\after,\post_\after)$ already satisfy
	the corresponding condition of $(\pre',\post')$ (for some suitable instantiation
	of the logical variables); hence, such $(\pre_\after,\post_\after)$ do not need to
	be added to $\symb(\cloc')$, which supports the termination of the fixed point
	termination;


    \item if there is no $(\pre',\post')\in \symb(\cloc')$ that covers $(\pre_\after,\post_\after)$, we add    $(\alpha(\pre_\after),\alpha(\post_\after))$ into $\symb(\cloc')$, where $\alpha$ is a widening procedure in the form of a \emph{list abstraction} that is quite common in the area---cf., e.g., \cite{SLNestedLists07}.
        In its simplest form, the abstraction $\alpha$ searches for patterns of the form $\segment(\varX,\varY) * \segment(\varY,\varZ)$ (resp. $\segment(\varX,\varV,\varY) * \segment(\varY,\varX,\varZ)$) and replaces them by $\sll{\segment(\varA,\varB)}(\varX, \varZ)$ (resp.   $\dll{\segment(\varA,\varB,\varC)}(\varX,\varY,\varV,\varZ)$) provided that there is no pointer nor other list segment incoming to $\varY$ (i.e., the current symbolic state cannot, e.g., imply $\ptsto{\varU}{\varY} * \varU \neq \varX$).
        The actual abstraction is, of course, more complex---e.g., apart from applying to sequences of the $\segment$ predicates as above, which of course can be longer than just two appearances of $\segment$, it also applies to sequences consisting of the $\segment$ predicates and of compatible    singly-/doubly-linked list segments (e.g., $\segment(\varX,\varY) * \sll{\segment(\varA,\varB)}(\varY, \varZ)$ may abstract to $\sll{\segment(\varA,\varB)}(\varX, \varZ)$).
        Since, however, the abstraction is rather standard in the area, we do not develop it further here and refer the interested reader to \cite{SLNestedLists07}.
  \end{itemize}

\end{enumerate}
The worklist algorithm continues until a fixed point is reached,
i.e., until no more pre-/post-condition pairs are added to $\symb$ and all pairs have been processed.

Next, we add some more specific details about the two rounds of the
analysis:
\begin{description}

  \item[1. Round:] We set $\pre_\iinit \equiv \post_\iinit \equiv \varX_1 =
  \varXX_1 * \cdots * \varX_n = \varXX_n$ (recall that we assume
  $\varX_1,\dots,\varX_n$ to be all the program variables that appear in the
  body of $f(\varX_1,\dots,\varX_n)$).
  Once the worklist algorithm has reached a fixed point, we return all
  pre-/post-condition pairs $(\pre,\post) \in \symb(\eexit)$ as the result of
  the first analysis round.

  \item[2. Round:] For each $(\pre,\post)$ computed by the first analysis round,
  there is a second analysis round where we set $\pre_\iinit \equiv \post_\iinit
  \equiv \pre$ and where we \emph{disallow the strengthening of the
  pre-condition throughout the symbolic execution} (i.e., we require $M = \emp$
  for each bi-abduction call performed during symbolic execution; if this is not
  the case, we fail).
  Once the worklist algorithm has reached a fixed point, we consider the
  computed pre-/post-condition pairs $\symb(\eexit) = \{(\pre,\post_1), \ldots,
  (\pre,\post_l)\}$ (note that all these pairs must have the same pre-condition
  $\pre$) and return the contract $(\pre,\post_1 \vee \cdots \vee \post_l)$.
  We point out to the reader that the second round in general creates multiple
  contracts---one contract for each possible pre-condition.

\end{description}

\paragraph*{Assume and Assert Statements}

As we have already said above, we model branching and looping as usual, i.e., by
having multiple outgoing control flow edges labelled by an $\assume$ statement
with appropriate conditions.
However, following~\cite{BiAbd11}, we also observe that it is sometimes beneficial
to replace $\assume$ statements by $\assert$ statements because this adds
context sensitivity to the analysis.
For example, when we encounter an \texttt{if}-statement with the condition
$\varX = \varY$ (where $\varX$ and $\varY$ are paremeters of the function under
analysis), replacing $\assume$ by $\assert$ statements leads to (at least) two
distinct contracts with preconditions $\varX = \varY * \cdots$ and $\varX \neq
\varY * \cdots$.
Hence, we try to treat all $\assume$ as $\assert$ statements.
However, we revert back to $\assume$ statements in case using $\assert$ leads to
an analysis failure---in particular, this happens when the variables used in the
conditions are not parameters of the function under analysis and hence not
controllable through the contracts.

\bigskip

We refer an interested reader to Appendix~\ref{app:illustr-cond-loops} where we
provide several examples illustrating the analysis of conditions and loops,
using abstraction and the second analysis round.

\section{Bi-Abduction Procedure} \label{sec:bi-abduction}

Assume that
%
%
%
$\post \equiv \exists\prefix_\post. \ \woProgVar * \ProgVarEQ$ is the current
symbolic state.
%
%
We now explain how to compute a solution to the bi-abduction problem $\woProgVar
* [?] \models C' * [?] \ (\BAP)$.
That is, we show how to compute an \emph{anti-frame} $M$ and a \emph{frame}
$\Frame$ such that $\woProgVar * M \models \exists \prefix. \ C' * \Frame$ where
$\prefix = \varIn{C'*\Frame} \setminus \varIn{\woProgVar * \ProgVarEQ * M}$ and
where $M$ does not contain any quantified variable from $\prefix_\post$.
Following~\cite{BiAbd11}, we proceed in three steps: \begin{enumerate}

  \item We solve the \emph{abduction problem} $\woProgVar  * [?] \models C' *
  \True$ and compute a quantifier-free symbolic heap $M'$ such that $\woProgVar
  * M' \models C' * \True$.
  (Note that, in this first step, the solution $M'$ is allowed to contain
  variables from $\prefix_\post$.)
  %
  %
  We state our rules for finding a solution to the abduction problem $\woProgVar
  * [?] \models C' * \True$ in Section~\ref{sec:abduction-rules}.

  \textbf{Example~\ref{ex:abdInSeq} (continued).}
  For $C'\equiv (C[\varY/\varX])[\varXX/\varX] \equiv \ptsto{\varYY}{\varU} * \ptsto{\varYY+8}{\varW} * \ptsto{\varU}{\varV} * \varXX=\varYY$
  and $\woProgVar \equiv \ptsto{\varXX}{\varA} * \ptsto{\varXX+8}{\varZ}$, we obtain the solution $M' \equiv \varXX=\varYY * \varA=\varU * \varZ = \varW * \ptsto{\varU}{\varV}$ to the abduction
  problem $\woProgVar  * [?] \models C' * \True$.

  \item We solve the \emph{frame problem} $\woProgVar  * M' \models C' * [?]$
  and compute a quantifier-free symbolic heap $\Frame$ such that $\woProgVar  *
  M' \models C' * \Frame$.
  $\Frame$ is computed as a by-product
  of our rules for finding the abduction solution $M'$, and no special frame
  inference procedure is needed---see Section \ref{sec:frame-inference}.

  \textbf{Example~\ref{ex:abdInSeq} (continued).}
  We obtain the solution $\Frame
  \equiv  \varXX=\varYY * \varA=\varU * \varZ = \varW$ to the frame problem $\woProgVar  * M' \models C' * [?]$.

  \item We finally compute a formula $M$ such that $\woProgVar * M$ is equivalent to $\woProgVar * \exists  K. \ M'$ where $K = \varIn{M'} \setminus \varIn{\woProgVar * \ProgVarEQ * M}$.
      The objective is to minimize the missing anti-frame $M$ as much as possible and to eliminate all occurrences of variables  $\prefix_\post$ in $M'$ (which, however, is not always possible).
      Ideally, we obtain $M \equiv \emp$ because this means that no strengthening of the precondition is needed (which is required in the second round of the general procedure as described in Section~\ref{section:branching-looping}).
      In case $M$ does not contain any occurrences of variables $\prefix_\post$, we then return $M$ and $\Frame$ as solutions to the bi-abduction problem ($\BAP$);
      otherwise, we fail.

      The minimization of $M'$ proceeds in two simple steps:
      \begin{enumerate}
        \item We compute $M'' = \elim(\prefix,M')$ for $\prefix = \varIn{M'} \setminus \varIn{\woProgVar * \ProgVarEQ}$, eliminating as many variables as possible from $M'$.


    \item If $M''$ contains some pure subformula $\expression \bowtie \expressionAlt$ with $\woProgVar \models \expression \bowtie \expressionAlt * \True$, we delete this formula from $M''$.
        If we cannot delete any further pure subformula, we return the resulting formula $M''$ as the result $M$.


  \end{enumerate}


  \textbf{Example~\ref{ex:abdInSeq} (continued).}
  For $M'$, $\woProgVar$ and $\ProgVarEQ$ as above, we obtain $M \equiv \ptsto{\varU}{\varV}$.
  %
\end{enumerate}

We now state the soundness of the bi-abduction procedure:

\begin{lemma}
\label{lem:bi-abduction-correctness}
 Assume that $\post \equiv \exists \prefix_\post. \ \woProgVar *
\ProgVarEQ$ is the symbolic state.
Let $M$ and $\Frame$ be some formulae returned by the bi-abduction procedure.
Then, $\woProgVar * M \models \exists \prefix. \ C' * \Frame$ for $\prefix =
\varIn{C'*\Frame} \setminus \varIn{\woProgVar * \ProgVarEQ * M}$. \end{lemma}

\begin{proof} By Steps 1 and 2 of the bi-abduction procedure, we have $\woProgVar
* M' \models C' * \Frame$ (*).
By Step~3 of the bi-abduction procedure, we have $\woProgVar * M \models
\woProgVar * \exists K. \ M'$ (\#) where $K = \varIn{M'} \setminus
\varIn{\woProgVar * \ProgVarEQ * M}$.
As it is sound to move the quantifiers to the front (note that $K\cap\woProgVar
= \emptyset$), we obtain from~(\#) that $\woProgVar * M \models \exists K. \
\woProgVar *  M'$.
With (*), we get that $\woProgVar * M \models \exists K. \ C' * \Frame$.
We then note that we can drop any variable $K' = K \setminus \varIn{C' *
\Frame}$ from the existential quantification, i.e., we get $\woProgVar * M
\models \exists K'. \ C' * \Frame$.
We observe that $K' \subseteq \prefix = \varIn{C' * \Frame} \setminus
\varIn{\woProgVar * \ProgVarEQ * M}$.
We finally observe that we can existentially quantify over every variable
$\prefix \setminus K'$ as this only weakens the formulae on the right-hand side.
Hence, we obtain $\woProgVar * M \models \exists \prefix. \ C' * \Frame$.
\end{proof}


\subsection{Abduction Rules} \label{sec:abduction-rules}

We now state our rules for computing a solution to the abduction problem.
In the below rules, we will use the notation $\phi * [M] \EMA \psi$
to denote that we are deriving the solution M to the abduction problem $\varphi * [?] \models \psi$.
The rules are to be applied in the stated order.\footnote{As for non-determinism
within single rules, which can sometimes be applied in multiple ways, our
implementation currently uses the first applicable option (with backtracking to
the other options only in case that the first option turns out to result in an
unsatisfiable abduction strategy). A better strategy is an open question for
future research.}

We start with a rule allowing us to learn \emph{missing pure constraints}.

$$\inference[\texttt{learn-pure}]{
  \ABD{
    \varphi * \pure
  }{
    M
  }{
    \psi
  }}{
  \ABD{
    \varphi
  }{
    \pure * M
  }{
    \psi * \pure
  }
}[$\pure$ pure]
$$\vspace*{-2mm}

The \texttt{match} rule presented below allows one to \emph{match points-to
predicates} from the LHS and RHS that have the same source location ($\varepsilon =
\varepsilon'$) and points-to fields $\zeta$, $\zeta'$ of the same size.
Then we learn that the target fields are the same too.
We note here that this rule is as a special case of the
\texttt{split-pt-pt-right} rule presented further on, but we show it here as an
easy case to start from.
As we will discuss in
Section~\ref{sec:Entailment}, we discharge entailment checks of the form
$\varphi_1 \models \varphi_2 *\True$ where $\varphi_2$ is a pure formula (e.g.
$\varepsilon = \varepsilon'$) by checking
unsatisfiability of the formula $\psi_1 \wedge \neg \psi_2$ where $\psi_i$ is a
translation of the SL formula $\varphi_i$ to bitvector logic. We will sketch
the translation procedure in Section~\ref{sec:SL2BV}.

\vspace*{-4mm}$$\inference[\texttt{match}]{
  \ABD{
    \varphi * \expressionAlt = \expressionAlt'
  }{
    M
  }{
    \psi
  }}{
  \ABD{
    \varphi * \ptsto{\expression}{\expressionAlt}
  }{
    \expressionAlt=\expressionAlt' * M
  }{
    \psi * \ptsto{\expression'}{\expressionAlt'}
  }
}[$\size(\expressionAlt) = \size(\expressionAlt')$ and
  $\varphi \models \expression = \expression' * \True$]
$$\vspace*{-2mm}

As illustrated in Fig.~\ref{fig:split-pt-pt-right}, the next presented
\texttt{split-pt-pt-right} rule allows one to deal with pointers $\varepsilon$,
$\varepsilon'$ to fields $\zeta$, $\zeta'$ that lie at possibly different addresses
but within blocks of the same base address.
Moreover, the RHS target field $\zeta'$ can be larger.
In this case, the field $\zeta'$ is \emph{split} to three \emph{byte sequences}
$\substitution{\expressionAlt'}{0}{k}$,
$\substitution{\expressionAlt'}{k}{k+l}$, and
$\substitution{\expressionAlt'}{k+l}{l'}$,
some of which can be empty, and the middle byte sequence is matched with the LHS target field $\zeta$.
(We recall that $\substitution{\val}{i}{j}$ denotes the substring of $\val$ that starts at index $i$ and ends at index $j$.)

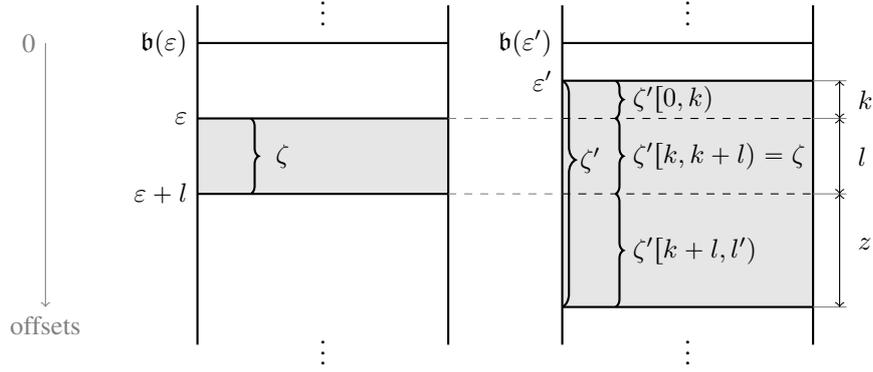
\begin{figure}[t]
  \centering
  \begin{tikzpicture}[line width=0.8pt]
  \pgfmathsetmacro{\xw}{3.3};
  \pgfmathsetmacro{\xa}{2};
  \pgfmathsetmacro{\xb}{\xa+\xw};
  \pgfmathsetmacro{\xc}{\xa+\xw+1.5};
  \pgfmathsetmacro{\xd}{\xc+\xw};
  \pgfmathsetmacro{\ya}{0.5};
  \pgfmathsetmacro{\ye}{4.0};

  \draw[->,thin,gray] (0,\ye) -- (0,0.5) node[below]{offsets};
  \draw (0,\ye) node[left,gray]{0};
  \filldraw[draw=black,color=gray!20] (\xa,2) rectangle (\xb,3);
  \draw (\xa,0) -- (\xa,\ye+0.5) ++(\xw/2,0) node{\vdots};
  \draw (\xb,\ye+0.5) -- (\xb,0) ++(-\xw/2,0) node{\vdots};
  \draw (\xa,\ye) node[left]{$\bBlock(\varepsilon)$} -- (\xb,\ye);
  \draw (\xa,3) node[left]{$\varepsilon$} -- (\xb,3);
  \draw (\xa,2) node[left]{$\varepsilon+l$} -- (\xb,2);
  \filldraw[draw=black,color=gray!20] (\xc,\ya) rectangle (\xd,3.5);
  \draw (\xc,0) -- (\xc,\ye+0.5) ++(\xw/2,0) node{\vdots};
  \draw (\xd,\ye+0.5) -- (\xd,0) ++(-\xw/2,0) node{\vdots};
  \draw (\xc,\ye) node[left]{$\bBlock(\varepsilon')$} -- (\xd,\ye);
  \draw (\xc,3.5) node[left]{$\varepsilon'$} -- (\xd,3.5);
  \draw[thin,dashed] (\xc,3) -- (\xd,3);
  \draw[thin,dashed] (\xc,2) -- (\xd,2);
  \draw (\xc,\ya) -- (\xd,\ya);
  \draw[gray,thin,dashed] (\xb,3) -- (\xc,3);
  \draw[gray,thin,dashed] (\xb,2) -- (\xc,2);
  \draw [decorate,decoration={brace,amplitude=3pt},xshift=0pt,yshift=0pt]
        (\xa+0.7,3) -- (\xa+0.7,2) node [right,black,midway,xshift=0.2cm] {$\zeta$};

  \pgfmathsetmacro{\ts}{0.35};          
  \pgfmathsetmacro{\tx}{\ts+0.09};      
  \pgfmathsetmacro{\txtx}{\ts+0.11};      
  \draw (\xd+\txtx,3.25) node[right]{$k$};
  \draw (\xd+\txtx,2.50) node[right]{$l$};
  \draw (\xd+\txtx,1.55) node[below right,align=left]{$z$};
  \draw [decorate,decoration={brace,amplitude=3pt}]
        (\xc+0.7,3.5) -- (\xc+0.7,3.0) node [right,black,midway,xshift=0.1cm] {$\substitution{\expressionAlt'}{0}{k}$};
  \draw [decorate,decoration={brace,amplitude=3pt}]
        (\xc+0.7,3) -- (\xc+0.7,2) node [right,black,midway,xshift=0.1cm] {$\substitution{\expressionAlt'}{k}{k+l}=\zeta$};
  \draw [decorate,decoration={brace,amplitude=3pt}]
        (\xc+0.7,2) -- (\xc+0.7,\ya) node [right,black,midway,xshift=0.1cm] {$\substitution{\expressionAlt'}{k+l}{l'}$};
  \draw [decorate,decoration={brace,amplitude=5pt,aspect=0.35}]
        (\xc,3.5) -- (\xc,\ya) node [right,black,pos=0.35,xshift=0.1cm] {$\zeta'$};
  \draw[gray,thin] (\xd,3.5) -- ++(\tx,0);
  \draw[gray,thin] (\xd,3.0) -- ++(\tx,0);
  \draw[gray,thin] (\xd,2.0) -- ++(\tx,0);
  \draw[gray,thin] (\xd,\ya) -- ++(\tx,0);
  \draw[<->,thin] (\xd+\ts,\ya) -- (\xd+\ts,2);
  \draw[<->,thin] (\xd+\ts,2) -- (\xd+\ts,3);
  \draw[<->,thin] (\xd+\ts,3) -- (\xd+\ts,3.5);

\end{tikzpicture}

  \caption{An illustration of the \texttt{split-pt-pt-right} rule where $z = l' - k - l$.}

  \label{fig:split-pt-pt-right}
\end{figure}

\vspace*{-4mm}$$\inference[\texttt{split-pt-pt-right}]{
  \ABD{
    \varphi * \expressionAlt = \substitution{\expressionAlt'}{k}{k+l}
  }{
    M
  }{
    \psi * \ptsto{\expression'}{\substitution{\expressionAlt'}{0}{k}} * \ptsto{(\expression + l)}{\substitution{\expressionAlt'}{k+l}{l'}}
  }}{
  \ABD{
    \varphi * \ptsto{\expression}{\expressionAlt}
  }{
    \expressionAlt = \substitution{\expressionAlt'}{k}{k+l} * M
  }{
    \psi * \ptsto{\expression'}{\expressionAlt'}
  }
}[$C$]$$\vspace*{-2mm}

\noindent In the above rule, the condition $C$ requires that there are some
$k,l,l' \in \N$ with $\varphi \models \bBlock(\expression) =
\bBlock(\expression') ~*~ \expression = \expression' + k ~*~ \True$,
$\size(\expressionAlt) = l$, $\size(\expressionAlt') = l'$, and $l + k \le l'$.
We note that, in the above formulation of the rule $\texttt{split-pt-pt-left}$,
we assume $0 < k$ and $k+l<l'$ in order to avoid cluttering the rule by
additional case distinctions; in the case of $0 = k$ or $k+l = l'$, however, we
need to remove  $\ptsto{\expression'}{\substitution{\expressionAlt'}{0}{k}}$ or
$\ptsto{(\expression + l)}{\substitution{\expressionAlt'}{k+l}{l'}}$,
respectively, from the RHS of the premise of the rule.
There is a symmetric rule \texttt{split-pt-pt-left} for the LHS.


The rule \texttt{split-pt-bl-right} presented below and illustrated in
Fig.~\ref{fig:split-pt-bl-right} is an analogy of the rule \texttt{split-pt-pt-right}
presented above, but, this time, with the RHS field, which is being split,
of non-constant size.
The rule covers both types of such fields that we allow: sequences of
bytes of undefined values (then $m = \top$ in the rule) or sequences of the same
byte (then $m \in Byte$).

\begin{figure}[t]
  \centering
  \begin{tikzpicture}[line width=0.8pt]
  \pgfmathsetmacro{\xw}{3.3};
  \pgfmathsetmacro{\xa}{2};
  \pgfmathsetmacro{\xb}{\xa+\xw};
  \pgfmathsetmacro{\xc}{\xa+\xw+1.5};
  \pgfmathsetmacro{\xd}{\xc+\xw};
  \pgfmathsetmacro{\ya}{0.5};
  \pgfmathsetmacro{\ye}{4.0};

  \draw[->,thin,gray] (0,\ye) -- (0,0.5) node[below]{offsets};
  \draw (0,\ye) node[left,gray]{0};
  \filldraw[draw=black,color=gray!20] (\xa,2) rectangle (\xb,3);
  \draw (\xa,0) -- (\xa,\ye+0.5) ++(\xw/2,0) node{\vdots};
  \draw (\xb,\ye+0.5) -- (\xb,0) ++(-\xw/2,0) node{\vdots};
  \draw (\xa,\ye) node[left]{$\bBlock(\varepsilon)$} -- (\xb,\ye);
  \draw (\xa,3) node[left]{$\varepsilon$} -- (\xb,3);
  \draw (\xa,2) node[left]{$\varepsilon+l$} -- (\xb,2);
  \filldraw[draw=black,color=gray!20] (\xc,\ya) rectangle (\xd,3.5);
  \draw (\xc,0) -- (\xc,\ye+0.5) ++(\xw/2,0) node{\vdots};
  \draw (\xd,\ye+0.5) -- (\xd,0) ++(-\xw/2,0) node{\vdots};
  \draw (\xc,\ye) node[left]{$\bBlock(\varepsilon')$} -- (\xd,\ye);
  \draw (\xc,3.5) node[left]{$\varepsilon'$} -- (\xd,3.5);
  \draw[thin,dashed] (\xc,3) -- (\xd,3);
  \draw[thin,dashed] (\xc,2) -- (\xd,2);
  \draw (\xc,\ya) -- (\xd,\ya);
  \draw[gray,thin,dashed] (\xb,3) -- (\xc,3);
  \draw[gray,thin,dashed] (\xb,2) -- (\xc,2);
  \draw [decorate,decoration={brace,amplitude=3pt},xshift=0pt,yshift=0pt]
        (\xa+0.7,3) -- (\xa+0.7,2) node [right,black,midway,xshift=0.2cm] {$\zeta$};

  \pgfmathsetmacro{\ts}{0.35};          
  \pgfmathsetmacro{\tx}{\ts+0.09};      
  \pgfmathsetmacro{\txtx}{\ts+0.11};      
  \draw [decorate,decoration={brace,amplitude=3pt}]
        (\xc+0.3,3.5) -- (\xc+0.3,3.0) node [right,black,midway,xshift=0.1cm] {$m[\varepsilon-\varepsilon']$};
 \draw [decorate,decoration={brace,amplitude=3pt}]
        (\xc+0.3,3) -- (\xc+0.3,2) node [right,black,midway,xshift=0.1cm] {$m[l]$};
  \draw [decorate,decoration={brace,amplitude=3pt}]
        (\xc+0.3,2) -- (\xc+0.3,\ya) node [right,black,midway,xshift=0.1cm]
{$m[z]$};
  \draw [decorate,decoration={brace,amplitude=5pt,aspect=0.35}]
        (\xd,3.5) -- (\xd,\ya) node [right,black,pos=0.35,xshift=0.1cm] {$m[\expressionSize]$};

\end{tikzpicture}

  \caption{An illustration of the \texttt{split-pt-bl-right} rule where $z =
  \expressionSize - (\expression - \expression') - l$.}

  \label{fig:split-pt-bl-right} \end{figure}

\vspace*{-4mm}$$\inference[\texttt{split-pt-bl-right}]{
  \ABD{
    \varphi * \chi
  }{
    M
  }{
    \psi * \ptstobyte{\expression'}{m}{\expression-\expression'} *
    \ptstobyte{(\expression + l)}{m}{\varZ} * K
  }}{
  \ABD{
    \varphi * \ptsto{\expression}{\expressionAlt}
  }{
    \chi * M
  }{
    \psi * \ptstobyte{\expression'}{m}{\expressionSize}
  }
}[$C$]$$\vspace*{-2mm}

\noindent In the rule, we require that $m = \top$ and $\chi \equiv \emp$, or $m
\in \Byte$ and $\chi \equiv \expressionAlt = m^l$.
Further, $\size(\expressionAlt) = l$, $C$ requires that $\varphi \models
\bBlock(\expression) = \bBlock(\expression') ~*~ \expression' \le \expression ~*~
\expression + l \le \expression' + \expressionSize ~*~
\True$, $\varZ$ is some fresh variable with $\size(\varZ) = N$, and $K \equiv
\varZ = \expressionSize - (\expression - \expression') - l$.
There is a symmetric rule \texttt{split-pt-bl-left} for the LHS.
%

We now present an analogy of the above rule for the case when we need to split a field of
constant size that appears on the RHS.
In order to be able to split the RHS field we will also require the LHS field to be of constant size.

\vspace*{-4mm}$$\inference[\texttt{split-bl-pt-right}]{
  \ABD{
    \varphi * \chi
  }{
    M
  }{
    \psi * \ptsto{\expression'}{\substitution{\expressionAlt'}{0}{k}} * \ptsto{(\expression + l)}{\substitution{\expressionAlt'}{k+l}{l'}}
  }}{
  \ABD{
    \varphi * \ptstobyte{\expression}{m}{\expressionSize}
  }{
    \chi * M
  }{
    \psi * \ptsto{\expression'}{\expressionAlt'}
  }
}[$C$]$$\vspace*{-2mm}

\noindent
In the above rule, the condition $C$ requires that there are some $k,l,l' \in \N$ with $\varphi \models \expressionSize = l ~*~ \True$, $\varphi
\models \bBlock(\expression) = \bBlock(\expression') ~*~ \expression = \expression' + k ~*~ \True$,
$\size(\expressionAlt') = l'$, and
$k + l \le l'$.
In the rule, either $m = \top$ and $\chi \equiv \emp$, or $m \in
\Byte$ and $\chi \equiv \substitution{\expressionAlt'}{k}{k+l} = m^l$.
There is a symmetric rule \texttt{split-bl-pt-left} for the LHS.

We are finally getting to the \texttt{split-bl-bl-right} rule that matches two
fields that are both of non-constant sizes while splitting the RHS field if need
be.

\vspace*{-4mm}$$\inference[\texttt{split-bl-bl-right}]{
  \ABD{
    \varphi
  }{
    M
  }{
    \psi * \ptstobyte{\expression'}{m'}{\expression-\expression'} *
    \ptstobyte{\expression + \expressionSize}{m'}{\varZ} * K
  }}{
  \ABD{
    \varphi * \ptstobyte{\expression}{m}{\expressionSize}
  }{
    M
  }{
    \psi * \ptstobyte{\expression'}{m'}{\expressionSize'}
  }
}[$C$]$$\vspace*{-2mm}

\noindent In the rule, either $m' = \top$ or $m = m'$.
Further, $C$ is the condition that requires $\varphi \models \bBlock(\expression)
= \bBlock(\expression') ~*~ \expression' \le \expression ~*~ \expression +
\expressionSize \le \expression' + \expressionSize' ~*~ \True$ and $K \equiv
\varZ = \expressionSize' - (\expression - \expression') - \expressionSize$.
As before, there is also a symmetric rule \texttt{split-bl-bl-left} for
splitting on the LHS.

Next, we present a rule that allows one to match a points-to predicate on the
LHS against a singly-linked list segment on the RHS.
In fact, the rule does not directly perform the matching, but it facilitates it
by \emph{materialising} the first cell out of the list segment.
The matching itself (possibly combined with splitting) is then performed by the
above rules.
We expect that the cells of the list segment are described using a formula of
the form $\segment(\varX,\varY) \equiv \exists \varU_1,\ldots,\varU_k.
\lambda(\varX,\varY,\varU_1,\ldots,\varU_k)$.

\vspace*{-4mm}$$\inference[\texttt{slseg-pt-ls-right}]{
  \ABD{
    \hspace*{-2mm} \varphi * \ptsto{\expression}{\expressionAlt}
  }{
    M
  }{
    \psi *
    \lambda[\expression'/\varX,\varZ/\varY,\varZ_1/\varU_1,\ldots,\varZ_k/\varU_k] *
    \sll{\segment(\varX,\varY)}(\varZ, \expressionAlt') \hspace*{-2mm}
  }}{
  \ABD{
    \varphi * \ptsto{\expression}{\expressionAlt}
  }{
    M
  }{
    \psi * \sll{\segment(\varX,\varY)}(\expression', \expressionAlt')
  }}[$C$]
$$\vspace*{-2mm}

\noindent In the rule, $C$ is the condition that $\varphi \models
\bBlock(\expression) = \bBlock(\expression') ~*~ \True$ and
$\varZ,\varU_1,\ldots,\varU_k$ are some fresh variables.

We next present a version of the above rule for the case of a list segment on
the LHS.
Note that, in this case, we must require the list segment be non-empty.
In the rule, $C$ is the condition that $\varphi \models \bBlock(\expression) =
\bBlock(\expression') ~*~ \expression \neq \expressionAlt ~*~ \True$ and
$\varZ,\varU_1,\ldots,\varU_k$ are some fresh variables.

\vspace*{-4mm}$$\inference[\texttt{slseg-pt-ls-left}]{
  \ABD{
    \hspace*{-2mm} \varphi *
    \lambda[\expression/\varX,\varZ/\varY,\varZ_1/\varU_1,\ldots,\varZ_k/\varU_k] *
    \sll{\segment(\varX,\varY)}(\varZ, \expressionAlt)
  }{
    M
  }{
    \psi * \ptsto{\expression'}{\expressionAlt'}
  }}{
  \ABD{
    \varphi * \sll{\segment(\varX,\varY)}(\expression, \expressionAlt)
  }{
    M
  }{
    \psi * \ptsto{\expression'}{\expressionAlt'} \hspace*{-2mm}
  }
}[$C$]
$$\vspace*{-2mm}

The following rule allows one to remove from the LHS a list segment that forms
an initial part of a list segment that appears on the RHS.
The condition $C$ requires that $\varphi \models \expression = \expression' ~*~
\True$ and that $\segment(\varX,\varY) \models
\segment'(\varX,\varY)$\footnote{We note that this kind of entailment query
cannot be discharged as we said before for the case when the RHS of the
entailment is a pure formula (intuitively, one would need some negation over
SL). However, as we will show in
Section~\ref{sec:Entailment}, such queries can be discharged by a slight
modification of the bi-abduction procedure presented in this section.}.

\vspace*{-4mm}$$\inference[\texttt{slseg-ls-ls}]{
  \ABD{
    \varphi
  }{
    M
  }{
    \psi * \sll{\segment'(\varX,\varY)}(\expressionAlt, \expressionAlt')
  }}{
  \ABD{
    \varphi * \sll{\segment(\varX,\varY)}(\expression, \expressionAlt)
  }{
    M
  }{
    \psi *  \sll{\segment'(\varX,\varY)}(\expression', \expressionAlt')
  }
}[$C$]
$$\vspace*{0mm}
The further rule allows one to remove a possibly empty list segment from the
RHS.
A corresponding rule for list segments of the LHS is only needed for entailment checking (see Section~\ref{sec:Entailment}).

\vspace*{-2mm}$$\inference[\texttt{slseg-remove-right}]{
  \ABD{
    \varphi
  }{
    M
  }{
    \psi
  }}{
  \ABD{
    \varphi
  }{
    M
  }{
    \psi * \sll{\segment(\varX,\varY)}(\expression, \expressionAlt)
  }
}[$\varphi \models \expression = \expressionAlt ~*~ \True$]
$$\vspace*{-2mm}

We have similar rules for \emph{doubly-linked lists} as the ones stated above, which we omit here for ease of exposition (we point out that \texttt{dllseg-pt-ls-left} and \texttt{dllseg-pt-ls-right} come in two versions because a doubly-linked list can be unrolled from the left as well as from the right).

Next, we state a rule that allows one to \emph{finish} the abduction process.

\vspace*{-4mm}$$\inference[\texttt{learn-finish}]{
  }{
  \ABD{
    \varphi
  }{
    \psi
  }{
    \psi * \True
  }
}[$\varphi * \psi$ is satisfiable]
$$\vspace*{-2mm}

The side condition ``$\varphi * \psi$ is satisfiable'' is intended to ensure that the abduction solution $\psi$ does not lead to useless contracts:
a contract $\set{\varphi * \psi} f(\cdots) \set{\cdots}$ with $\varphi * \psi$ unsatisfiable does not have a configuration that satisfies its pre-condition!
Unfortunately, we only have an approximate procedure for checking the satisfiability of symbolic heaps (see Section~\ref{sec:SL2BV}).
However, contracts with an unsatisfiable pre-condition are still sound.
Hence, we employ the best-effort strategy of using our approximate procedure to prevent as many useless abduction solutions as possible in order to minimize the number of inferred contracts; however, inferring such contracts will not compromise the soundness of our analysis.

Finally, we state two rules of ``last resort'' that involve quite some guessing
and hence can mislead the abduction process and make it fail (or lead to its
exponential explosion when all possible variants of applying the rules are
attempted).
Therefore, they are to be tried only if the abduction process cannot terminate
without them.
Intuitively, they allow one to \emph{claim equal} fields whose \emph{equality}
is not known, but whose \emph{disequality} is not known either (moreover, in the
weaker case, one also checks that it can be shown that the fields lie within the
same memory block).

\vspace*{-2mm}$$\inference[\texttt{alias-weak}]{
  \ABD{
    \varphi * \chi(\expression) * \expression = \expression'
  }{
    M
  }{
    \psi * \chi'(\expression')
  }}{
  \ABD{
  \varphi * \chi(\expression)
  }{
    \expression = \expression' * M
  }{
    \psi * \chi'(\expression')
  }
}[$C_1$]
$$\vspace*{-2mm}

\vspace*{-2mm}$$\inference[\texttt{alias-strong}]{
  \ABD{
    \varphi * \chi(\expression) * \expression = \expression'
  }{
    M
  }{
    \psi * \chi'(\expression')
  }}{
  \ABD{
    \varphi * \chi(\expression)
  }{
    \expression = \expression' * M
  }{
    \psi * \chi'(\expression')
  }
}[$C_2$]
$$\vspace*{-2mm}

\noindent In the rules, $\chi(\varX)$ and $\chi'(\varX)$ are any predicates of
the form $\ptsto{\varX}{\_}$, $\sll{\_}(\varX,\_)$, $\dll{\_}(\varX,\_,\_,\_)$,
or $\dll{\_}(\_,\_,\varX,\_)$.
Further, $C_1$ is the condition that $\varphi \models \bBlock(\expression) =
\bBlock(\expression') ~*~ \True$ and that \emph{not} $\varphi \models \expression
\neq \expression'  ~*~ \True$.
On the other hand, $C_2$ requires that not $\varphi \models \expression \neq
\expression'  ~*~ \True$ only.

The \emph{alias-weak/strong} rules are used in the following
situations:\begin{itemize}

  \item There is \emph{no other applicable rule}.
  Instead of failing due to the impossibility of applying other rules, we try to
  introduce an alias (if possible, by the \emph{alias-weak} rule) and continue
  with the abduction using the \emph{match}, \emph{split}, or
  \emph{slseg}/\emph{dllseg} rules.

  \item We wish to infer \emph{multiple abduction solutions}.
  In such a case, whenever \emph{learn-finish} is applicable, we use it to
  derive one abduction solution, record it, revert \emph{learn-finish}, and then
  try to derive other solutions by applying an~\emph{alias} rule, followed by
  applying the other rules again.

\end{itemize}

We now state the correctness of the abduction procedure:

\begin{theorem} Let $M$ be any solution computed by the abduction rules, i.e.,
we have $\varphi * [M] \EMA \psi$. Then, $\varphi * M \models \psi$.
\end{theorem}

\begin{proof} We prove the property by induction on the number of rule applications.
We observe that the claim holds for the axiom, i.e., the rule
\texttt{learn-finish}).
We further note that, for all non-axiomatic rules of the shape

\vspace*{-4mm}$$\inference[rule-name]{
  \ABD{
    \varphi'
  }{
    M'
  }{
    \psi'
  }}{
  \ABD{
    \varphi
  }{
    M
  }{
    \psi
  }
}[$C$,]
$$\vspace*{-2mm}

\noindent we have that $\varphi' * M' \models \psi'$ implies that $\varphi * M
\models \psi$ (under the condition $C$).
Hence, the claim holds.\end{proof}

Moreover, we observe that the antiframe $M$ is guaranteed to be a quantifier-free symbolic heap in case the input $\varphi$ to the abduction procedure is a quantifier-free symbolic heap (the abduction rules maintain this shape of $\varphi$).

\subsection{Frame Inference} \label{sec:frame-inference}

We now explain how we solve the \emph{frame problem} $\varphi * M \models \psi * [?]$, computing a quantifier-free symbolic heap $\Frame$ such that $\varphi * M \models \psi * \Frame$.
Inspired by the approach of \cite{BiAbd11}, $\Frame$ is computed as a by-product of our rules for finding the abduction solution $M$, and no special frame inference procedure is needed.
%
%
Namely, when solving the abduction problem $\varphi * [?] \models \psi * \True$ by means of the rules presented in Section~\ref{sec:abduction-rules}, one eventually needs to apply the \texttt{learn-finish} rule:

\vspace*{-4mm}$$\inference[\texttt{learn-finish}]{
  }{
  \ABD{
    \varphi
  }{
    \psi
  }{
    \psi * \True
  }
}[$\varphi * \psi$ is satisfiable]
$$\vspace*{-2mm}

\noindent
In this rule, the unused part $\varphi$ of the LHS of the abduction problem is consumed by $\True$.
Instead, we return the $\varphi$ from the application of the \texttt{learn-finish} rule as the solution $\Frame$ to the frame problem, which is guaranteed to be sound as stated by the following theorem:

\begin{theorem} Let $M$ and $\Frame$ be the solution returned by the combined frame and abduction procedure.
Then, $\varphi * M \models \psi * \Frame$.
\end{theorem}

\begin{proof}
We prove the property by the number of rule applications.
Indeed, we have the claim holds for
\texttt{learn-finish}), i.e., we have $\varphi * M \models \psi * \Frame$.
We further note that, for all non-axiomatic rules we can instantiate the right-hand sides to obtain rules of the shape

\vspace*{-4mm}$$\inference[rule-name]{
  \ABD{
    \varphi'
  }{
    M'
  }{
    \psi'* \Frame
  }}{
  \ABD{
    \varphi
  }{
    M
  }{
    \psi * \Frame
  }
}[$C$,]
$$\vspace*{-2mm}

\noindent where we have that $\varphi' * M' \models \psi' * \Frame$ implies that $\varphi * M \models \psi * \Frame$ (under the condition $C$).
Hence, the claim holds.
\end{proof}
Moreover, we observe that the frame $\Frame$ is guaranteed to be a quantifier-free symbolic heap in case the input $\varphi$ to the abduction procedure is a quantifier-free symbolic heap (the abduction rules maintain this shape of $\varphi$).

\subsection{An Example of Solving A Bi-Abduction Problem}
\label{sec:biabd-example}

We recall the setting of Example~\ref{ex:abdInSeq}:
We consider the call $g(x)$ of a function with contract $\set{C} g(y) \set{D}$ during the analysis of some function $f(x)$, where the so-far derived precondition $\pre$, the current symbolic state $\post$ as well as the pre- and post-conditions $C$ and $D$ look as follows:

\begin{itemize}
  \item $\pre \equiv \ptsto{\varXX}{\varA} * \ptsto{\varXX+8}{\varB} *  \varX=\varXX$;

  \item $\post \equiv \exists \varZ. \ \ptsto{\varXX}{\varA} * \ptsto{\varXX+8}{\varZ} *  \varX=\varXX$;

  \item $C \equiv \ptsto{\varYY}{\varU} * \ptsto{\varYY+8}{\varW} * \ptsto{\varU}{\varV} * \varY=\varYY$, and

  \item $D \equiv \exists\varC.\ D_\postFormulaSubscript * D_\postEQSubscript$ where $D_\postFormulaSubscript \equiv
    \ptsto{\varYY}{\varU} * \ptsto{\varYY+8}{\varW} * \ptsto{\varU}{\varV} * \ptsto{\varV}{\varC}$
  and $D_\postEQSubscript \equiv \varY=\varV$.

\end{itemize}

\noindent The formula $C'$ wrt which we will be solving the bi-abduction problem will look as follows:
$$C'\equiv (C[\varY/\varX])[\varXX/\varX] \equiv \ptsto{\varYY}{\varU} * \ptsto{\varYY+8}{\varW} * \ptsto{\varU}{\varV} * \varXX=\varYY.$$
Hence, we need to solve the abduction problem $$\ptsto{\varXX}{\varA} * \ptsto{\varXX+8}{\varZ} *
[?] \models \ptsto{\varYY}{\varU} * \ptsto{\varYY+8}{\varW} * \ptsto{\varU}{\varV} * \varXX=\varYY* \True.$$
Its solution by our abduction rules looks as follows:

{\footnotesize
$$
  \inference[\hspace*{-10mm}\texttt{learn-pure}]{
    \inference[\texttt{match}]{
        \inference[\texttt{match}]{
              \inference[\texttt{learn-finish}]{
              }{
                \ABD{
                  \varXX=\varYY * \varA=\varU * \varZ = \varW
                }{
                  \ptsto{\varU}{\varV}
                }{
                  \ptsto{\varU}{\varV} * \True
                }
              }
        }{
            \ABD{
                \ptsto{\varXX+8}{\varZ} * \varXX=\varYY * \varA=\varU
            }{
                \varZ = \varW * \ptsto{\varU}{\varV}
            }{
                \ptsto{\varYY+8}{\varW} * \ptsto{\varU}{\varV} * \True
            }
        }
    }{
        \ABD{
          \ptsto{\varXX}{\varA} * \ptsto{\varXX+8}{\varZ} * \varXX=\varYY
        }{
          \varA=\varU * \varZ = \varW * \ptsto{\varU}{\varV}
        }{
            \ptsto{\varYY}{\varU} * \ptsto{\varYY+8}{\varW} * \ptsto{\varU}{\varV} * \True
        }
    }
  }{
    \ABD{
      \ptsto{\varXX}{\varA} * \ptsto{\varXX+8}{\varZ}
    }{
      \varXX=\varYY * \varA=\varU * \varZ = \varW * \ptsto{\varU}{\varV}
    }{
        \ptsto{\varYY}{\varU} * \ptsto{\varYY+8}{\varW} * \ptsto{\varU}{\varV} * \varXX=\varYY * \True  }
  }
$$
}

\noindent The so-far missing pre-condition is thus $M' \equiv \varXX=\varYY * \varA=\varU * \varZ = \varW * \ptsto{\varU}{\varV}$.
After applying $\elim$ on $M'$ we obtain $M \equiv \ptsto{\varA}{\varV}$.
Further, from the LHS of the \texttt{learn-finish} rule, we get
$\Frame \equiv  \varXX=\varYY * \varA=\varU * \varZ = \varW$.

\section{Implementation and Experimental Evaluation} \label{sec:implementation}

In this section, we first discuss how we deal with entailment and satisfiability
checking over our SL, for which we employ a reduction to SMT solving over the
bitvector logic.
Then, we present our prototype tool implementing the proposed approach, followed
by experiments with it.

\subsection{Entailment Checking} \label{sec:Entailment}

Our approach requires to answer entailment queries $\phi_1 \models \phi_2$ at several places.
While it might be possible to develop a general (sound and complete) entailment procedure, e.g., extending \cite{JensFlorian:BeyondSymHeap:20},
we decided to use an approximation to entailment queries for performance reasons.
This is justified as follows:
Our approach only relies on the correctness of positive answers for an entailment query $\phi_1 \models \phi_2$, while false negatives will not lead to unsound answers of our overall procedure.
We will use two different approximations to  $\phi_1 \models \phi_2$ depending on whether the right hand-side $\phi_2$ is a pure formula:

\begin{description}
  \item[$\phi_1 \models \expression \bowtie \expressionAlt * \True$ for a quantifier-free symbolic heap $\phi_1$ and pure formula $\expression \bowtie \expressionAlt$]:
Entailment queries of this kind are needed at several places throughout our abduction procedure.
We will employ the idea that $\phi_1 \models \phi_2$ holds iff $\phi_1 \wedge \neg \phi_2$ is unsatisfiable.
Note, however, that our logic is not closed under negation $\neg$ (nor under standard conjunction $\wedge$) and that such an approach is not possible for general entailment queries;
however, we can negate pure formulae $\expression \bowtie \expressionAlt$, i.e.,
we will check the satisfiability of the formula
$\phi_1 * \neg (\expression \bowtie \expressionAlt)$.
We will use an approximated answer by an (incomplete) translation of separation logic formulae to SMT formula over bitvectors.
We only require that $\mathit{UNSAT}_\mathit{BV} \Rightarrow \mathit{UNSAT}_\mathit{SL}$ for maintaining the correctness of positive answers to our entailment queries.
We describe our translation to SMT formulae in Section~\ref{sec:SL2BV}.

\item[$\exists A. \ \phi_1 \models \exists B. \phi_2$ for some quantifier-free symbolic heaps $\phi_1$ and $\phi_2$]:
We can assume that $A \cap B = \emptyset$ because quantified variables can always be renamed.
The entailment check then proceeds similar to the abduction procedure described in Section \ref{sec:abduction-rules}.
That is, we solve an abduction problem in the first step and in the second step check whether the inferred abduction solution is in fact equivalent to $\emp$.

\begin{enumerate}
  \item We solve the abduction problem $\phi_1 * [?] \models \phi_2$ in order to obtain a quantifier-free symbolic heap $M$ with $\phi_1 * M \models \phi_2$.
      If there is no such $M$ the entailment check fails during this first step.
      Note that in contrast to our bi-abduction procedure we do not add $*\True$ to $\phi_2$ because we require an exact matching of the spatial part of $\phi_1$ and $\phi_2$.
      Hence, the $\texttt{learn-finish}$ is not adequate for finishing such abduction queries.
      Instead, we use a new \emph{emp-finish} rule of the following form:

      \vspace*{-4mm}$$\inference[\texttt{emp-finish}]{
        }{
        \ABD{
            \pi
        }{
            \emp
        }{
            \emp
        }
      }[$\pi$ pure]
      $$\vspace*{-2mm}

    For proving entailments, we also need a rule for removing remaining list segments of the LHS
    (unmatched list segments on the LHS of an abduction query are covered by $~*~ \True$ on the RHS; in the case of entailment however, we need to explicitly match everything on the LHS):

    \vspace*{-2mm}$$\inference[\texttt{slseg-remove}]{
      \ABD{
            \varphi
        }{
            M
        }{
            \psi
      }}{
      \ABD{
          \varphi * \sll{\segment(\varX,\varY)}(\expression, \expressionAlt)
        }{
            M
        }{
            \psi
        }
      }[$\varphi \models \expression = \expressionAlt ~*~ \True$]
    $$\vspace*{-2mm}

    (Further, we disable the $\texttt{alias-weak}$ and $\texttt{alias-strong}$ rules, because these rules will only introduce equalities that cannot be removed in the second step of the entailment check.)

  \item Let $M$ be a formula inferred in step 1, i.e., we have $\phi_1 * M \models \phi_2$.
      If $M$ is not a pure formula, we fail.
      Otherwise we check whether $\phi_1 \models M$ using the SMT reduction from the previous case.
      (We remark that compared to bi-abduction procedure we require $M \equiv \emp$.
      The soundness of the entailment check then follows similar to Lemma~\ref{lem:bi-abduction-correctness}.)
\end{enumerate}
\end{description}

\subsection{Reduction of (Un-)Satisfiability of SL to SMT} \label{sec:SL2BV}

In this section, we describe a reduction of (un-)satisfiability of
quantifier-free symbolic heaps to SMT.
In particular, we use the SMT theories of bitvectors and uninterpreted functions.
Our reduction will only be approximate, i.e., we will ensure that $\mathit{UNSAT}_\mathit{BV} \Rightarrow \mathit{UNSAT}_\mathit{SL}$ but we do not have such an implication for $\mathit{SAT}$.
The main use of this reduction is checking the entailment of a pure formula as described in the previous subsection (where we only require the correctness of unsatisfiability).
We are checking also the satisfiability of a separation-logic formula in the  $\texttt{learn-finish}$ rule;
as explained, detecting unsatisfiability allows us to rule out useless contracts (because of an unsatisfiable pre-condition) but returning useless contracts (in case we do not detect the unsatisfiability of a pre-condition) does not compromise the soundness of our approach.
Like for the entailment check, the main motivation for having an approximate (un-)satisfiability check is improved performance as this check needs to be performed many times during an analysis.

We now develop the reduction to SMT over bitvectors and uninterpreted functions:
All values (including memory addresses) and variables are encoded as bitvectors (BV) of appropriate size.
We recall that we work with an intermediate language where the size is the only available type information on values and variables.
That is, we rely on that all type information  from the high-level language (in
our case C)---such as whether a variable is signed or unsigned---has been encoded into operators of our intermediate, e.g., we assume different plus operations for bitvectors that should be interpreted as signed resp. unsigned integers.
Besides bitvectors, our translation relies on two uninterpreted functions $\bBlock,\eBlock: \Byte^N \rightarrow \Byte^N$ representing the $\bBlock$ (resp. $\eBlock$) operators from our logic.
We now give the translation of a symbolic heap $\psi = \phi_1 * \cdots * \phi_l$ to SMT, where each $\phi_i$ is one of the formulas
$\expression_1 \bowtie \expression_2$,
$\sll{\segment(\varX,\varY)}(\expression_1,\expression_2)$,
$\dll{\segment(\varX,\varY,\varZ)}(\expression_1,\expression_2,\expression_1',\expression_2')$,
or $\ptsto{\expression}{\varUps}$ (where $\varUps$ represents either $\expressionAlt$, $\ptstobyterhs{\val}{\expressionAlt}$, or
$\ptstobyterhs{\top}{\expressionAlt}$;
we recall that we have defined $\size(\varUps)=\expressionAlt$ for the two later cases);
the translation of $\psi$ is a conjunction of the following constraints:
\begin{itemize}

  \item an axiom requiring that each allocated location lies within its block
  $$\forall \loc. \ \bBlock(\loc) = 0 \vee  \bBlock(\loc) \le \loc <
  \eBlock(\loc);$$

  \item an axiom requiring that blocks never overlap $$\forall \loc, \loc'. \ (0
  < \bBlock(\loc) < \eBlock(\loc')\le \eBlock(\loc) \vee 0 < \bBlock(\loc') <
  \eBlock(\loc)\le \eBlock(\loc')) \rightarrow  \bBlock(\loc) = \bBlock(\loc')
  \wedge \eBlock(\loc) = \eBlock(\loc');$$

  \item the formula $\phi_i$ for each pure formula $\phi_i \equiv \expression_1
  \bowtie \expression_2$;

  \item for each points-to predicate $\phi_i \equiv
  \ptsto{\expression}{\varUps}$, the formula $$\bBlock(\expression) > 0 \wedge
  \expression \le \expression + \size(\varUps) \wedge \expression +
  \size(\varUps) \le \eBlock(\expression),$$ requiring that the source is
  allocated and that there is sufficient space in the corresponding block to
  store the pointer (note that the second conjunct is necessary because we are
  using bit-vector semantics and have to take possible overflows into account);

  \item for each list-segment predicate $\phi_i \equiv
  \sll{\segment(\varX,\varY)}(\expression_1,\expression_2)$ resp. $\phi_i \equiv
  \dll{\segment(\varX,\varY,\varZ)}(\expression_1,\expression_2,\expression_1',\expression_2')$,
  the formula $$\expression_1 = \expression_2 \vee (\expression_1 \neq
  \expression_2 \wedge \bBlock(\expression_1) > 0)$$ resp.  $$(\expression_1 =
  \expression_2' \wedge \expression_2 = \expression_1') \vee (\expression_1 \neq
  \expression_2' \wedge \expression_1 \neq \expression_2' \wedge
  \bBlock(\expression_1) > 0),$$ requiring that start- resp. end-address(es) of
  the list are equal or at least the first list node is allocated;

  \item for each pair of points-to predicates  $\phi_i \equiv
  \ptsto{\expression_1}{\varUps_1}$ and $\phi_j \equiv \ptsto{\expression_2}{\varUps_2}$
  with $i\neq j$, the formula $$\bBlock(\expression_1) \neq  \bBlock(\expression_2) \vee
  (\bBlock(\expression_1) = \bBlock(\expression_2) \wedge (\expression_1+\size(\varUps_1) \le \expression_2 \vee
  \expression_2+\size(\varUps_2) \le \expression_1)),$$ encoding the semantics of the separating
  conjunction, requiring that two pointers are either in different blocks or
  they belong to the same blocks but their associated memory does not overlap;

  \item for each pair of predicates $\phi_i \equiv \ptsto{\expression}{\varUps}$
  and $\phi_j \equiv \sll{\segment(\varX,\varY)}(\expression_1,\expression_2)$
  with $i\neq j$, the formula $$\bBlock(\expression)\neq \bBlock(\expression_1) \vee \expression_1
  = \expression_2,$$ requiring that either the list is empty or the list head
  and the pointer belong to different blocks (which follows from the requirement
  that $\segment$ is block-closed);

    for $\phi_i \equiv \ptsto{\expression}{\varUps}$ and $\phi_j \equiv
    \dll{\segment(\varX,\varY,\varZ)}(\expression_1,\expression_2,\expression_1',\expression_2')$,
    we derive a similar formula which we do not give here explicitly;

  \item for each pair of predicates $\phi_i \equiv
  \sll{\segment(\varX,\varY)}(\expression_1,\expression_2)$ and $\phi_j \equiv
  \sll{\segment'(\varX,\varY)}(\expression_1',\expression_2')$ with $i\neq j$,
  the formula $$\bBlock(\expression_1)\neq \bBlock(\expression_1') \vee \expression_1 = \expression_2
  \vee \expression_1 = \expression_2,$$ requiring that either at least one list
  is empty or the list heads belong to different blocks (which follows from the
  requirement that $\segment$ and $\segment'$ are block-closed);

    if either $\phi_i$ or $\phi_j$ or both are predicate of form
    $\dll{\segment(\varX,\varY,\varZ)}(\expression_1,\expression_2,\expression_1',\expression_2')$,
    we derive similar formulae, which we do not give here explicitly.

\end{itemize}

It is easy to verify the correctness of the reduction:
Every model $(\stack,\blocks,\mem)$ of $\psi = \phi_1 * \cdots * \phi_l$ gives rise to an interpretation of $\bBlock$, $\eBlock$ and the variables $\psi$ that satisfies the above constraints.
Hence, we get that $\mathit{UNSAT}_\mathit{BV} \Rightarrow \mathit{UNSAT}_\mathit{SL}$.
We further remark that the generated constraints are clearly incomplete (e.g., we are ignoring the segment predicate $\segment$ in our translation).
Hence, we cannot expect that $\mathit{SAT}_\mathit{BV} \Rightarrow \mathit{SAT}_\mathit{SL}$ holds in general.
We finally add that in our implementation we do not give the forall-quantified axioms to the SMT solver (the first two bullet points in the above list) but rather instantiations of the axioms for all expressions appearing on the left-hand side of points-to predicates and as head nodes of list-segment predicates;
this improves the performance of the solver but adds another source of incompleteness.

\subsection{Prototype Implementation}\label{sec:prototype}

%

We have implemented the proposed techniques in a prototype tool called Broom.
%
%
Its source code is publicly available\footnote{\url{https://pajda.fit.vutbr.cz/rogalew/broom}} under GNU GPLv3.
%
%
The tool itself is implemented in OCaml.
The SMT queries discussed above are answered using the Z3 solver \cite{Z3:08}.
The front-end of Broom is based on Code Listener \cite{CodeListener:11}, a
framework providing access to the intermediate code of a~compiler (as, e.g.,
\texttt{gcc}).

We note that, in the implementation of Broom, we relaxed the requirement put on
the \emph{ptrplus} operation of our minilanguage
(Sec.~\ref{sec:language-plus-semantics}), which requires that the pointer
resulting from the expression $y+z$ stays within the allocated block---i.e.,
$\bBlock(S_B(y))\leq S_B(y)+S_B(z) \leq \eBlock(S_B(y))$.
According to the C standard, the relaxation of this condition leads to undefined
behaviour, but it is often used in low-level system code as, e.g., in the Linux
list implementation.
In our implementation, we allow pointers to have values outside of the allocated
blocks, but we explicitly track their \emph{provenance} (i.e., the basis wrt
which they are defined) using the $\bBlock$ predicate.
%

Broom comes with a number of parameters that can be set for the analysis, with
the most important being the following ones:
\begin{itemize}

  \item \emph{Solver timeouts:} Timeouts of the underlying solver can be set
  separately for symbolic execution, widening, and formula simplification.
  Using a timeout, one can balance between speed and precision.
  With a lower timeout, the analysis is faster, but some functions need not be
  fully analysed due to an abduction or widening failure.
  The default timeouts used in our later presented experiments are 2000ms for
  the symbolic execution, 200ms for widening, and 100ms for formula
  simplification. \footnote{Note that the tool sometimes
  produces an SMT query that may kill the solver’s heuristics.
  The solver then does not provide any result for a long time (tens of minutes) or
  even does not terminate. The user can disable the timeouts for the price of nontermination of
  the analysis of some examples. According to our experience, these \emph{bad} queries are usually related to
  infeasible branches of the computation. Properly set timeouts can bypass this issue for the
  price of failing on some abduction queries (which, however, does not compromise the
  soundness).}

  \item \emph{Number of loop unfoldings:} A limit on the number of loop
  unfoldings is used to stop the loop analysis when a fixpoint is not computed
  within a~given number of loop iterations.
  Then, either no contract or partial contracts are returned.
  The default value used in our experiments is 5.

  \item \emph{Abduction strategy:} The abduction strategy can be set as follows:
  In the standard configuration, it  follows the order of rules presented in
  Sec. \ref{sec:abduction-rules}.
  The tool also supports an alternative strategy where the
  \texttt{alias-weak/strong} rules are used to derive multiple abduction
  solutions as discussed in Sect.~\ref{sec:abduction-rules}.
  This may lead to an exponential blow-up in the number of contracts for
  particular functions (a lot of them useless) together with a blowup of the
  running time.
  On the other hand, this strategy allows us to fully verify some of our most
  complicated code fragments (namely, the intrusive lists discussed below).
  As a part of our future research, we would like to study some
  heuristically-driven application of this strategy that would not explore so
  many useless contracts.

\end{itemize}

The errors are detected as follows: When a precondition from a contract for \emph{free},
\emph{load}, or \emph{store} can not be satisfied, an error message is provided to the
user. Memory leaks are reported when a variable on the
left of some points-to or list-segment predicate is found to be unreachable (via points-to
assertions, including list-segments) from the global variables, the function parameters,
and the return variable.
Note that a contract computed for a used defined function underapproximates its behaviour
and, therefore, we can not easily flag as an error the situation, where a contract for
an user defined function is not applicable.

Finally, we would like to stress that Broom is now in a stage of a very early
prototype, intended mainly to illustrate the theoretical potential of our
technique, with huge space for performance optimizations.
As a primary source of possible optimisations, we see the way how Broom
interacts with the SMT solver (the cost of SMT queries represents a very
significant part of the cost of the entire analysis).
One way that we see as highly promising for optimisations in this direction is
to use static pre-evaluation of some SMT queries---if one can statically
evaluate a query, an expensive solver call can be avoided.
This can significantly limit the number of SMT queries and improve the running
time.
We have already partially implemented some static pre-evaluation for the $\phi
\models \expression = \expression' * \True$ queries within
\texttt{match}/\texttt{split} abduction rules, which alone reduced the running
time by 25~\% at some examples.
Further optimization possibilities then lie, e.g., in incremental solving,
caching solver results, and/or introducing heuristics to decrease the amount of
nondeterminism in the abduction rules.
As for the last mentioned possibility, especially in the case of the
\texttt{match}/\texttt{split} rules there can be several candidate predicates
$\ptsto{\expression}{\expressionAlt}$ on the LHS and several candidate
predicates $\ptsto{\expression'}{\expressionAlt'}$ on the RHS, which one needs
to consider, and it would be very helpful to have some guidance in this process.

\subsection{Experiments}\label{sec:Experiments}

We evaluate our tool Broom on a set of experiments in which we analyse various
fragments of list manipulating code.
Since Broom is in a highly prototypical stage, we do not venture into analysing
large code bases.
Instead, we concentrate on shorter but complex code highlighting what our
approach implemented in the tool can handle (and what other tools do typically
not manage).

The considered code was pre-processed in the following ways:
(1) All appearances of the so-far unsupported constructions
\verb|&var| and \verb|var.next|
were replaced by
\verb|p_var| and \verb|p_var->next|,
respectively, where
\verb|p_var = alloca(sizeof(*p_var))|.
(2) We replaced \texttt{for} loops with integer bounds by non-deterministic
\texttt{while} loops because our abstraction and entailment are currently very
limited when working with integers.
Both of the above is planned to be resolved within our future work.
Further, we analysed all the code assuming that heap allocation always succeeds.

The experiments were run on a machine with an Intel i7-4770 processor with 32
GiB of memory.
The current implementation of Broom uses a single core only.
We compare our results with those of Infer
v1.1.0\footnote{\url{https://github.com/facebook/infer/releases/tag/v1.1.0}} and
Gililan (PLDI'20
version)\footnote{\url{https://github.com/GillianPlatform/Gillian/releases/tag/PLDI20}},
which are the only tools we are aware of that can analyse at least some of the
code we are interested in.
We note that Infer was running with debug information enabled (using the command
\texttt{infer run -\hspace{-2mm} -debug}) as we wanted to manually check the
obtained contracts.
The debug option may increase the running time of Infer, but, as one can see in
Table \ref{table:dll-standard-strategy}, the running times are not an issue for
Infer.

\newcommand{\NE}{OK} 
\newcommand{\IE}{IE} 
\newcommand{\PE}{PE} 
\newcommand{\DF}{DF} 
\newcommand{\ML}{ML} 
\newcommand{\bm}[1]{\small{\fontfamily{lmr}\selectfont #1}} 

Table~\ref{table:dll-standard-strategy} presents a comparison of the results
obtained using Broom, Infer, and Gillian on our collection of list-manipulating
code fragments.\footnote{All the code is available together with our tool.}
To get the results, Broom was used with its standard abduction strategy where
the \texttt{alias-weak}/\texttt{alias-strong} rules are used only if no other
rule is applicable.
For each of the cases, the table gives first the total number of functions that
the benchmark consists of.
Next to it, separately for Broom, Infer, and Gillian, we give the time the tools
took for the analysis.
Further, we list the number of functions for which the respective tool produced
a non-trivial contract. There are up to three numbers in the form $a/b/c$ ($b$
or $c$ can be omitted), representing the number of functions for which the
respective tool computes (a) complete contracts, (b) sound but only partial
contracts, and (c) error contracts---i.e., preconditions under which a given
function is bound to fail, which are provided by Gillian only.
Finally, we also provide a remark whether the tool reported some error (or
whether it itself hit some internal error).
The expected and really obtained analysis results are encoded as follows
(including internal errors of an~analyser): \NE = \emph{no error
found}\footnote{We note that, as far as our experience reaches, Gillian produces
its error contracts whenever there is a risk of a null-pointer dereference. In
many cases, e.g., in the Linux list library, the error summaries provide a
correct result, which, however, does not take into account the fact that the
library is designed such that the appropriate functions are never called with a
null argument. At the same time, Gillian may miss real, higher-level errors
present in the code, which were those we expected to be reported. In such cases,
we say in the table in the column for obtained results that the (expected) error
was not found.}, \DF = \emph{double free}, \ML = \emph{memory leak}, \IE =
\emph{internal error}, \PE = \emph{internal parsing error}.


\begin{table}

\caption{Experiments with the standard abduction strategy of Broom and a
comparison with results obtained from Infer and Gillian.
%
%
}

\label{table:dll-standard-strategy}
\begin{tabular} {l | l | c | c | c | l | c | c | l | c | c | l }
\hline
Name & Exp. & Fncs & \multicolumn{3}{c|}{Broom} & \multicolumn{3}{c|}{Infer} & \multicolumn{3}{c}{Gillian} \\
 & result & total & T & Fncs & Res & T & Fncs & Res & T & Fncs & Res \\
 & & & [s] & contr & & [s] & contr & & [s] & contr & \\
\hline
\bm{circ-DLL} & \ML & 4 & 6 & 4 & \ML & 0.5 & 1/1 & \IE & 1.2 & 1/0/2 & \NE\\
\bm{circ-DLL-err} & \DF & 4 & 6 & 4 & \DF & 0.5 & 1/1 & \IE & 1.2 & 1/0/2 & \NE\\
\bm{circ-DLL-embedded} & \NE & 4 & 9 & 4 & \NE & 0.5 & 1/2 & \IE & 0.6 & 0 & \PE\\
\bm{Linux-list-1} & \ML & 11 & 56 & 10 & \ML & 1.5 & 2/3 & \NE & 0.6 & 0 & \PE\\
\bm{Linux-list-2} & \NE & 11 & 42 & 11 & \NE & 0.7 & 1/6 & \IE & 0.6 & 0 & \PE\\
\bm{Linux-list-2-err} & \ML & 11 & 28 & 11 & \ML & 0.6 & 1/6 & \IE & 0.6 & 0 & \PE \\
\bm{Linux-list-all} & \NE & 23 & 267 & 21/2 & \NE & 1.0 & 7/15 &
  \IE & 44 & 8/0/9 & \NE\\
\bm{intrusive-list} & \NE & 15 & 99 & 10/5 & \NE & 0.7 & 4/3 & \NE  & 0.6 & 0 & \PE\\
\bm{intrusive-list-min} & \NE & 9 & 45 & 6/2 & \NE & 0.7 & 1/3 &
  \IE & 0.6 & 0 & \PE\\
\bm{intrusive-list-smoke} & \NE & 20 & 133 & 10/5 & \NE & 0.9 & 4/3 & \NE & 0.6 & 0 & \PE\\
\hline
\end{tabular}
\end{table}

We now discuss the individual cases in more detail---when doing so, we
concentrate on comparing the results of Broom with those of Infer that can get
somewhat closer to the results of Broom:\begin{itemize}

  \item \texttt{circ-DLL}: This example deals with a simple implementation of
  \emph{circular doubly-linked lists} (whose part is, in fact, used as the
  running example in Fig.~\ref{fig:example1}).
  The code includes functions for inserting the first element, inserting another
  element after an existing one, and for removing elements.
  Apart from that there is a higher-level function that inserts the first
  element, the second element, and them removes one of them.\footnote{This
  function can be viewed as an analysis harness while we were stressing that our
  analysis does not need such a harness. Here, we would like to stress that this
  indeed holds---none of the considered tools needs (nor in any way uses) the
  top-level function to be able to analyze the other functions. We use the
  harness as a model of any higher-level code using the list. Moreover, it
  allows us to show that the contracts that got generated for the particular
  functions are not complete enough, which shows up in the inability of the
  appropriate tool to analyse the higher-level functions.}
  The code contains no pointer arithmetic nor any other advanced features.
  It is intended to show that even in such a case our abduction rules restated
  wrt \cite{BiAbd09, BiAbd11} can bring some advantage.
  Namely, this is a consequence of that our use of the per-field separation
  allows us to cover more shapes of the data structures within a single
  contract.
  Indeed, as discussed already in Sect.~\ref{sec:illustration}, it produces a
  single contract for insertion into a cyclic list with one element and with
  more elements (intuitively, the per-field separation allows us to merge or keep
  separate the elements before/after the one being inserted: merging is
  necessary in a one-element list where the element before and after are in fact
  the same element).
  Infer cannot use the same reasoning and since it primarily favours
  scalability, it will come with a contract for inserting into lists with at
  least two elements.
  Consequently, it then fails to analyse the top level function.
  As for the memory leak reported by Broom, it is a real error caused by that
  one of the introduced elements is not deleted.

  \item \texttt{circ-DLL-err} is a variation on \texttt{circ-DLL} into which we
  introduced a double-free error.

  \item \texttt{circ-DLL-embedded} is another variation on \texttt{circ-DLL} in
  which the list implementation from \texttt{circ-DLL} is used as a basis of a
  simple intrusive list in which the list structure with the linking fields from
  \texttt{circ-DLL} is nested into a larger data structure.
  Effectively, this leads to a list consisting of these larger structures where,
  however, the list manipulating functions have no information about that the
  linking structure is a part of some bigger block.
  Contracts generated for the list-manipulating functions are thus to be applied
  on larger structures, leading to an application of our block splitting rules.

\eject 

  \item \texttt{Linux-list-1} is our first experiment with intrusive lists in
  the form they are used in the Linux kernel (for some more impression about
  Linux lists, see Fig.~\ref{fig:linux-list}).
  This particular code comes in particular from the benchmark suite of the
  Predator analyser \cite{PredatorSAS13}.\footnote{We note here that Predator
  can analyse the code, but---unlike Broom, Infer, or Gillian---it entirely
  relies on that the code is closed, i.e., it comes with a main function and has
  no further inputs.}
  The code contains multiple different functions for initialisation of the
  lists, for inserting into it, and for traversing the lists.
  The top-level code that is present then creates a circular Linux list nested
  into another circular Linux list.
  As can be from Fig.~\ref{fig:linux-list},
  the code involves pointer
  arithmetic (even in a form not supported by the C standard), and the use of
  nested structures leads to an application of our block splitting rules.
  The only function that Broom fails to handle is the function for traversing
  the entire list---the reason is that our so-far quite simple implementation of
  list abstraction fails in this case, and the otherwise correct computation
  diverges (which we, however, believe to be solvable in the close
  future---indeed, we can get fully inspired by the abstraction used in
  Predator; the concrete list abstraction used is not specific for our
  approach).
  The memory leak reported is a real one---it comes from the top-level function
  that does not destroy the list.

\begin{figure}[t]
  \centering
    \begin{minipage}{.7\textwidth}
        \hspace{-1em}
        \includegraphics[width=\textwidth]{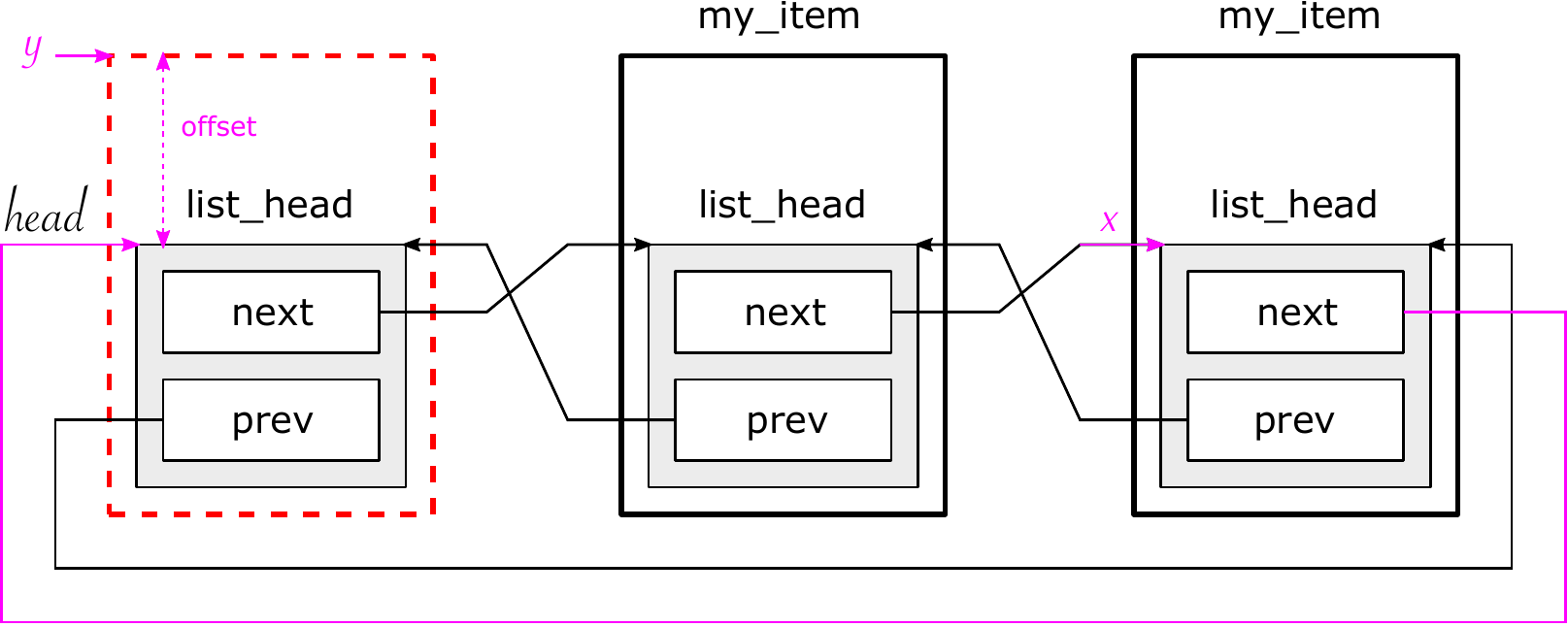}
    \end{minipage}%
    \begin{minipage}{.3\textwidth}
        \hspace{-1em}
        \centering
{\scriptsize
\begin{verbatim}
// moving to the next item
y = (struct my_item *)(
  (char *)x->next -
  offsetof(struct my_item,
           link)
  );
// testing, if it is head
if (head != &y->link) ...
\end{verbatim}}
    \end{minipage}

  \caption{An illustration of the Linux list data structure. The gray boxes
  represent the linking structure \texttt{list\_head} that is nested into a
  user-defined structure \texttt{my\_item}, whose instances the user needs to be
  linked into a list. List-manipulating functions know nothing about the
  user-defined structures: they work with the linking structures only. The user
  data are accessible through pointer arithmetic only. Note that the head node
  of the list does \emph{not} have the user-defined envelope. The code shown on
  the right illustrates how the list is traversed. Note also that when one
  passes from the last element pointed by \texttt{x} to its successor (hence
  back to the head), the involved pointer arithmetic causes that the pointer
  \texttt{y} will be pointing out of the allocated space, which is, however,
  correct since it will never be dereferenced (just used for further pointer
  arithmetic).}

  \label{fig:linux-list}
\end{figure}

  \item \texttt{Linux-list-2} is a variation on the above case. It contains
  functions for an initialisation of a Linux list, inserting elements at its
  tail, and for deleting the elements.
  The top-level function initializes the list, inserts several elements,
  traverses the elements one by one, and deletes them.
  Note that the sets of functions present in \texttt{Linux-list-1} and
  \texttt{Linux-list-2} have a non-empty intersection but are incomparable.

  \item \texttt{Linux-list-2-err} is a variation on \texttt{Linux-list-2} where
  one of the inserted elements is not deleted and hence a memory leak is caused.

  \item \texttt{Linux-list-all} contains the entire collection of functions
  defined for working with Linux lists without any top-level function.
  The collection includes functions for different kinds of insertion of
  elements, removal of elements, swapping of elements (both within a list and
  between lists), moving to the end or to another list, rotation, splicing, etc.
  We can see that Broom produced complete contracts for many more of the
  functions.
  The contracts from Infer often do not cover cases of lists of length
  0 or 1.
  In one of the remaining cases, Infer produced no result; and for the last one,
  it produced a partial result (that appears not to cover one of the branches of
  the function).

  \item \texttt{intrusive-list} is the intrusive list library\footnote{Described
  in the Patrick Wyatt's blog post “Avoiding game crashes related to linked
  lists”,
  \url{http://www.codeofhonor.com/blog/avoiding-game-crashes-related-to-linked-lists},
  on September 9th, 2012,
  and implemented in \url{https://github.com/robbiev/coh-linkedlist}.}.
  See Fig.~\ref{fig:intrusive-list} for an illustration how the data structure
  and the code looks like.
  Apart from features seen already above (pointer arithmetic and a need to deal
  with linking fields embedded into larger structures with a need to apply block
  splitting), the code contains also \emph{bit-masking}.
  In particular, one bit of the next pointers is used to mark pointers back to
  the head node, thus effectively marking the ``end'' of the circular list.
  A further intricacy of the code is that the insertion into the list touches
  three nodes that may be different but that may also collapse into a single
  node.
  In the case of the Linux list, we have mentioned a similar situation but with
  two nodes only.
  Having three nodes that possibly collapse is not only beyond the capabilities
  of Infer but also Broom if it is used with its basic abduction strategy.
  This is the reason why some of the contracts produced by Broom are also not
  complete in this case (e.g., effectively allowing insertion into a list with
  more than one node only).
  We will, however, show below that Broom can resolve the problem when using
  more power of the \texttt{alias-weak}/\texttt{alias-strong} rules, though for
  the price of quite increased runtime requirements.
  As for Infer, it is clearly visible that its coverage of the functions is much
  weaker (interestingly, we noticed that it completely ignored the bit-masking
  when deriving some of the contracts).

\begin{figure}[t]
  \centering
    \begin{minipage}{.7\textwidth}
        \hspace{-1em}
        \includegraphics[width=\textwidth]{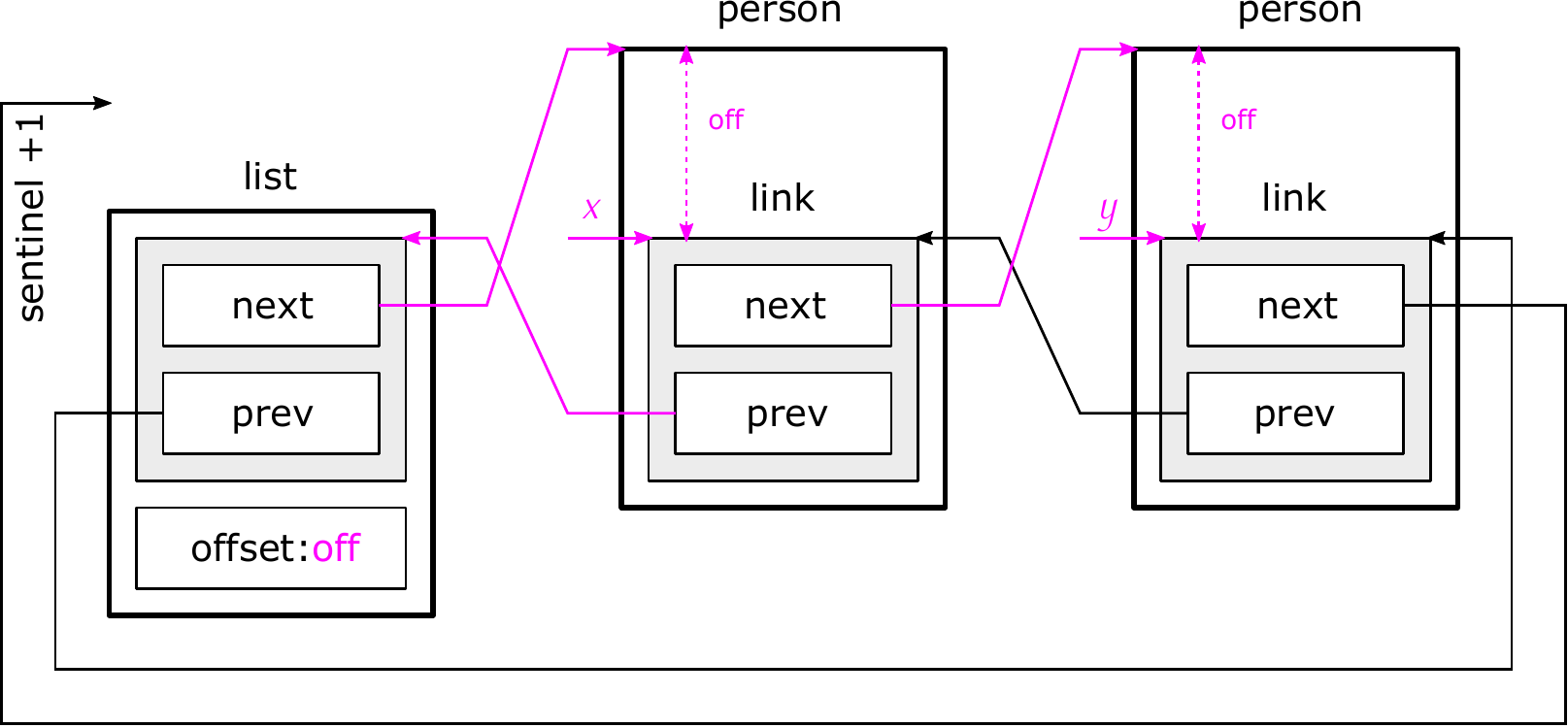}
    \end{minipage}%
    \begin{minipage}{.3\textwidth}
        \hspace{-1em}
        \centering
{\scriptsize
\begin{verbatim}
// offset from a node pointer
// to a link structure
size_t off = (size_t) x -
  ((size_t)x->prev->next & ~1);
// link of the next node
y = (link *)(
  (size_t)lnk->next & ~1 + off
  );
\end{verbatim}}

    \end{minipage}

  \caption{An illustration of the intrusive list data structure. The code fragment shown to the right of the figure gives the code used
  in the function \texttt{link\_get\_next} to obtain the linking structure of
  the next node. Note the use of the pointer arithmetic including bit-masking
  (to clear the bit whose bit-masking on the next pointers is used to denote the
  sentinel node of the list).}

  \label{fig:intrusive-list}

\end{figure}

  \item \texttt{intrusive-list-min} contains a subset of the functions
  considered above (for initializing a list, inserting an element, removing an
  element) together with a top-level function utilising these functions.
  Essentially, the intention here was to create an as small as possible example
  of the given kind already problematic for Infer.
  Again, even Broom cannot handle it fully under its standard abduction
  strategy.

\enlargethispage{4mm}

  \item \texttt{intrusive-list-smoke} contains the entire intrusive list library
  from above together with several top-level functions provided by the author to
  test the library.
  The tests use structures modelling some personal records to be linked into a
  list via the embedded linking structures.
  They create a few such records, link them into a list, traverse them
  (forward/backward), and destroy the list.

\end{itemize}

We now proceed to our experiments with the
\texttt{alias-weak}/\texttt{alias-strong} rules.
As we have said above, these rules involve a lot of guessing.
Hence, if they are used to explore various possible abduction solutions based on
different aliasing scenarios, the running time may grow considerably, but it may
resolve situations that are otherwise not resolved.
To confirm this, we have applied Broom with the strategy of using the
\texttt{alias-weak}/\texttt{alias-strong} rules to explore different possible
abduction solutions with different possible aliasing scenarios on the intrusive
list case study.
The results are shown in Table~\ref{tab:strategy2}.
The first row concerns the experiment \texttt{intrusive-list-min} discussed
already above.
At that time, we noted that Broom could not fully handle some of the intrusive
list functions since they required it to merge three possibly independent nodes
into a single one.
As can be seen in Table~\ref{tab:strategy2}, with the help of the
\texttt{alias-weak}/\texttt{alias-strong} rules, Broom does fully manage even
this problem (though the runtime grew a lot).
The next two rows---\texttt{intrusive-list-min-2} and
\texttt{intrusive-list-min-3}---are variations on the previous case where we
intentionally introduced some bugs, which were correctly discovered.
Finally, the last row shows a significant improvement even for the entire
library of intrusive lists together with its ``smoke'' tests.

\begin{table}
\caption{Experiments with \texttt{alias-weak}/\texttt{alias-strong} in Broom}
\label{tab:strategy2}
\begin{tabular} {l | l | c | c | c | l }
\hline
Name & Expected result & Fncs total & T [m] & Funcs contr & Res \\
\hline
\bm{intrusive-list-min}   & no error    &  9 &  46 &  9 & no error found\\
\bm{intrusive-list-min-2} & memory leak &  9 &  47 &  9 & memory leak\\
\bm{intrusive-list-min-3} & double free &  9 &  49 &  9 & double free\\
\bm{intrusive-list-smoke} & no error    & 20 & 505 & 16 & no error found\\
\hline
\end{tabular}
\end{table}

To sum up, we believe that, despite the highly prototypical nature of Broom, the
presented experiments show that the proposed approach is indeed capable of
handling code that is beyond the capabilities of other currently existing
approaches.

\section{Conclusion and Future Works} \label{sec:conclusion}

We have presented a new SL-based bi-abduction analysis capable of analysing
fragments of code that manipulates with various forms of dynamic linked lists
implemented using advanced low-level pointer operations.
This includes operations such as pointer arithmetic, bit-masking on pointers,
block operations, dealing with blocks of in-advance-unknown size, splitting them
into fields of not-fixed size, which can then be merged again, etc.
Although our approach builds on a body of previous research, especially,
\cite{SLNestedLists07, BiAbd09, BiAbd11, PredatorSAS13}, it extends it
significantly to handle the mentioned features.
In particular, to be able to handle the considered kind of code, we build on a
flavor of SL that uses a per-field separating conjunction instead of a
per-object separating conjunction, and we also introduce a number of new
abduction rules that allow us to deal with pointer arithmetic, block splitting
and merging, and so on.
We have implemented the proposed approach in a prototype tool Broom.
Despite Broom is a very early prototype, our experiments with it allowed us to
handle code fragments that are---to the best of our knowledge---out of the
capabilities of currently existing analysers.

We believe that there is a lot of space for further improvements of our results
in the future.
First, we would like to significantly optimize Broom to make it applicable to
larger code bases.
Here, we are thinking of applying many of the low-level optimisations applied in
other tools of a similar kind (replacing as many as possible of SMT queries by
answering them using simple static rules, using incremental SMT solving, caching
as much information as possible, etc.).
Next, we would like to explore possibilities how to reduce the amount of
non-determinism present in the abduction when the
\texttt{alias-weak}/\texttt{strong} rules are applied.
The goal is to preserve as much as possible of the power of these rules but
reduce the cost of applying them.
Perhaps, we could rely here partially on some pre-defined heuristics and
partially even on some techniques from machine learning, which are now being
applied even in SMT solvers and elsewhere.
Next, we would like to significantly improve our implementation of list
abstractions (inspired, e.g., by \cite{PredatorSAS13}) as well as numerical
abstractions.
Last but not least, we would also like to think of adding support for other
classes of dynamic data structures than lists.


\bibliography{main}

\vfill\eject

\appendix

\section{An Illustration of the Analysis of Conditions and Loops}
\label{app:illustr-cond-loops}

When the code of a function contains conditional branching, one splits the
computation for each particular branch, leading to a set of
$\{(\pre_i,\post_i)\}$ pairs at the end of the function.
When all the $\pre_i$ are mutually disjoint (i.e. the branching is determined by
the precondition), one can use directly all the thus obtained contracts.
This is illustrated by the analysis of the function $a$ in
Fig.~\ref{fig:ex-determ-cond}.
Note that the analysis leads to two contracts for the function $a$ determined by
the value of the parameter $x$.

\begin{figure}[h]
\begin{minipage}{\textwidth-1em}
\small
{\fontfamily{lmr}\selectfont
struct sll \{ struct sll *next;\}; }

\begin{cfunction}{struct sll *\textbf{a}(struct sll *x)}{}
    \state{(\pre\equiv x=X, \quad\post\equiv x=X)}

  \LINE{if (x == NULL) \{}
    \state{\ \ \ \ (\pre\equiv X=\NULL * x=X, \quad\post\equiv X=\NULL * x=X)}

  \LINE{\ \ \ \ struct sll *y = malloc(sizeof(struct sll));}
    \state{\ \ \ \ (\pre\equiv X=\NULL * x=X, \quad
	  \post\equiv \exists Y. \ \ptstobyte{Y}{\top}{8} * \bBlock(Y)=Y * X=\NULL * x=X * y=Y)}
  \LINE{\ \ \ \ y-->next = NULL;}
    \state{\ \ \ \ (\pre\equiv X=\NULL * x=X, \quad
	  \post\equiv \exists Y. \ \ptstobyte{Y}{0}{8} * \bBlock(Y)=Y * X=\NULL * x=X * y=Y)}

  \LINE{\ \ \ \ x = y;}
    \state{\ \ \ \ (\pre\equiv X=\NULL * x=X, \quad
	  \post\equiv \exists Y. \ \ptstobyte{Y}{0}{8} * \bBlock(Y)=Y * X=\NULL * x=y=Y)}

  \LINE{\}}
    \state{(\pre_1\equiv X\neq \NULL * x=X, \quad
	  \post_1\equiv X \neq \NULL * x=X)}
    \state{(\pre_2\equiv X=\NULL * x=X, \quad
	  \post_2\equiv \exists Y. \ \ptstobyte{Y}{0}{8} * \bBlock(Y)=Y * X=\NULL * x=y=Y)}

  \LINE{return x;}
\end{cfunction}
\finalstate{
     (\pre_1\equiv X\neq \NULL * x=X, \quad
      \post_1\equiv X \neq \NULL * ret=x=X) \\
     (\pre_2\equiv X=\NULL * x=X, \quad
      \post_2\equiv \exists Y. \ \ptstobyte{Y}{0}{8} * \bBlock(Y)=Y * X=\NULL * ret=x=Y)}

\end{minipage}
  \caption{An illustration of the analysis of a deterministic condition}
  \label{fig:ex-determ-cond}
\end{figure}

For an illustration of how a loop is analysed, see Fig.~\ref{fig:ex-loop}.
When analysing a loop, branching is encountered whenever the computation reaches
the loop header.
For instance, in the example, starting from $(\pre_{11},\post_{11})$, one
continues inside the loop with $(\pre_{21},\post_{21})$ and out of the loop with
$(\pre_{41},\post_{41})$.
At the end of the loop, the analysis returns to the location corresponding to
the loop header, which is a part of the set of the control-flow cut-points.
%
%
Therefore one checks whether the current pre-/post-pair is covered by some
so-far computed one related to the loop header.
If no, the abstraction is run, and the new pair is added to those so-far
computed for the loop header.
In particular, the pair $(\pre_{32},\post_{32})$ is not covered, and so we add
$(\pre_{32}^{\alpha},\post_{32}^{\alpha})$ to the line corresponding to the loop
header.
When the pair is covered by another one, the analysis discards the pair.
For example, $(\pre_{33},\post_{33})$ is covered by $(\pre_{13},\post_{13})$
because $(\exists L2,L3. \pre_{33})\models (\exists L_2. \pre_{13})$ and $(\exists
L2,L3. \post_{33})\models (\exists L_2. \post_{13})$.
For all the \emph{contract candidates} at the end of the loop, one has to use
the second analysis phase to prove that the abstraction applied on the
precondition did not introduce any bad configuration into the precondition.

\begin{figure}[t]
\begin{minipage}{\textwidth-1em}
\small
{\fontfamily{lmr}\selectfont
struct sll \{ struct sll *next;\}; }

\begin{cfunction}{struct sll *\textbf{loop}(struct sll *x)}{}
    \state{(\pre_{11}\equiv x=X, \quad\post_{11}\equiv x=X)}
    \state{(\pre_{12}\equiv \ptsto{X}{l_1} * x=X, \quad
	  \post_{12}\equiv \ptsto{X}{L_1} * x=L_1)}
    \state{(\pre_{13}\equiv \sll{}(X,L_2) * x=X, \quad
	\post_{13}\equiv \sll{}(X,L_2) * x=L_2 )}

  \LINE{while (x != NULL) \{}
    \state{\ \ \ \ (\pre_{21}\equiv X\neq \NULL * x=X, \quad
	  \post_{21}: X\neq \NULL * x=X)}
    \state{\ \ \ \ (\pre_{22}\equiv \ptsto{X}{L_1} * L_1\neq \NULL, \quad
	  \post_{22}\equiv \ptsto{X}{L_1}  * L_1\neq \NULL * x=L_1)}
    \state{\ \ \ \ (\pre_{23}\equiv \sll{}(X,L_2) * L_2\neq \NULL * x=X, \quad
	\post_{23}\equiv \sll{}(X,L_2) * L_2 \neq \NULL * x=L_2)}

  \LINE{\ \ \ \ x = x-->next;}
    \state{\ \ \ \ (\pre_{31}\equiv \ptsto{X}{L_1} * x=X,\quad
       \post_{31}\equiv \ptsto{X}{L_1} * x=L_1)}
    \state{\ \ \ \ (\pre_{32}\equiv X\rightarrow L_1 *  \ptsto{L_1}{L_2} * x=X,\quad
	   \post_{32}\equiv \ptsto{X}{L_1} * \ptsto{L_1}{L_2} * x=L_2 )}
    \state{\ \ \ \ (\pre_{32}^{\alpha}\equiv \sll{}(X,l_2) * x=X,\quad
	  \post_{32}^{\alpha}\equiv \sll{}(X,L_2) * x=L_2 )}
    \state{\ \ \ \ (\pre_{33}\equiv \sll{}(X,L_2) * \ptsto{L_2}{L_3} * x=X,\quad
	  \post_{33}\equiv \sll{}(X,L_2) * \ptsto{L_2}{L_3} * x=L_3 )}
    \discardstate{(\pre_{33},\post_{33})$---entailment $(\exists
	L_3. \pre_{33})\models (\exists L_2. \pre_{13})$ and
	$(\exists L_3. \post_{33})\models (\exists L_2. \post_{13})}

  \LINE{\}}
    \state{(\pre_{41}\equiv X=\NULL * x=X, \quad\post_{41}\equiv X=\NULL * x=X)}
    \state{(\pre_{42}\equiv \ptsto{X}{L_1} * L_1=\NULL * x=X, \quad
	  \post_{42}\equiv \ptsto{X}{L_1} * x=L_1=\NULL)}
    \state{(\pre_{43}\equiv \sll{}(X,L_2) * L_2=\NULL * x=X,\quad
	  \post_{43}\equiv \sll{}(X,L_2) * x=L_2=\NULL)}
  \LINE{\dots}
\end{cfunction}

\end{minipage}
  \caption{An illustration of the analysis of a loop}
  \label{fig:ex-loop}
\end{figure}

Finally, we provide in Fig.~\ref{fig:ex-nondet} an illustration of how the
analysis deals with non-deterministic branching---plus we also illustrate the use of the two analysis
phases.
The function $f$ from Fig.~\ref{fig:ex-nondet} has an if-statement with a non-deterministic condition, which is translated to two outgoing control-flow edges that are not labelled by any condition in our modelling.
The analysis of the \emph{then} branch leads to the candidate precondition
$\pre_1\equiv\ptsto{X}{L} * x=X$.
The analysis of the \emph{else} branch leads to the candidate precondition
$\pre_2\equiv  x=X$.
%
%
Then, we run the second phase which can be done in a quite similar way as the first phase.
However, whenever we need to strengthen the precondition, we immediately stop
with a failure.
The computation first starts with the pair $(\pre_1,\pre_1)$, and we compute two
postconditions $\post_{11}$ and $\post_{12}$.
Therefore, the final contract is $(\pre_1,\post_{11}\vee\post_{12})$.
Then we start the second phase from $(\pre_2,\pre_2)$, but we realize that the
precondition is not safe enough to execute all possible branches of the
function---we realize that we need to infer a stronger one by using the
\texttt{learn-finish} abduction rule.
Therefore we drop $\pre_2$.

\begin{figure}[hp]
\begin{minipage}{\textwidth-1em}
\small
{\fontfamily{lmr}\selectfont
struct sll \{ struct sll *next;\}; }

~

{\fontfamily{lmr}\selectfont\color{gray}// The first analysis phase}

\begin{cfunction}{struct item *\textbf{f}(struct item *x)}{}
    \state{(\pre\equiv x=X, \quad\post\equiv x=X)}
  \LINE{if (random())}
    \state{\ \ \ \ (\pre\equiv x=X, \quad\post\equiv x=X)}
  \LINE{\ \ \ \ return x-->next;}
    \state{\ \ \ \ (\pre\equiv \ptsto{X}{L} * x=X, \quad
	  \post\equiv \ptsto{X}{L} *ret=L * x=X)}
  \LINE{else}
    \state{\ \ \ \ (\pre\equiv x=X, \quad\post\equiv  x=X)}
  \LINE{\ \ \ \ return x;}
    \state{\ \ \ \ (\pre\equiv x=X, \quad\post\equiv  ret=x=X)}
\end{cfunction}
\candidate{\pre_1\equiv \ptsto{X}{L} * x=X \\ \pre_2\equiv x=X}

~

{\fontfamily{lmr}\selectfont\color{gray}// The second analysis phase}

\begin{cfunction}{struct item *\textbf{f}(struct item *x)}{}
    \state{(\pre_1\equiv \ptsto{X}{L} * x=X, \quad
	  \post_1\equiv \ptsto{X}{L} * x=X)}
    \state{(\pre_2\equiv x=X, \quad\post_2\equiv x=X)}
  \LINE{if (random())}
    \state{\ \ \ \ (\pre_1\equiv \ptsto{X}{L} * x=X, \quad
	  \post_1\equiv \ptsto{X}{L}  * x=X)}
    \state{\ \ \ \ (\pre_2\equiv x=X, \quad\post_2\equiv x=X)}
  \LINE{\ \ \ \ return x-->next;}
    \state{\ \ \ \ (\pre_1\equiv \ptsto{X}{L} * x=X, \quad
	  \post_1\equiv \ptsto{X}{L} * ret=L * x=X)}
    \state{\ \ \ \ (\pre_2\equiv {\color{red}\ptsto{X}{L}} * x=X, \quad
	  \post_2\equiv \ptsto{X}{L} *  ret=L * x=X)}
    \discardstate{(\pre_2,\post_2)$ because we should learn a missing part of the
      precondition $\ptsto{X}{L}}
  \LINE{else}
    \state{\ \ \ \ (\pre_1\equiv \ptsto{X}{L} *  x=X, \quad
	  \post_1\equiv \ptsto{X}{L} * x=X)}
  \LINE{\ \ \ \ return x;}
    \state{\ \ \ \ (\pre_1\equiv \ptsto{X}{L} *  x=X, \quad
	  \post_1\equiv \ptsto{X}{L} * ret=x=X)}
\end{cfunction}
\finalstate{
    (\pre_{11}\equiv \ptsto{X}{L} * x=X, \quad
	  \post_{11}\equiv \ptsto{X}{L} * ret=L * x=X) \\
    (\pre_{12}\equiv \ptsto{X}{L} * x=X, \quad
	  \post_{12}\equiv \ptsto{X}{L} * ret=x=X)}
\end{minipage}
  \caption{An illustration of the analysis of a non-deterministic condition}
  \label{fig:ex-nondet}
\end{figure}

\newpage

\section{An Illustration of Bit-masking and Pointer Arithmentics}
\label{app:illustr-bitmasking}

The Figure \ref{fig:example5} containes an additional example taken from our intrusive-list
experiment (see Section \ref{sec:Experiments}). This example uses a (i) bit-masking to mark an first node in
the circular doubly-linked list and (ii) pointer arithmetic.

\begin{figure}[h]
\begin{minipage}{\textwidth-1em}
\small
{\fontfamily{lmr}\selectfont
struct link \{ struct link *prev; void *next; \}; }

\begin{cfunction}{struct link *\textbf{link\_get\_next}(struct link *l)}{}
    \state{\pre\equiv l=L, \quad\post\equiv l=L}
  \LINE{struct link *a = l-->prev;}
    \state{\pre\equiv \ptsto{L}{A} * l=L, \\
       \post\equiv \ptsto{L}{A} * l=L * a=A}

  \LINE{void *b = a-->next;}

    \state{\pre\equiv \ptsto{L}{A} * \ptsto{A+8}{B} * l=L, \\
       \post\equiv \ptsto{L}{A} * \ptsto{A+8}{B} * l=L * a=A * b=B}

  \LINE{size\_t c = (size\_t) b \& $\sim$1;}

    \state{\pre\equiv \ptsto{L}{A} * \ptsto{A+8}{B} * l=L, \\
       \post\equiv \exists C.\ \ptsto{L}{A} * \ptsto{A+8}{B} * C=B~\&~$-2$ ~*~l=L * a=A * b=B * c=C}

  \LINE{size\_t o = (size\_t) l - c; {\color{gray}// offset from a node pointer to a link structure}}
    \state{\pre\equiv \ptsto{L}{A} * \ptsto{A+8}{B} * l=L, \\
       \post\equiv \exists C,O.\ \ptsto{L}{A} * \ptsto{A+8}{B} * C=B~\&~$-2$ ~*~O=L-C * l=L * a=A * b=B * c=C * o=O}

  \LINE{void *d = (size\_t) l-->next;}
    \state{\pre\equiv \ptsto{L}{A} * \ptsto{L+8}{D} * \ptsto{A+8}{B} * l=L, \\
       \post\equiv \exists C,O.\ \ptsto{L}{A} * \ptsto{L+8}{D} * \ptsto{A+8}{B} * C=B~\&~$-2$
       ~*~O=L-C * l=L * a=A * \\ b=B * c=C * o=O * d=D}

  \LINE{size\_t f = (size\_t) d \& $\sim$1; {\color{gray}// link field for the next node}}
    \state{\pre\equiv \ptsto{L}{A} * \ptsto{L+8}{D} * \ptsto{A+8}{B} * l=L, \\
       \post\equiv \exists C,F,O.\ \ptsto{L}{A} * \ptsto{L+8}{D} * \ptsto{A+8}{B} * C=B~\&~$-2$ ~*~O=L-C *
       \\ F=D~\&~$-2$ ~*~ l=L * a=A * b=B * c=C * o=O * d=D * f=F}

  \LINE{size\_t e = f + o;}
    \state{\pre\equiv \ptsto{L}{A} * \ptsto{L+8}{D} * \ptsto{A+8}{B} * l=L, \\
       \post\equiv \exists C,E,F,O.\ \ptsto{L}{A} * \ptsto{L+8}{D} * \ptsto{A+8}{B} * C=B~\&~$-2$ ~*~O=L-C *
       \\ F=D~\&~$-2$ ~*~ E=F+O ~*~ l=L * a=A * b=B * c=C * o=O * d=D * f=F * e=E}

  \LINE{return (struct link *) e;}
    \state{\pre\equiv \ptsto{L}{A} * \ptsto{L+8}{D} * \ptsto{A+8}{B} * l=L, \\
       \post\equiv \exists C,E,F,O.\ \ptsto{L}{A} * \ptsto{L+8}{D} * \ptsto{A+8}{B} * C=B~\&~$-2$ ~*~O=L-C *
       \\ F=D~\&~$-2$ ~*~ E=F+O ~*~ l=L * a=A * b=B * c=C * o=O * d=D * f=F * ret=e=E}

\end{cfunction}
\finalstate{\pre\equiv \ptsto{L}{A} * \ptsto{L+8}{D} * \ptsto{A+8}{B} * l=L, \\
       \post\equiv \exists C,E,F,O.\ \ptsto{L}{A} * \ptsto{L+8}{D} * \ptsto{A+8}{B} * C=B~\&~$-2$ ~*~O=L-C *
       \\ F=D~\&~$-2$ ~*~ E=F+O ~*~ l=L * ret=E}

\end{minipage}
  \caption{An illustrative example of a code from Fig.~\ref{fig:intrusive-list} using pointer arithmetic and bit-masking}
  \label{fig:example5}
\end{figure}

\end{document}